\documentclass[11pt]{article}

\usepackage{style}
\usepackage{notation}

\title{Lower Bound Techniques in the Comparison-Query Model and Inversion Minimization on Trees}
\author{Ivan Hu \and Dieter van Melkebeek \and Andrew Morgan}
\date{\today}

\begin{document}

\maketitle

\begin{abstract}

Given a rooted tree and a ranking of its leaves, what is the minimum number of inversions of the leaves that can be attained by ordering the tree? This variation of the well-known problem of counting inversions in arrays originated in mathematical psychology. It has the evaluation of the Mann--Whitney statistic for detecting differences between distributions as a special case. 

We study the complexity of the problem in the comparison-query model, the standard model for problems like sorting, selection, and heap construction. The complexity depends heavily on the shape of the tree: for trees of unit depth, the problem is trivial; for many other shapes, we establish lower bounds close to the strongest known in the model, namely the lower bound of $\log_2(n!)$ for sorting $n$ items. For trees with $n$ leaves we show, in increasing order of closeness to the sorting lower bound:
\begin{itemize}
\item[(a)] $\log_2((\alpha(1-\alpha)n)!) - O(\log n)$ queries are needed whenever the tree has a subtree that contains a fraction $\alpha$ of the leaves. This implies a lower bound of $\log_2((\frac{k}{(k+1)^2}n)!) - O(\log n)$ for trees of degree $k$.
\item[(b)] $\log_2(n!) - O(\log n)$ queries are needed in case the tree is binary. 
\item[(c)] $\log_2(n!) - O(k \log k)$ queries are needed for certain classes of trees of degree $k$, including perfect trees with even $k$.
\end{itemize}

The lower bounds are obtained by developing two novel techniques for a generic problem $\Pi$ in the comparison-query model and applying them to inversion minimization on trees. Both techniques can be described in terms of the Cayley graph of the symmetric group with adjacent-rank transpositions as the generating set, or equivalently, in terms of the edge graph of the permutahedron, the polytope spanned by all permutations of the vector $(1,2,\dots,n)$. Consider the subgraph consisting of the edges between vertices with the same value under $\Pi$. We show that the size of any decision tree for $\Pi$ must be at least:
\begin{itemize}
\item[(i)] the number of connected components of this subgraph, and
\item[(ii)] the factorial of the average degree of the complementary subgraph, divided by $n$.
\end{itemize}
Lower bounds on query complexity then follow by taking the base-2 logarithm. Technique (i) represents a discrete analog of a classical technique in algebraic complexity and allows us to establish (c) and a tight lower bound for counting cross inversions, as well as unify several of the known lower bounds in the comparison-query model. Technique (ii) represents an analog of sensitivity arguments in Boolean complexity and allows us to establish (a) and (b). 

Along the way to proving (b), we derive a tight upper bound on the maximum probability of the distribution of cross inversions, which is the distribution of the Mann--Whitney statistic in the case of the null hypothesis. Up to normalization, the probabilities alternately appear in the literature as the coefficients of polynomials formed by the Gaussian binomial coefficients, also known as Gaussian polynomials.

\end{abstract}


\section{Overview}
\label{sec:overview}

The result of a hierarchical cluster analysis on a set $X$ of items can be thought of as an unordered rooted tree $T$ with leaf set $X$. To visualize the tree, or to spell out the classification in text, one needs to decide for every internal node of $T$ in which order to visit its children. Figure~\ref{fig:example:initial} represents an example of a classification of eight body parts from the psychology literature \cite{Degerman1982}. It is obtained by repeatedly clustering nearest neighbors where the distance between two items is given by the number of people in a survey who put the items into different classes \cite{Miller1969}. The ordering of the resulting binary tree in Figure~\ref{fig:example:initial} is the output produced by a particular implementation of the clustering algorithm. 

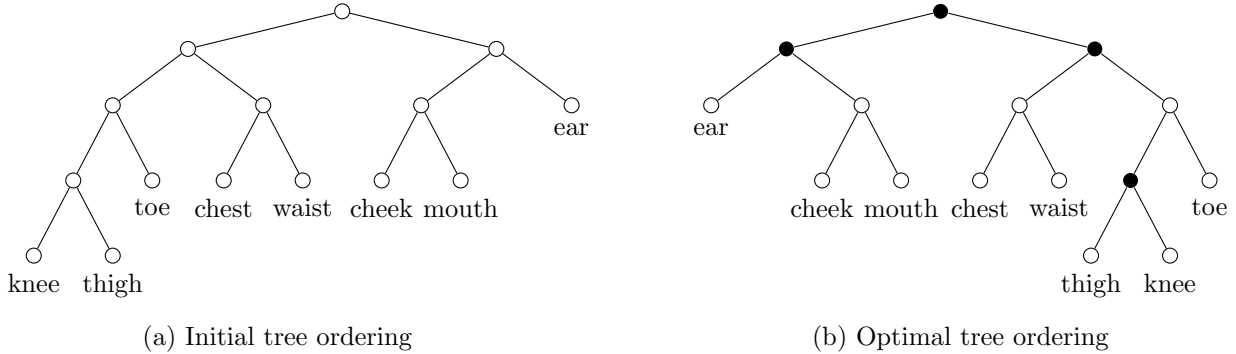
\begin{figure}[h]
\centering
\begin{subfigure}[b]{0.45\textwidth}
\centering
\begin{tikzpicture}
[
level 1/.style = {sibling distance = 4.1cm, level distance = 0.5cm},
level 2/.style = {sibling distance = 2cm, level distance = 0.75cm},
level 3/.style = {sibling distance = 1.05cm, level distance = 1cm}
]
\node [circle,draw,inner sep=2pt]{}
    child 
    {node [circle,draw,inner sep=2pt]{}
        child
        {node [circle,draw,inner sep=2pt]{}
            child
            {node [circle,draw,inner sep=2pt]{}
                child
                {node [circle,draw,inner sep=2pt, label={below:\small knee}]{}}
                child
                {node [circle,draw,inner sep=2pt, label={below:\small thigh}]{}}}
            child
            {node [circle,draw,inner sep=2pt, label={below:\small toe}]{}}}
        child
        {node [circle,draw,inner sep=2pt]{}
            child
            {node [circle,draw,inner sep=2pt, label={below:\small chest}]{}}
            child
            {node [circle,draw,inner sep=2pt, label={below:\small waist}]{}}}}
    child 
    {node [circle,draw,inner sep=2pt]{}
        child
        {node [circle,draw,inner sep=2pt]{}
            child
            {node [circle,draw,inner sep=2pt, label={below:\small cheek}]{}}
            child
            {node [circle,draw,inner sep=2pt, label={below:\small mouth}]{}}}
        child
        {node [circle,draw,inner sep=2pt, label={below:\small ear}]{}}};
\end{tikzpicture}
\caption{Initial tree ordering}\label{fig:example:initial}
\end{subfigure}
\hfill
\begin{subfigure}[b]{0.45\textwidth}
\centering
\begin{tikzpicture}
[
level 1/.style = {sibling distance = 4.1cm, level distance = 0.5cm},
level 2/.style = {sibling distance = 2cm, level distance = 0.75cm},
level 3/.style = {sibling distance = 1.05cm, level distance = 1cm}
]
\node [circle,fill,inner sep=2pt]{}
    child
    {node [circle,fill,inner sep=2pt]{}
        child
        {node [circle,draw,inner sep=2pt, label={below:\small ear}]{}}
        child
        {node [circle,draw,inner sep=2pt]{}
            child
            {node [circle,draw,inner sep=2pt, label={below:\small cheek}]{}}
            child
            {node [circle,draw,inner sep=2pt, label={below:\small mouth}]{}}}}
    child
    {node [circle,fill,inner sep=2pt]{}
        child
        {node [circle,draw,inner sep=2pt]{}
            child
            {node [circle,draw,inner sep=2pt, label={below:\small chest}]{}}
            child
            {node [circle,draw,inner sep=2pt, label={below:\small waist}]{}}}
        child
        {node [circle,draw,inner sep=2pt]{}
            child
            {node [circle,fill,inner sep=2pt]{}
                child
                {node [circle,draw,inner sep=2pt, label={below:\small thigh}]{}}
                child
                {node [circle,draw,inner sep=2pt, label={below:\small knee}]{}}}
            child
            {node [circle,draw,inner sep=2pt, label={below:\small toe}]{}}}};
\end{tikzpicture}
\caption{Optimal tree ordering}\label{fig:example:optimal}
\end{subfigure}
\caption{Classification of body parts}\label{fig:example}
\end{figure}

Another ordering is given in Figure~\ref{fig:example:optimal}; black marks the nodes whose children have been swapped from the ordering in Figure~\ref{fig:example:initial}. Figure~\ref{fig:example:optimal} has the advantage over Figure~\ref{fig:example:initial} that the leaves now appear in an interesting global order, namely head-to-toe: ear, cheek, mouth, chest, waist, thigh, knee, toe. Indeed, Figure~\ref{fig:example:optimal} makes apparent that the anatomical order correlates perfectly with the clustering. In general, given a tree $T$ and a ranking $\ranking$ of its leaves, one might ask ``how correlated'' is $T$ with $\ranking$? Degerman~\cite{Degerman1982} suggests evaluating the orderings of $T$ in terms of the number of inversions of the left-to-right ranking $\ordering$ of the leaves with respect to the given ranking $\ranking$, and use the minimum number over all orderings as a measure of (non)correlation.
\begin{definition}[ranking, inversion, $\Inv_{\cdot}(\cdot)$]
A \emph{ranking} $\ranking$ of a set $X$ of $n$ items is a bijection from $X$ to $[n]$. 
Given two rankings $\ordering$ and $\ranking$, an \emph{inversion} of $\ordering$ with respect to $\ranking$ is a pair of items $x_1,x_2\in X$ such that $\ranking(x_1) < \ranking(x_2)$ but $\ordering(x_1) > \ordering(x_2)$. The number of inversions is denoted by $\Inv_\ranking(\ordering)$.
An inversion in an array $A$ of values is an inversion of $\ordering$ with respect to $\ranking$ where $\ordering$ denotes the ranking by array index and $\ranking$ the ranking by value; in this setting we write $\Inv(A)$ for $\Inv_\ranking(\sigma)$.
\end{definition}

The minimum number of inversions can be used to compare the quality of different trees $T$ for a given ranking $\ranking$, or of different rankings $\ranking$ for a given tree $T$. This mimics the use of the number of inversions in applications like collaborative filtering in recommender systems, rank aggregation for meta searching the web, and Kendall's test for dependencies between two random variables. 
In particular, the Mann--Whitney test for differences between random variables can be viewed as a special case of our optimization problem. The test is widely used because of its nonparametric nature, meaning that no assumptions need to be made about the distribution of the two variables; the distribution of the statistic in the case of the null hypothesis (both variables have the same distribution) is always the same. The test achieves this property by only considering the relative order of the samples. It takes a sequence $A$ of $a$ samples from a random variable $Y$, a sequence $B$ of $b$ samples from another random variable $Z$, and computes the statistic $U \doteq \min(\XInv(A,B),\XInv(B,A))$ that is the minimum of the number $\XInv(A,B)$ of cross inversions from $A$ to $B$, and vice versa. 
\begin{definition}[cross inversions, $\XInv_{\cdot}(\cdot,\cdot)$]
Let $\ranking$ be a ranking of $X$, and $A, B \subseteq X$. A \emph{cross inversion} from $A$ to $B$ with respect to $\ranking$ is a pair $(x_1,x_2) \in A \times B$ that is out of order with respect to $\ranking$, i.e., such that $\ranking(x_1) > \ranking(x_2)$. The number of cross inversions is denoted by $\XInv_\ranking(A,B)$.
For two arrays $A$ and $B$ of values, a cross inversion from $A$ to $B$ is a cross inversion from the set of entries in $A$ to the set of entries in $B$ where $\ranking$ denotes the ranking by value; in this setting we write $\XInv(A,B)$ for $\XInv_\ranking(A,B)$.
\end{definition}
The statistic $U$ coincides with the optimum value of our optimization problem with the tree $T$ in Figure~\ref{fig:Mann-Whitney} as input. The leftmost $a$ leaves correspond to the samples $A$, the rightmost $b$ leaves to the samples $B$, and the ranking $\ranking$ to the value order of the combined $a+b$ samples.
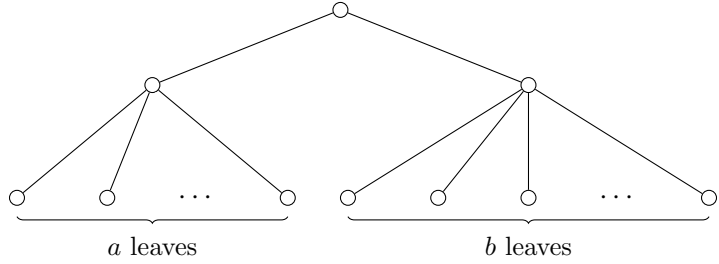
\begin{figure}
\centering
\begin{tikzpicture}
[
level 1/.style = {sibling distance = 5cm, level distance = 1cm},
level 2/.style = {sibling distance = 1.2cm, level distance = 1.5cm}
]
\node [circle,draw,inner sep=2pt]{}
    child
    {node [circle,draw,inner sep=2pt]{}
        child
        {node (a1) [circle,draw,inner sep=2pt]{}}
        child
        {node [circle,draw,inner sep=2pt]{}}
        child
        {node {$\cdots$} edge from parent [color=white]}
        child
        {node (a2) [circle,draw,inner sep=2pt]{}}}
    child
    {node [circle,draw,inner sep=2pt]{}
        child
        {node (b1) [circle,draw,inner sep=2pt]{}}
        child
        {node [circle,draw,inner sep=2pt]{}}
        child
        {node [circle,draw,inner sep=2pt]{}}
        child
        {node {$\cdots$} edge from parent [color=white]}
        child
        {node (b2) [circle,draw,inner sep=2pt]{}}};

\draw[decorate, decoration=brace] ($(a2) + (0,-0.25)$) -- ($(a1) + (0,-0.25)$);
\draw[decorate, decoration=brace] ($(b2) + (0,-0.25)$) -- ($(b1) + (0,-0.25)$);
\node at ($0.5*(a1) + 0.5*(a2) + (0,-0.25)$)[label={below:\small $a$ leaves}]{};
\node at ($0.5*(b1) + 0.5*(b2) + (0,-0.25)$)[label={below:\small $b$ leaves}]{};
\end{tikzpicture}
\caption{Mann--Whitney instance}\label{fig:Mann-Whitney}
\end{figure}

We mainly study the value version of our optimization problem, which we denote by $\MInv$.
\begin{definition}[inversion minimization on trees, $\MInv(\cdot,\cdot)$, $\Pi_{\cdot}$]
Inversion minimization on trees is the computational problem with the following specification:
\begin{description}
\item[Input:] A rooted tree $T$ with leaf set $X$ of size $n$, and a ranking $\ranking$ of $X$.
\item[Output:] $\MInv(T,\ranking)$, the minimum of $\Inv_\ranking(\ordering)$ over all possible orderings of $T$, where $\ordering$ denotes the left-to-right ranking of $X$ induced by the ordering of $T$.
\end{description}
For any fixed tree $T$ with leaf set $X$, we use the short-hand $\Pi_T$ to denote the computational problem that takes as input a ranking $\ranking$ of $X$ and outputs $\MInv(T,\ranking)$. 
\end{definition}

Degerman~\cite{Degerman1982} observes that the ordering at each internal node can be optimized independently in a greedy fashion. In the setting of binary trees, for each node $v$, we can count the cross inversions from the leaves in the left subtree of $v$ to the leaves in the right subtree of $v$. Between the two possible orderings of the children of a node $v$, we choose the one that yields the smaller number of cross inversions. 
Based on his observation, Degerman presents a polynomial-time algorithm for the case of binary trees $T$. A more refined implementation and analysis yields a running time of $O(\avgdepth(T) \cdot n)$, where $\avgdepth(T)$ denotes the average depth of a leaf in $T$. For balanced binary trees the running time becomes $O(n \log n)$. All of this can be viewed as variants of the well-known $O(n \log n)$ divide-and-conquer algorithm for counting inversions in arrays of length $n$.

For trees of degree $\deg(T) > 2$, the local greedy optimization at each internal node becomes more complicated, as there are many ways to order the children of each internal node. Exhaustive search results in a running time of $O((\deg(T)! + \deg(T) \cdot \avgdepth(T)) \cdot n)$, which can be improved to $O((\deg(T)^2 2^{\deg(T)} + \deg(T) \cdot \avgdepth(T)) \cdot n)$ using dynamic programming. The problem is closely related to the classical problem of minimum arc feedback set, and becomes NP-hard without any constraints on the degree. We refer to Section~\ref{sec:turing} for more details.

\paragraph{Query complexity.}
Rather than running time in the Turing machine model, our focus lies on query complexity in the comparison-query model. There we can only access the ranking $\ranking: X \to [n]$ via queries of the form: Is $\ranking(x_1) < \ranking(x_2)$? For any fixed tree $T$, we want to determine the minimum number of queries needed to solve the problem. 

The comparison-query model represents the standard model for analyzing problems like sorting, selection, and heap construction. Sorting represents the hardest problem in the comparison-query model as it is tantamount to knowing the entire ranking $\ranking$. Its query complexity has a well-known information-theoretic lower bound of $\log_2(n!) = n \log_2(n/e) + \frac{1}{2} \log_2(n) + O(1)$. Standard algorithms such as mergesort and heapsort yield an upper bound of $\log_2(n!) + O(n)$, which has been improved to $\log_2(n!) + o(n)$ recently \cite{Sergeev2020}.
We refer to Section~\ref{sec:model} for an overview of results and techniques for lower bounds in the model.

Information theory only yields a very weak lower bound on the query complexity of inversion minimization on trees: $\log_2 \binom{n}{2} = 2 \log_2(n) - O(1)$. The complexity of the problem critically depends on the shape of the tree $T$ and can be significantly lower than the one for sorting. For starters, the problem becomes trivial for trees of depth one as their leaves can be arranged freely in any order. More precisely, the trees $T$ for which the answer is identically zero, irrespective of the ranking $\ranking$, are exactly those such that all root-to-leaf paths have only the root in common. 

Arguably, the simplest nontrivial instances of inversion minimization are for trees $T$ of the Mann--Whitney type in Figure~\ref{fig:Mann-Whitney} with $a=1$ and $b=n-1$. Depending on the rank $r$ of the isolated leaf, an optimal ordering of $T$ is either the left or the right part in Figure~\ref{fig:simplest-non-trivial}, where the label of each leaf is its rank under $\ranking$.

\begin{figure}[h]
\centering
    \begin{tikzpicture}
      \node (11) at (1,0) [circle, draw, inner sep=1.5pt, label={below:\small \(r\)}]{};
      \node (21) at (2,0) [circle, draw, inner sep=1.5pt, label={below:\small \(1\)}]{};
      \node (22) at (3,0) [circle, draw, inner sep=1.5pt, label={below:\small \(2\)}]{};
      \node (2d) at (4,0) {\footnotesize \(\dots\)};
      \node (2n) at (5,0) [circle, draw, inner sep=1.5pt, label={below:\small \(n\)}]{};
      \node (mid) at (3.5,1)[circle, draw, inner sep=1.5pt]{};
      \node (top) at (2.5,2)[circle, draw, inner sep=1.5pt]{};
      \draw (21) -- (mid);
      \draw (22) -- (mid);
      \draw (2n) -- (mid);
      \draw (11) -- (top);
      \draw (mid) -- (top);
    \end{tikzpicture}
    \hspace{0.75in}
    \begin{tikzpicture}
      \node (11) at (5,0) [circle, draw, inner sep=1.5pt, label={below:\small \(r\)}]{};
      \node (21) at (1,0) [circle, draw, inner sep=1.5pt, label={below:\small \(1\)}]{};
      \node (22) at (2,0) [circle, draw, inner sep=1.5pt, label={below:\small \(2\)}]{};
      \node (2d) at (3,0) {\footnotesize \(\dots\)};
      \node (2n) at (4,0) [circle, draw, inner sep=1.5pt, label={below:\small \(n\)}]{};
      \node (mid) at (2.5,1)[circle, draw, inner sep=1.5pt]{};
      \node (top) at (3.5,2)[circle, draw, inner sep=1.5pt]{};
      \draw (21) -- (mid);
      \draw (22) -- (mid);
      \draw (2n) -- (mid);
      \draw (11) -- (top);
      \draw (mid) -- (top);
    \end{tikzpicture}
\caption{Rank instance}\label{fig:simplest-non-trivial}
\end{figure}
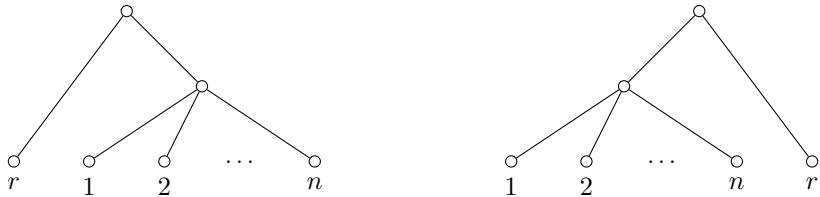

As the ordering on the left has $r-1$ inversions and the one on the right $n-r$, the answer is $\min(r-1,n-r)$. Thus, this instance of inversion minimization on trees is essentially equivalent to rank finding, which has query complexity exactly $n-1$.

\paragraph{Results.} 
We prove that for many trees $T$, inversion minimization on $T$ is nearly as hard as sorting. First, we exhibit a common structure that guarantees high complexity, namely a subtree that contains a fairly balanced fraction of the leaves. We make use of the following notation.
\begin{definition}[leaf set, $\leafset(\cdot)$, and subtree]
For a tree $T$, the leaf set of $T$, denoted $\leafset(T)$, is the set of leaves of $T$. For a node $v$ to $T$, $T_v$ denotes the subtree of $T$ rooted at $v$.
\end{definition}
The quantitative statement references the gamma function $\Gamma$, which is a proxy for any convex real function that interpolates the factorial function on the positive integers. More precisely, we have that $\Gamma(n+1) = n!$ for every integer $n \ge 1$. 
\begin{theorem}[lower bound for general trees]\label{thm:main:general}
Let $T$ be a tree with $n$ leaves, and $v$ a node with $\lvert\leafset(T_v)\rvert = \ell$. The query complexity of inversion minimization on $T$ is at least $\log_2(\Gamma( \frac{\ell(n-\ell)}{n}+1 ))$. In particular, the complexity is at least $\log_2(\Gamma(\frac{k}{(k+1)^2} \cdot n+1))$ where $k$ denotes the degree of $T$.
\end{theorem}
For trees of constant degree, Theorem~\ref{thm:main:general} yields a lower bound that is as strong as the one for sorting up to a constant multiplicative factor. For the important case of binary trees (like the classification trees from the motivating example), we obtain a lower bound that is only a logarithmic additive term shy of the lower bound for sorting.
\begin{theorem}[lower bound for binary trees]\label{thm:main:binary}
For binary trees $T$ with $n$ leaves, the query complexity of inversion minimization on $T$ is at least $\log_2(n!) - O(\log n)$.
\end{theorem}
The logarithmic loss can be reduced to a constant for certain restricted classes of trees. The full statement is somewhat technical. First, it assumes that the tree has no nodes of degree 1. This is without loss of generality, as we can short-cut all degree-1 nodes in the tree without affecting the minimum number of inversions. For example, trivial trees for inversion minimization have depth 1 without loss of generality. Second, the strength of the lower bound depends on the maximum size of a leaf child set, defined as follows.
\begin{definition}[leaf child set, $\leafchildset(\cdot)$]\label{def:leaf-child-set} The \emph{leaf child set} $\leafchildset(v)$ of a vertex $v$ in a tree $T$ is the set $\leafchildset(v)$ of all the children of $v$ that are leaves in $T$.
\end{definition}
Most importantly, the result requires certain fragile parity conditions to hold. That said, there are interesting classes satisfying all requirements, and the bounds are very tight. 
\begin{theorem}[lower bound for restricted classes]\label{thm:main:special}
Let $T$ be a tree without nodes of degree 1 such that the leaf child sets have size at most $k$, at most one of them is odd, and if there exists an odd one, say $\leafchildset(v^*)$, then all ancestors of $v^*$ have empty leaf child sets. The query complexity of inversion minimization on $T$ is at least $\log_2(n!) - O(k \log k)$. In particular, the lower bound applies to:
\begin{itemize}
\item perfect trees of even degree $k$, and
\item full binary ($k=2$) trees with at most one leaf without a sibling leaf.
\end{itemize}
\end{theorem}
Recall that a tree of degree $k$ is full if every node has degree 0 or $k$. It is perfect if it is full and all leaves have the same depth.

For the Mann--Whitney statistic, Theorem~\ref{thm:main:general} provides an $\Omega(n \log n)$ lower bound for balanced instances, i.e., when $a$ and $b$ are $\Theta(n)$. For unbalanced instances there is a more efficient way to count cross inversions and thus evaluate the statistic: Sort the smaller of the two sides, and then do a binary search for each item of the larger side to find its position within the sorted smaller side so as to determine the number of cross inversions that it contributes. For $a \le b$ the approach makes $b \log_2(a) + O(a \log a)$ comparisons. We establish a lower bound that shows the approach is optimal up to a constant multiplicative factor.
\begin{theorem}[lower bound for counting cross inversions]\label{thm:Mann-Whitney}
Counting cross inversions from a set $A$ of size $a$ to a set $B$ of size $b \ge a$ with respect to a ranking $\ranking$ of $X \doteq A \sqcup B$ requires $\Omega((a+b)\log(a))$ queries in the comparison-query model, as does inversion minimization on the tree of Figure~\ref{fig:Mann-Whitney}.
\end{theorem}

\paragraph{Techniques.} We obtain our results by developing two new query lower bound techniques for generic problems $\Pi$ in the comparison-query model, and then instantiating them to the problem $\Pi_T$ of inversion minimization on a fixed tree $T$. Although some of our techniques extend to relations, we restrict attention to computational problems $\Pi$ that are functions, just like the problem $\Pi_T$ that we focus on.
\begin{definition}[computational problem and algorithm in the comparison-query model]
\label{def:computational-problem}
A computational problem in the comparison-query model is a total function on the rankings of a set $X$. An algorithm in the comparison-query model can access an input ranking $\ranking: X \to [n]$ using comparison queries: For given $x_1,x_2 \in X$, test if $\ranking(x_1)< \ranking(x_2)$.
\end{definition}

Both of our techniques follow the common pattern of lower bounding the number of distinct execution traces that any algorithm for $\Pi$ needs to have.
\begin{definition}[execution trace, complexity measures $\traces(\cdot)$ and $\DQ(\cdot)$]
\label{def:execution-and-deterministic}
Consider an algorithm $A$ for a problem $\Pi$ in the comparison-query model. An execution trace of $A$ is the sequence of comparisons that $A$ makes on some input $\ranking$, as well as the outcomes of the comparisons. The complexity $\traces(\Pi)$ is the minimum over all possible algorithms for $\Pi$ of the number of distinct traces the algorithm has over the set of all inputs $\ranking$. The complexity $\DQ(\Pi)$ is the minimum, over all possible algorithms for $\Pi$ of the maximum number of comparisons that the algorithm makes over the set of all inputs $\ranking$.
\end{definition}
The complexity measure $\DQ$ is what we refer to as query complexity. Since the maximum number of queries that an algorithm $A$ makes is at least the base-2 logarithm of the number of execution traces, we have that $\DQ(\Pi) \ge \log_2(\traces(\Pi))$. 
Note that, in order to avoid confusion with the tree $T$ specifying an instance of inversion minimization, we refrain from the common terminology of decision trees in the context of the complexity measure $\traces$. In those terms, we lower bound the number of leaves of any decision tree for $\Pi$, and use the fact that the depth of this binary decision tree is at least the base-2 logarithm of the number of leaves.

Both techniques proceed by considering the effect on the output of perturbations to the input ranking $\ranking$ that are hard for queries to observe. More specifically, we consider the following perturbations:
\begin{definition}[adjacent-rank transposition, affected items]
An \emph{adjacent-rank transposition} is a permutation $\transposition$ of $[n]$ of the form $\transposition = (r,r+1)$, where $r \in [n-1]$ and $n$ denotes the number of items. Given $\transposition$ and a ranking $\ranking: X \to [n]$, the \emph{affected items} are the two elements $x \in X$ for which $\transposition(\ranking(x)) \ne \ranking(x)$, i.e., the items with ranks $r$ and $r+1$ under $\ranking$. 
\end{definition}
As with any permutation of the set of ranks, the effect of $\transposition$ on a ranking $\ranking$ is the ranking $\transposition \ranking$. 
Adjacent-rank transpositions are the least noticeable perturbations one can apply to a ranking in the following sense: If two rankings differ by an adjacent-rank transposition, then the only query that distinguishes them is the query that compares the affected items. 

\paragraph{Sensitivity.}
Our first technique turns this observation around to obtain a lower bound on query complexity. We adopt the terminology of sensitivity from Boolean query complexity. 

\begin{definition}[sensitivity, average sensitivity, $s(\cdot)$]\label{def:sensitivity}
Let $\Pi$ be a computational problem in the comparison-query model on a set $X$ of items. For a fixed ranking $\ranking$ and adjacent-rank transposition $\transposition$, we say that $\Pi$ is \emph{sensitive} to $\transposition$ at $\ranking$ if $\Pi(\ranking) \ne \Pi(\transposition \ranking)$. The sensitivity of $\Pi$ at $\ranking$ is the number of adjacent-rank transpositions $\transposition$ such that $\Pi$ is sensitive to $\transposition$ at $\ranking$. The \emph{average sensitivity} of $\Pi$, denoted $s(\Pi)$, is the average sensitivity of $\Pi$ at $\ranking$ when $\ranking$ is drawn uniformly at random from all rankings of $X$.
\end{definition}
On input a ranking $\ranking$, any algorithm for $\Pi$ needs to make a number of queries that is at least the sensitivity of $\Pi$ at $\ranking$. Indeed, consider an adjacent-rank transposition $\transposition$ to which $\Pi$ is sensitive at $\ranking$. If the algorithm does not make the query that compares the affected items, then it must output the same answer on input $\transposition \ranking$ as on input $\ranking$. Since the value of $\Pi$ differs on both inputs, this means the algorithm makes a mistake on at least one of the two. It follows that the average number of queries that any algorithm for $\Pi$ makes is at least the average sensitivity $s(\Pi)$. A fortiori, $\DQ(\Pi) \ge s(\Pi)$. 

As sensitivity cannot exceed $n-1$, the best lower bound on query complexity that we can establish based on the above basic observation alone, is $n-1$. The following improvement yields a the lower bound $\traces(\Pi) \ge n!/n = (n-1)!$, and therefore $\DQ(\Pi) \ge \log_2(n!) - \log_2(n)$ for problems $\Pi$ of maximum average sensitivity $s(\Pi)=n-1$. The argument hinges on an efficient encoding of rankings that share the same execution trace. See Section~\ref{sec:sensitivity} for more details.

\begin{lemma}[Sensitivity Lemma]\label{lemma:sensitivity}
For any problem $\Pi$ in the comparison-query model with $n$ items, $\traces(\Pi) \ge \Gamma(s(\Pi)+2)/n$. 
\end{lemma}

The lower bound for general trees $T$ in Theorem~\ref{thm:main:general} and the strengthening for binary trees in Theorem~\ref{thm:main:binary} follow from corresponding lower bounds on the average sensitivity $s(\Pi_T)$. Theorem~\ref{thm:main:general} only requires a short analysis to establish the sensitivity lower bound needed for the application of the Sensitivity Lemma; this illustrates the power of the lemma and of the lower bound technique. Theorem~\ref{thm:main:binary} requires a more involved sensitivity analysis, but then yields a very tight lower bound. Owing to the average-case nature of the underlying measure, the technique also exhibits some degree of robustness. For the particular problem of inversion minimization on trees, we show that small changes to the tree $T$ do not affect the average sensitivity $s(\Pi_T)$ by much. See Section~\ref{sec:sensitivity:general} and Section~\ref{sec:sensitivity:binary}.

For sorting, counting inversions, and inversion parity, the average sensitivity reaches its maximum value of $n-1$, and Lemma~\ref{lemma:sensitivity} recovers the standard lower bounds up to a small loss. In contrast, for selection, the average sensitivity equals 1 for ranks 1 and $n$, and 2 for other ranks, so the bound from Lemma~\ref{lemma:sensitivity} is no good. This reflects that, just like in the Boolean setting, (average) sensitivity is sometimes too rough of a measure and not always capable of proving strong lower bounds. Our second technique looks at a more delicate structural aspect, which enables it to sometimes yield stronger lower bounds. 

\paragraph{Permutahedron graph.} Before introducing our second technique, we cast our first technique in graph theoretic terms. In fact, both our techniques can be expressed naturally in subgraphs of the graph with the rankings as vertices and adjacent-rank transpositions as edges. The latter graph can be viewed as the Cayley graph of the symmetric group with adjacent-rank transpositions as the generating set. It is also the edge graph of the permutahedron, the convex polytope spanned by all permutations of the vertex $(1,2,\dots,n)$ in $\RR^n$. The permutahedron resides inside the hyperplane where the sum of the coordinates equals $\binom{n}{2}$, has positive volume inside that hyperplane, and  can thus be represented naturally in dimension $n-1$; see Figure~\ref{fig:permutahedron} for a rendering of the instance with $n=4$ \cite{permutahedron-figure}. 

\begin{wrapfigure}{r}{0.5\linewidth}
\centerline{\includegraphics[scale=.3]{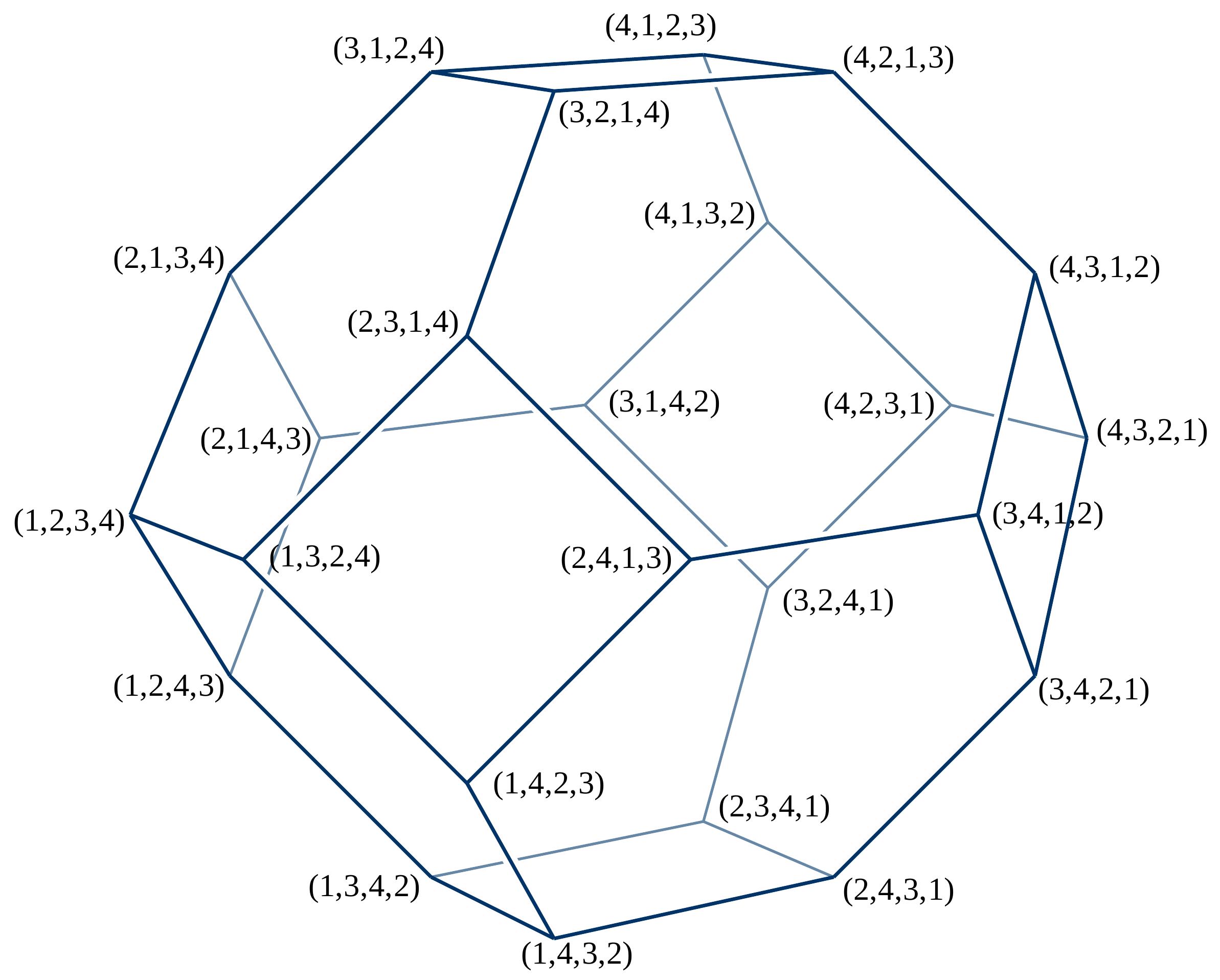}}
\caption{Permutahedron for $n=4$ items}\label{fig:permutahedron}
\end{wrapfigure}

We think of coloring the vertices of the permutahedron with their values under $\Pi$ and make use of the subgraph with the same vertex set but only containing the monochromatic edges, i.e., the edges whose end points have the same value under $\Pi$. We also consider the the complementary subgraph containing all bichromatic edges.
\begin{definition}[permutahedron graph, $G(\cdot)$, $\overline{G}(\cdot)$]
Let $\Pi$ be a computational problem in the comparison-query model on a set $X$ of items. The \emph{permutahedron graph} of $\Pi$, denoted $G(\Pi)$, has the rankings of $X$ as vertices, and an edge between two rankings $\ranking_1$ and $\ranking_2$ if $\Pi(\ranking_1) = \Pi(\ranking_2)$ and there exists an adjacent-rank transposition such that $\ranking_2 = \transposition \ranking_1$. The \emph{complementary permutahedron graph} of $\Pi$, denoted $\overline{G}(\Pi)$, is defined similarly by replacing the condition $\Pi(\ranking_1) = \Pi(\ranking_2)$ by its complement, $\Pi(\ranking_1) \ne \Pi(\ranking_2)$.
\end{definition}
Our first technique looks at degrees in the complementary permutahedron graph $\overline{G}(\Pi)$, and more specifically at the \emph{average degree} $\avgdeg(\overline{G}(\Pi)) \doteq \Expect(\deg_{\overline{G}(\Pi)}(\ranking))$, where the expectation is with respect to a uniform choice of the ranking $\ranking$. Our second technique looks at the connected components of the permutahedron graph $G(\Pi)$. 

\paragraph{Connectivity.} Our second technique is reminiscent of a result in algebraic complexity theory, where the number of execution traces of an algorithm for a problem $\Pi$ in the algebraic comparison-query model is lower bounded in terms of the number of connected components that $\Pi$ induces in its input space $\RR^n$ \cite{Ben-Or1983}. In the comparison-query setting, we obtain the following lower bound.

\begin{lemma}[Connectivity Lemma]\label{lemma:connectivity} For any problem $\Pi$ in the comparison-query model, $\traces(\Pi)$ is at least the number of connected components of $G(\Pi)$.
\end{lemma}

The Connectivity Lemma allows for a simple and unified exposition of many of the known lower bounds. For counting inversions and inversion parity the argument goes as follows. Every adjacent-rank transposition changes the number of inversions by exactly one (up or down), and therefore changes the output of $\Pi$, so all $n!$ vertices in $G(\Pi)$ are isolated. This means that any algorithm for $\Pi$ actually needs to sort and has to make at least $\log_2(n!)$ queries. See Section~\ref{sec:connectivity} for a proof of the Connectivity Lemma and more applications to classical problems, including the $\Omega(n)$ lower bound for median finding.

The Connectivity Lemma also enables us to establish strong lower bounds for inversion minimization on special types of trees $T$, namely those of Theorem~\ref{thm:main:special} and the Mann--Whitney instances in Theorem~\ref{thm:Mann-Whitney}, closely related to counting inversions. Both theorems involve an analysis of the size of the connected component of a random ranking $\ranking$ in $G(\Pi_T)$, and Theorem~\ref{thm:main:special} uses the delicate parity conditions of its statement to keep $G(\Pi_T)$ as sparse as possible. See Section~\ref{sec:connectivity:both} for more details, including a more general property that guarantees the required sparseness (the partition property from Definition~\ref{def:partition-property}) and the resulting lower bound for a generic problem $\Pi$ that satisfies the property (Lemma~\ref{lemma:connectivity:sensitive}). 

The Mann--Whitney setting illustrates well the relative power of our techniques. In the Mann--Whitney instances of inversion minimization, the leaves are naturally split between a subtree containing $a$ of them and a subtree containing $b$ of them. The argument behind Theorem~\ref{thm:main:general} yields a lower bound of $\frac{ab}{a+b}$ on the sensitivity $s(\Pi_T)$. The true sensitivity is just $O(1)$ below the one for counting cross inversions, which is $\frac{2ab}{a+b}$. The resulting lower bounds on the query complexity in case $a \le b$ are $\Theta(a \log a)$, which roughly account for sorting the smaller side but not for the $b \log_2(a)$ comparisons used in the subsequent binary searches for counting cross inversions. Our approach based on the Connectivity Lemma yields a lower bound that includes both terms. On the other hand, it is easier to estimate and obtain the lower bound via the Sensitivity Lemma than to argue the query lower bound via the Connectivity Lemma or from scratch. 

\paragraph{Other modes of computation.}
We stated our lower bounds for the standard, deterministic mode of computation. Both of our techniques provide lower bounds for the number of distinct execution traces that are needed to cover all input rankings, irrespective of whether these execution traces derive from a single algorithm. Such execution traces can be viewed as certificates or witnesses for the value of $\Pi$ on a given input $\ranking$, or as valid execution traces of a {\em nondeterministic} algorithm for $\Pi$. We define the minimum number of traces needed to cover all input rankings for a problem $\Pi$ as the nondeterministic complexity of $\Pi$ and denote it by $\NQ(\Pi)$, along the lines of the Boolean setting \cite{JRSW99}. All of our lower bounds on $\traces(\Pi)$ actually hold for $\NQ(\Pi)$. See Remark~\ref{remark:sensitivity} and Remark~\ref{remark:connectivity} for further discussion.

Since {\em randomized algorithms with zero error} are also nondeterministic algorithms, all of our lower bounds apply verbatim to the former mode of computation, as well. As for {\em randomized algorithms with bounded error}, we argue in Section~\ref{sec:random} that our lower bounds on the query complexity of inversion minimization on trees that follow from the Sensitivity Lemma carry over modulo a small loss in strength. We do so by showing generically that high average sensitivity implies high query complexity against such algorithms. 

The fact that our techniques yield lower bounds on $\NQ(\Pi)$ and not just $\traces(\Pi)$ also explains why our approaches sometimes fail. For example, for the problem $\Pi$ of finding the minimum of $n$ items, a total of $n$ certificates suffice and are needed, namely one for each possible item being the minimum. This means that our techniques cannot give a lower bound on the query complexity of $\Pi$ that is better than $\log_2(n)$. In contrast, as reviewed in Section~\ref{sec:model}, $\traces(\Pi)=2^{n-1}$ and the number of queries needed is $n-1$. 

\paragraph{Cross-inversion distribution.}
As a technical result in the sensitivity analysis for inversion minimization on binary trees (Theorem~\ref{thm:main:binary}), we need a strong upper bound on the probability that the number of cross inversions $\XInv_\ranking(A,B)$ takes on any particular value when the ranking $\ranking$ of the set $X = A \sqcup B$ is chosen uniformly at random. This is the distribution of the Mann--Whitney statistic under the null hypothesis. Mann and Whitney \cite{MannWhitney1947} argued that it converges to a normal distribution with mean $\mu=ab/2$ and variance $\sigma^2=ab(a+b+1)/12$ as $a \doteq |A|$ and $b \doteq |B|$ grow large. Since the normal distribution has a maximum density of $1/(\sqrt{2\pi}\sigma)$, their result suggests that the maximum of the underlying probability distribution is $O(1/\sigma)=O(1/\sqrt{ab(a+b+1)})$. Tak\'acs \cite{Takacs1986} managed to formally establish such a bound for all pairs $(a,b)$ with $|a-b|=O(\sqrt{a+b})$, Stanley and Zanello \cite{SZ15} for all pairs $(a,b)$ with $\min(a,b)$ bounded, and Melczer, Panova, and Pemantle \cite{MPP20} for all pairs $(a,b)$ with $|a-b|\le \alpha\cdot(a+b)$ for some constant $\alpha<1$. However, these results do not cover all regimes and leave open a single bound of the same form that applies to all pairs $(a,b)$, which is what we need for Theorem~\ref{thm:main:binary}. We establish such a bound in Section~\ref{sec:cross-inversions}. The counts of the rankings $\ranking$ with a particular value for $\XInv_\ranking(A,B)$ appear as the coefficients of the \emph{Gaussian polynomials}.
Our bound can be stated equivalently as a bound on those coefficients.

\paragraph{Organization.} 
We have organized the material so as to provide a shortest route to a full proof of Theorem~\ref{thm:main:general}. Here are the sections needed for the different main results:
\begin{itemize}
\item Theorem~\ref{thm:main:general} (lower bound for general trees): \ref{sec:sensitivity}, \ref{sec:sensitivity:general}.
\item Theorem~\ref{thm:main:binary} (lower bound for binary trees):  \ref{sec:sensitivity}, \ref{sec:sensitivity:binary}, \ref{sec:cross-inversions}.
\item Theorem~\ref{thm:main:special} (lower bound for restricted classes): \ref{sec:connectivity}, \ref{sec:connectivity:both} up to \ref{sec:connectivity:general} inclusive.
\item Theorem~\ref{thm:Mann-Whitney} (lower bound for counting cross inversions): \ref{sec:connectivity}, \ref{sec:connectivity:both} but not \ref{sec:connectivity:binary} nor \ref{sec:connectivity:general}.
\end{itemize}
In Section~\ref{sec:model}, we provide some background on known lower bounds in the comparison-query model, several of which are unified by the Sensitivity Lemma and Connectivity Lemma. In Section~\ref{sec:random}, we present our lower bounds against randomized algorithms with bounded error. The tight bound on maximum probability of the cross-inversion distribution is covered in Section~\ref{sec:cross-inversions}. For completeness, we end in Section~\ref{sec:turing} with proofs of the results we stated on the Turing complexity of inversion minimization on trees.


\section{The Comparison-Query Model}
\label{sec:model}

In this section we provide an overview of known results and techniques for lower bounds in the comparison-query model. This section can be skipped without a significant loss in continuity.

Tight bounds have been established for problems like sorting, selection, and heap construction. 
\begin{itemize}
\item We already discussed the central problem of sorting in Section~\ref{sec:overview}.
 \item In \emph{selection} we are told a rank $r$, and must identify the item with rank $r$. The query complexity is known to be $\Theta(n)$ \cite{BFPRT1973,DorZwick1995,DorZwick1996}. There is also \emph{multiple selection}, in which one is given multiple ranks $r_1,\dots,r_k$, and must identify each of the corresponding items. The query complexity of multiple selection is likewise known up to a $\Theta(n)$ gap between the upper and lower bounds \cite{KMMS2005}.
\item In \emph{heap construction} we must arrange the items as nodes in a complete binary tree such that every node has a rank no larger than its children. The query complexity is known to be $\Theta(n)$ \cite{CC92,GM86}. 
\end{itemize}

All the problems above can be cast as instantiations of a general framework known as \emph{partial order production}~\cite{Schoenhage1976}. Here, in addition to query access to the ranking $\ranking$ of the items, we are given $n$ \emph{slots} and regular access to a partial order $\slotorder$ on the slots. The objective is to put each item into a slot, one item per slot, so that whenever two slots, $s_1$ and  $s_2$, are related by $s_1 \slotorder s_2$, we also have $\ranking(s_1) < \ranking(s_2)$. Sorting coincides with the case where $\slotorder$ is a total order. In selection of rank $r$, there is a designated slot $s^*$, and there are exactly $r-1$ slots $s$ with $s \slotorder s^*$ and exactly $n-r$ slots $s$ with $s^* \slotorder s$; there are no other relations in $\slotorder$ (see the example at the end of Section~\ref{sec:connectivity} for more details).  Multiple selection is similar. For heap construction, $\slotorder$ matches the complete binary tree arrangement. 

Partial order production for a given $\slotorder$ naturally decomposes into the same problem for each of the connected components of the undirected graph underlying $\slotorder$. In the case of a single connected component, an elementary adversary argument shows that $\DQ(\Pi) \ge n-1$: Any combination of less than $n-1$ queries to $\ranking$ leaves some pair of slots in $\slotorder$ undetermined with respect to $\ranking$. Another lower bound is the information-theoretic limit. For each way of putting items into slots, the number of input rankings $\ranking$ for which that way is a correct answer is bounded by $e(\slotorder)$, the number of ways to extend $\slotorder$ to a total order. Therefore, there must be at least $n!/e(\slotorder)$ distinct execution traces. Since each execution trace is determined by the outcomes of its queries, and each query has only two outcomes, we conclude that $\lambda(\slotorder) \doteq \log_2(n!/e(\slotorder))$ queries are necessary to solve partial order production. Complementing these lower bounds there exists an upper bound of $(1 + o(1))\cdot \lambda(\slotorder) + c \cdot (n-1)$ queries for some universal constant $c$ \cite{CFJJM2008}. For a generic instance $\Pi$ with partial order $\slotorder$ it follows that $\DQ(\Pi)= \Theta(\lambda(\slotorder) + n - \gamma(\slotorder))$, where $\gamma(\slotorder)$ denotes the number of connected components of the undirected graph underlying $\slotorder$. 

Not every problem of interest in the comparison model is an instance of partial order production. Here are a few examples.
\begin{itemize}
\item In \emph{rank finding} there is a designated item $x^*$, and we have to compute its rank. The rank can be computed by comparing $x^*$ with each of the $n-1$ other items. A similar elementary adversary argument as above shows that the query complexity is at least $n-1$.
\item In \emph{counting inversions} the items are arranged in some known order $\ordering$ and the objective is to count the number of inversions of $\ordering$ with respect to $\ranking$. As we reviewed in Section~\ref{sec:overview}, counting inversions has exactly the same query complexity as sorting.
\item The problem of \emph{inversion parity} is the same as counting inversions except that one need only count the number of inversions modulo 2. This problem, as well as counting inversions modulo $m$ for any integer $m > 1$, also has exactly the same complexity as sorting.
\end{itemize}

For each of the three problems above, information theory does not provide a satisfactory lower bound. For example, in the inversion parity problem there are only two possible outputs, which yields a lower bound of $\log_2(2)=1$. It so happens that for each of the preceding three examples, the query complexity is known quite precisely; however, the known arguments are rather problem-specific.

\emph{Inversion minimization on trees} is another example that does not fit the framework of partial order generation, and for which information theory only yields a weak lower bound: $\log_2 \binom{n}{2} = 2 \log_2(n) - \Theta(1)$. In contrast to the above examples, a strong lower bound does not seem to follow from a simple ad-hoc argument nor from a literal equivalence to sorting. 
\section{Sensitivity Lemma}
\label{sec:sensitivity}

In this section we develop Lemma~\ref{lemma:sensitivity}.
We actually prove a somewhat stronger version.

\begin{lemma}[Strong Sensitivity Lemma]\label{lemma:sensitivity:strong}
Consider an algorithm $A$ in the comparison-based
model with $n$ items, color each vertex of the permutahedron with its execution trace under $A$, and let $\overline{H}$ denote the subgraph with the same vertex set but only containing the bichromatic edges. The number of distinct execution traces of $A$ is at least $g(\avgdeg(\overline{H})+1)/n$, where
$g: [1,\infty) \to \mathbb{R}$ is any convex function with $g(x) = x!$ for $x \in [n]$.
\end{lemma}
The Sensitivity Lemma follows from Lemma~\ref{lemma:sensitivity:strong} because the coloring with execution traces of an algorithm $A$ for $\Pi$ is a refinement of the coloring with $\Pi$, so every edge of the permutahedron that is bichromatic under $\Pi$ is also bichromatic under $A$, and
\[
s(\Pi) \doteq \Expect(\deg_{\overline{G}(\Pi)}(\ranking)) \le \Expect(\deg_{\overline{H}}(\ranking)) \doteq \avgdeg(\overline{H}). 
\]
Provided $g$ is nondecreasing, it follows that $\traces(\Pi) \ge g(\avgdeg(\overline{H})+1)/n \ge g(s(\Pi)+1)/n$.

In the Sensitivity Lemma we set $g(x) = \Gamma(x+1)$. An optimal (but less elegant) choice for $g$ is the piece-wise linear function that interpolates the prescribed values at the integral points in $[n]$, namely
\[    g(x) \doteq
      (x - \floor{x})\cdot(\ceil{x}!)
      +
      (1 - (x - \floor{x}))\cdot(\floor{x}!).
\]      

For the proof of Lemma~\ref{lemma:sensitivity:strong} we take intuition from a similar result in the Boolean setting \cite[Exercise~8.43]{ODonnell2014}, where the hypercube plays the role of the permutahedron in our setting.
\begin{fact}\label{fact:Boolean} 
Let $A$ be a query algorithm on binary strings of length $n$. Color each vertex of the $n$-dimensional hypercube by its execution trace under $A$, and let $\overline{H}$ denote the subgraph with the same vertex set but only containing the bichromatic edges. Then the number of distinct execution traces is at least $2^{\avgdeg(H)}$.
\end{fact}

One way to argue Fact~\ref{fact:Boolean} is to think of assigning a weight $w(x)$ to each $x \in \{0,1\}^n$ so as to maximize the total weight on all inputs, subject to the constraint that the total weight on each individual execution trace is at most 1. Then the number of distinct execution traces must be at least the sum of all the weights. If the weight only depends on the degree, i.e., if we can write $w(x) = f(\deg_{\overline{H}}(x))$ for some function $f: [0,\infty) \to \mathbb{R}$, then we can lower bound the number $k$ of distinct execution traces as follows:
\begin{equation}\label{eq:Boolean}
k \ge \sum_x w(x) = \sum_x f(\deg_{\overline{H}}(x)) \ge 2^n \cdot f(\Expect(\deg(x))) = 2^n \cdot f(\avgdeg(\overline{H})),
\end{equation}
where the last inequality holds provided $f$ is convex. 

In the Boolean setting, the set $R$ of inputs $x \in \{0,1\}^n$ with a particular execution trace forms a subcube of dimension $n-\ell$, where $\ell$ denotes the length of the execution trace, i.e., the number of queries. Each $x \in R$ has degree $\ell$ in $H$; this is because a change in a single queried position results in a different execution trace, and a change in an unqueried position does not. Therefore, a natural choice for the weight of $x \in R$ is $w(x) = f(\ell)$ where $f(x) = 1/2^{n-\ell}$. It satisfies the constraint that the total weight on $R$ is (at most) one, and $f$ is convex. We conclude by \eqref{eq:Boolean} that the number of distinct execution traces is at least $2^n \cdot f(\avgdeg(\overline{H})) = 2^{\avgdeg(\overline{H})}$, as desired.

\begin{proof}[Proof of Lemma~\ref{lemma:sensitivity:strong}]
Let $k$ denote the number of distinct execution traces of $A$, and let $R_1, \dots, R_k$ denote the corresponding sets of rankings. 
Following a similar strategy, we want to find a convex function $f: [0,\infty)\to \RR$ such that the weight function $w(\ranking) = f(\deg_{\overline{H}}(\ranking))$ does not assign weight more than 1 to any one of the sets $R_i$. The following claim, to be proven later, is the crux of this.

\begin{claim}\label{claim:encoding}
Let $R$ denote the set of all rankings $\ranking$ that follow a particular execution trace on $A$, and let $d \in \{0,\dots,n-1\}$. The number of rankings $\ranking \in R$ with $\deg_{\overline{H}}(\ranking) = d$ is at most $\frac{n!}{(d+1)!}$.
\end{claim}

Based on Claim~\ref{claim:encoding}, a natural choice for $f$ is any convex function that satisfies $f(x) = \frac{1}{n} \frac{(x+1)!}{n!}$ for $x \in \{0,\dots,n-1\}$. The factor of $\frac{1}{n}$ comes from the fact that there are $n$ terms to sum together after the weights have been normalized. For every $i \in [k]$ we then have
\[ \sum_{\ranking \in R_i} w(\ranking) = \sum_{d=0}^{n-1} \left| \{ \ranking \in R_i \, : \, \deg_{\overline{H}}(\ranking) = d \} \right| \cdot f(d) \le \sum_{d=0}^{n-1} \frac{1}{n} \frac{(d+1)!}{n!} \cdot \frac{n!}{(d+1)!} = \sum_{d=0}^{n-1} \frac{1}{n} = 1. \]
Similar to \eqref{eq:Boolean} we conclude
\begin{equation}\label{eq:lb:k}
k \ge \sum_{i=1}^k \sum_{\ranking \in R_i} w(\ranking) 
= \sum_\ranking w(\ranking)
= \sum_\ranking f(\deg_{\overline{H}}(\ranking)) \ge n! \cdot f(\Expect(\deg_{\overline{H}}(\ranking))) = n! \cdot f(\avgdeg(\overline{H})). 
\end{equation}
Setting $f(x) = \frac{1}{n} \frac{g(x+1)}{n!}$ turns the requirements for $f$ into those for $g$ in the statement of the lemma, and yields that $k \ge n! \cdot f(\avgdeg(\overline{H})) = g(\avgdeg(\overline{H})+1)/n$. 
\end{proof}

We now turn to proving Claim~\ref{claim:encoding}. The comparisons and outcomes that constitute a particular execution trace of $A$ can be thought of as directed edges between the items in $X$. We refer to the resulting digraph on the vertex set $X$ as the comparison graph $C$. Since the outcomes of the comparisons are consistent with some underlying ranking, the digraph $C$ is acyclic. The rankings in $R$ are in one-to-one and onto correspondence with the linear orderings of the DAG $C$. For a given ranking $\ranking \in R$, the degree $\deg_{\overline{H}}(\ranking)$ equals the number of $r \in \{2,\dots,n\}$ such that swapping ranks $r-1$ and $r$ in $\ranking$ results in a ranking $\ranking' = \transposition \ranking$ that is not in $R$, where $\transposition$ denotes the adjacent-rank transposition $(r-1,r)$. The ranking $\ranking'$ not being in $R$ means that it is inconsistent with the combined comparisons and outcomes of the underlying execution trace, which happens exactly when there is a path in $C$ from the item $\ranking^{-1}(r-1)$ of rank $r-1$ in $\ranking$ to the item $\ranking^{-1}(r)$ with rank $r$ in $\ranking$. Thus, the degree $\deg_{\overline{H}}(\ranking)$ equals the number of $r \in \{2,\dots,n\}$ such that there is a path from $\ranking^{-1}(r-1)$ to $\ranking^{-1}(r)$ in $C$. See Figure~\ref{fig:encoding} for an illustration, where a squiggly edge $u \leadsto v$ denotes that there exists a path from $u$ to $v$ in $C$. We only draw squiggly edges from one position to the next, so $\deg_{\overline{H}}(\ranking)$ equals the number of squiggly edges in Figure~\ref{fig:encoding}. 

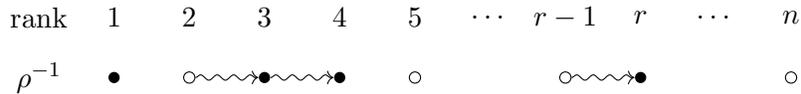
\begin{figure}[h]
    \centering
    \begin{tikzpicture}
        \node at (0,0)[]{rank};
        \node at (1,0)[]{$1$};
        \node at (2,0)[]{$2$};
        \node at (3,0)[]{$3$};
        \node at (4,0)[]{$4$};
        \node at (5,0)[]{$5$};
        \node at (6,0)[]{$\dots$};
        \node at (7,0)[]{$r-1$};
        \node at (8,0)[]{$r$};
        \node at (9,0)[]{$\dots$};
        \node at (10,0)[]{$n$};
        
        \node at (0,-0.8)[]{$\ranking^{-1}$};
        \node (1) at (1,-0.8)[circle, fill, inner sep=1.5pt]{};
        \node (2) at (2,-0.8)[circle, draw, inner sep=1.5pt]{};
        \node (3) at (3,-0.8)[circle, fill, inner sep=1.5pt]{};
        \node (4) at (4,-0.8)[circle, fill, inner sep=1.5pt]{};
        \node (5) at (5,-0.8)[circle, draw, inner sep=1.5pt]{};
        \node (r) at (7,-0.8)[circle, draw, inner sep=1.5pt]{};
        \node (r1) at (8,-0.8)[circle, fill, inner sep=1.5pt]{};
        \node (n) at (10,-0.8)[circle, draw, inner sep=1.5pt]{};
        
        \draw[->, decorate, decoration=snake, segment amplitude=1pt, segment length=2.2mm] (2)--(3);
        \draw[->, decorate, decoration=snake, segment amplitude=1pt, segment length=2.2mm] (3)--(4);
        \draw[->, decorate, decoration=snake, segment amplitude=1pt, segment length=2.2mm] (r)--(r1);

    \end{tikzpicture}
\caption{Ranking encoding}\label{fig:encoding}
\end{figure}

Our strategy is to give a \emph{compressed encoding} of the rankings in $R$ such that there is more compression as the number of squiggly edges increases. Our encoding is based on the well-known algorithm to compute a linear order of a DAG. Algorithm~\ref{alg:linear-order} provides pseudocode for
the algorithm, which we refer to as BuildRanking.

\begin{algorithm}[ht]
  \caption{BuildRanking$(C)$}\label{alg:linear-order}
  \begin{algorithmic}[1]
    \Require DAG $C$ on vertex set $X$
    \Ensure ranking of $X$ that is a linear order of $C$
    \State $T \gets \varnothing$
    \For{\(r=1\) to \(n\)}
      \State \label{alg:linear-order:step}%
      $x \gets$ arbitrary element of $S \doteq  \{ v \in X \;\mid\; \text{there is no } u \in X \setminus T \text{ with } u \leadsto v \text{ in } G \}$
      \State $\ranking^{-1}(r) \gets x$ 
      \State $T \gets T \cup \{x\}$
    \EndFor
  \end{algorithmic}
\end{algorithm}

In our formulation, BuildRanking is nondeterministic: There is a choice to make in step~\ref{alg:linear-order:step} for each $r=1,\dots,n$. The possible executions of BuildRanking are in one-to-one and onto correspondence with the linear orders of $C$, and thus with the rankings in $R$. 

Our encoding is a compressed description of how to make the decisions in BuildRanking such that the output is $\ranking$. Note that if $\ranking^{-1}(r-1) \leadsto \ranking^{-1}(r)$, then the item $x$ with rank $r$ cannot enter the set $S$ before iteration $r$. This is because before $\ranking^{-1}(r-1)$ is removed from $T$ at the end of iteration $r-1$, the edge $\ranking^{-1}(r-1) \leadsto \ranking^{-1}(r)$ prevents $x$ from being in $S$. Thus, whenever $\ranking^{-1}(r-1) \leadsto \ranking^{-1}(r)$, the item $x=\ranking^{-1}(r)$ is \emph{lucky} in the sense that it gets picked in step~\ref{alg:linear-order:step} as soon as it enters the set $S$. In fact, the lucky items with respect to a ranking $\ranking \in R$ are exactly those for which $\ranking^{-1}(r-1) \leadsto \ranking^{-1}(r)$ for some $r \in \{2,\dots,n\}$, as well as the item $\ranking^{-1}(1)$ with rank 1. In Figure~\ref{fig:encoding} the lucky items are marked black. Their number equals $\deg_{\overline{H}}(\ranking)+1$. 

In order to generate a ranking $\ranking$ using BuildRanking, it suffices to know:
\begin{itemize}
\item[(a)] the lucky items (as a set, not their relative ordering), and 
\item[(b)] the ordering of the non-lucky items (given which items they are).
\end{itemize}
This information suffices to make the correct choices in step~\ref{alg:linear-order:step} of Algorithm~\ref{alg:linear-order}:
\begin{itemize}
\item 
If the set $S$ contains a lucky item, there will be a unique lucky item in  $S$; pick it as the element $x$.
\item 
Otherwise, pick for $x$ the first item in the ordering of the non-lucky items that is not yet in $T$. Such an element will exist, and all the items that come after it in the ordering are not yet in $T$ either.
\end{itemize}
If $\ranking$ has degree $d = \deg_{\overline{H}}(\ranking)$, then there are $d+1$ lucky items, so there are at most $\binom{n}{d+1}$ choices for (a), and at most $(n-d-1)!$ choices for (b), resulting in a total of at most $\binom{n}{d+1} \cdot (n-d-1)! = \frac{n!}{(d+1)!}$ choices. This proves Claim~\ref{claim:encoding}.

\begin{remark}\label{remark:sensitivity}
Suppose we allow an algorithm $A$ to have multiple valid execution traces on a given input $\ranking$, and let $R_i$ denote the set of rankings on which the $i$-th execution trace is valid. The proof of Claim~\ref{claim:encoding} carries over as it considers individually sets $R_i$, and only depends on the DAG that the comparisons in $R_i$ induce. The rest of the proof of Lemma~\ref{lemma:sensitivity:strong} carries through modulo the first equality in \eqref{eq:lb:k}, which no longer holds as the sets $R_i$ may overlap. However, the equality can be replaced by the inequality $\ge$, which does hold and is sufficient for the argument. This means that we can replace $\traces(\Pi)$ in the statement of the Sensitivity Lemma by its nondeterministic variant $\NQ(\Pi)$.
\end{remark}

\section{Sensitivity Approach for General Trees}
\label{sec:sensitivity:general}

In this section we analyze the average sensitivity of the problem $\Pi_T$ of inversion minimization on a tree $T$ with a general shape. In Section~\ref{sec:sensitivity:subtree} we show that the existence of a subtree containing a fair fraction of the leaves implies high sensitivity. The lower bound on query complexity for $\Pi_T$ in Theorem~\ref{thm:main:general} then follows from the Sensitivity Lemma. In Section~\ref{sec:lipschitz} we prove that the average sensitivity measure is Lipschitz continuous. For the analysis, we make use of the decomposition of the objective of inversion minimization on trees mentioned earlier. We describe the decomposition in more detail in Section~\ref{sec:decomposition}; it will be helpful in later parts of this paper, as well. 

\subsection{Subtree-induced sensitivity}
\label{sec:sensitivity:subtree}

\begin{wrapfigure}{r}{0.4\linewidth}
    \centering
    \begin{tikzpicture}
        \node (top) at (0,0)[circle, draw, inner sep=2pt]{};
        \node (v) at (-0.5,-1.5)[circle, draw, inner sep=2pt, label={right:$v$}]{};
        \draw[decorate, decoration=snake, segment amplitude=1pt, segment length=2mm] (top)--(v);
        \draw[thin, gray] (-2.1,-3.1) -- (top) -- (2.1,-3.1);
        \draw[thin, gray] (-1.5,-3) -- (v) -- (0.5,-3);
        \draw (-2.1,-3.1)--(2.1,-3.1);
        \draw (-1.5,-3)--(0.5,-3);
        \node at (-0.5,-3.1)[label={$\leafset(T_v)$}]{};
        \node at (0,-3)[label={below:$\leafset(T)$}]{};
    \end{tikzpicture}
    \caption{Subtree rooted at $v$}
    \label{fig:k-ary:subtree}
\end{wrapfigure}
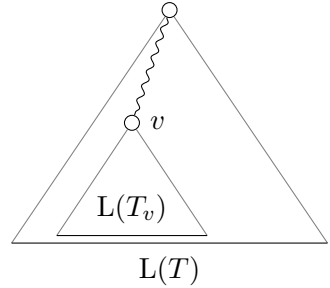

We first introduce a sensitivity bound for inversion minimization based on the size of a subtree. 

\begin{lemma}[subtree-induced sensitivity]
    \label{lemma:k-ary:subtree-sensitivity}
    Consider a tree $T$ with $n \doteq \lvert\leafset(T)\rvert$ leaves, and some node $v$ in $T$ with $\ell \doteq \lvert\leafset(T_v)\rvert$ leaves. We have
    \[s(\Pi_T) \ge \frac{\ell(n-\ell)}{n}-1. \]
\end{lemma}
Note that $v$ is not necessarily a direct child of the root, as shown in \cref{fig:k-ary:subtree}.

We now prove \cref{lemma:k-ary:subtree-sensitivity}. Let $\ranking$ be a ranking of the leaves of $T$, and let $\ordering_{\min}$ be a tree ordering that minimizes the number of inversions with respect to $\ranking$.

\begin{claim}
    \label{claim:k-ary:decreasing-transposition}
    $\ranking$ is sensitive to the transposition $\transposition=(r,r+1)$ if $\ordering_{\min}(\ranking^{-1}(r))>\ordering_{\min}(\ranking^{-1}(r+1))$.
\end{claim}

\begin{proof}
    If $\ordering_{\min}(\ranking^{-1}(r))>\ordering_{\min}(\ranking^{-1}(r+1))$, then $\Inv_{\transposition \ranking}(\ordering_{\min})=\Inv_{\ranking}(\ordering_{\min})-1$. Since $\Inv_{\ranking}(\ordering_{\min})=\MInv(T,\ranking)$, this means that $\MInv(T,\transposition \ranking)< \MInv(T,\ranking)$, or that $\ranking$ is sensitive to $\transposition$.
\end{proof}

In the case of general trees, a tree ordering $\ordering$ that minimizes the number of inversions with respect to $\ranking$ is difficult to find (see the discussion on NP-hardness in \cref{sec:turing}). Our strategy is to find a lower bound on the number of $r$ for which $\ordering(\ranking^{-1}(r))>\ordering(\ranking^{-1}(r+1))$ that applies regardless of $\ordering$.

\begin{claim}
    \label{claim:k-ary:transposition-to-gap}
    For any ordering $\ordering$, the number of $r$ such that $\ordering(\ranking^{-1}(r))>\ordering(\ranking^{-1}(r+1))$ is at least one less than the number of $s$ such that $\ranking^{-1}(s)\in \leafset(T_v)$ and $\ranking^{-1}(s+1)\not\in \leafset(T_v)$.
\end{claim}

\begin{proof}
    For all except at most one value of $s$ (the maximum $s$ for which $\ranking^{-1}(s)\in \leafset(T_v)$), there exists a minimal $s'>s$ such that $\ranking^{-1}(s')\in \leafset(T_v)$. We claim that at least one value of $r=s,\dots,s'-1$ satisfies $\ordering(\ranking^{-1}(r))>\ordering(\ranking^{-1}(r+1))$. If not, then $\ordering$ would rank $\ranking^{-1}(s),\ranking^{-1}(s+1),\dots,\ranking^{-1}(s')$ in increasing order. Because $\ordering$ is a tree ordering, the leaves of $\leafset(T_v)$ must be mapped into a contiguous range by $\ordering$, as shown in \cref{fig:k-ary:transposition-to-gap-visualization}. However, we have $\ranking^{-1}(s),\ranking^{-1}(s')\in \leafset(T_v)$ but $\ranking^{-1}(s+1)\not\in \leafset(T_v)$, which violates this property since $\ordering$ ranks a leaf outside $\leafset(T_v)$ between two leaves inside $\leafset(T_v)$.
    
    Because each value of $r$ is found between consecutive pairs of values in $\leafset(T_v)$, the values of $r$ are distinct.
\end{proof}

\begin{figure}[ht]
    \centering
    \begin{subfigure}[b]{0.45\textwidth}
        \centering
        \begin{tikzpicture}
            [
            level 1/.style = {sibling distance = 4cm, level distance = 0.75cm},
            level 2/.style = {sibling distance = 1.5cm, level distance = 1cm},
            level 3/.style = {sibling distance = 1.4cm, level distance = 1cm},
            level 4/.style = {sibling distance = 0.9cm, level distance = 1cm}
            ]
            \node at (0,0)[circle, draw, inner sep=1.5pt]{}
                child
                {node [circle, draw, inner sep=1.5pt]{}
                    child
                    {node [circle, draw, inner sep=1.5pt, label={below:$2$}]{}}
                    child
                    {node [circle, draw, inner sep=1.5pt, label={below:$7$}]{}}}
                child
                {node [circle, draw, inner sep=1.5pt]{}
                    child
                    {node [circle, draw, inner sep=1.5pt, label={above:$v$}, sibling distance=3cm]{}
                        child
                        {node [circle, draw, inner sep=1.5pt]{}
                            child
                            {node [circle, fill, inner sep=1.5pt, label={below:$1$}]{}}
                            child
                            {node [circle, fill, inner sep=1.5pt, label={below:$5$}]{}}}
                        child
                        {node [circle, draw, inner sep=1.5pt]{}
                            child
                            {node [circle, fill, inner sep=1.5pt, label={below:$4$}]{}}
                            child
                            {node [circle, fill, inner sep=1.5pt, label={below:$6$}]{}}}
                        child
                        {node [circle, fill, inner sep=1.5pt, label={below:$8$}]{}}}
                    child
                    {node [circle, draw, inner sep=1.5pt, label={below:$3$}]{}}
                    child
                    {node [circle, draw, inner sep=1.5pt, label={below:$9$}]{}}};
        \end{tikzpicture}
        \vspace{0.6cm}
        \caption{Leaves in $\ordering$-order, labeled with $\ranking$-ranks}
    \end{subfigure}
    \hfill
    \begin{subfigure}[b]{0.45\textwidth}
        \centering
        \begin{tikzpicture}[scale=0.5, yscale=0.8]
            \draw[very thin, gray] (0,0) grid (10,10);
            \draw[->,thick] (0,0)--(10,0) node [label={below:$\ranking$}]{};
            \draw[->,thick] (0,0)--(0,10) node [label={left:$\ordering$}]{};
            \draw[dashed,thick] (0,3)--(10,3);
            \draw[dashed,thick] (0,7)--(10,7);
            \node at (10.1,5)[label={right:$\leafset(T_v)$}]{};
            \draw[decorate, decoration=brace] (10.2,7)--(10.2,3);
            
            \node at (3,-1.5){\small $r_1$};
            \node at (6,-1.5){\small $r_2$};
            
            \node at (1,3)[circle, fill, inner sep=1.5pt]{};
            \node at (2,1)[circle, draw, fill=white, inner sep=1.5pt]{};
            \node at (3,8)[circle, draw, fill=white, inner sep=1.5pt]{};
            \node at (4,5)[circle, fill, inner sep=1.5pt]{};
            \node at (5,4)[circle, fill, inner sep=1.5pt]{};
            \node at (6,6)[circle, fill, inner sep=1.5pt]{};
            \node at (7,2)[circle, draw, fill=white, inner sep=1.5pt]{};
            \node at (8,7)[circle, fill, inner sep=1.5pt]{};
            \node at (9,9)[circle, draw, fill=white, inner sep=1.5pt]{};
            
            \foreach \i in {1,2,3,4,5,6,7,8,9}{
                \node at (0.3,{\i})[label={left:\scriptsize {\i}}]{};
                \node at ({\i},0.3)[label={below:\scriptsize {\i}}]{};
            }
            
            \draw (1.5,-2.4)--(3.5,-2.4);
            \draw (6.5,-2.4)--(7.5,-2.4);
            \node at (1,-2.4)[]{\small $s_1$};
            \node at (4,-2.3)[]{\small $s'_1$};
            \node at (6,-2.4)[]{\small $s_2$};
            \node at (8,-2.3)[]{\small $s'_2$};
        \end{tikzpicture}
        \caption{Corresponding plot of $\ranking,\ordering$ for each leaf}
    \end{subfigure}
    \caption{$\ordering$ maps leaves of $\leafset(T_v)$ in a contiguous range.}
    \label{fig:k-ary:transposition-to-gap-visualization}
\end{figure}
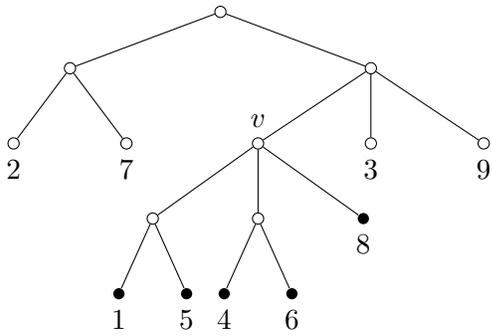
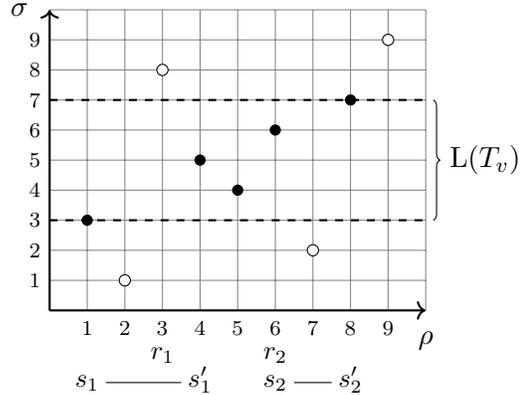

\begin{claim}
    \label{claim:k-ary:gap-count}
    Over a uniformly random $\ranking$, the expected number of $s$ such that $\ranking^{-1}(s)\in \leafset(T_v)$ and $\ranking^{-1}(s+1)\not\in \leafset(T_v)$ is $\frac{\ell(n-\ell)}{n}$.
\end{claim}

\begin{proof}
    For $s=1,\dots,n-1$, the probability that $\ranking^{-1}(s)\in \leafset(T_v)$ is $\frac{\ell}{n}$, and the probability that $\ranking^{-1}(s+1)\not\in \leafset(T_v)$ given that $\ranking^{-1}(s)\in \leafset(T_v)$ is $\frac{n-\ell}{n-1}$. Using linearity of expectation on the indicator random variables for $\ranking^{-1}(s)\in \leafset(T_v)$ and $\ranking^{-1}(s+1)\not\in \leafset(T_v)$, the expected number of $s$ satisfying this property is
    \[(n-1)\left(\frac{\ell(n-\ell)}{n(n-1)}\right)=\frac{\ell(n-\ell)}{n}. \]
\end{proof}

Combining \cref{claim:k-ary:decreasing-transposition}, \cref{claim:k-ary:transposition-to-gap}, and \cref{claim:k-ary:gap-count}, we can conclude with \cref{lemma:k-ary:subtree-sensitivity}.

\paragraph{Bounded degree.}
We apply our analysis to the case of trees of degree $k$. Observe that for fixed $n$, \cref{lemma:k-ary:subtree-sensitivity} is strongest when $\ell=n/2$. Not every tree $T$ has a subtree with exactly $n/2$ leaves, but \cref{lemma:k-ary:subtree-sensitivity} still gives a useful bound for subtrees that do not contain too few or too many leaves. In the case of trees of bounded degree, there always exists a subtree $T_v$ that contains a fairly balanced fraction of the leaves. The following quantification is folklore, but we include a proof for completeness.
 
\begin{fact}
    \label{fact:k-ary:balance}
    If $T$ is a tree of degree $k$ with $n$ leaves, there exists a node $v$ in $T$ such that $\ell \doteq \lvert\leafset(T_v)\rvert = \alpha 
    \cdot n$, where $\frac{1}{k+1} \le \alpha \le \frac{k}{k+1}$.
\end{fact}

\begin{proof}
    Let $r$ be the root of $T$ and construct a sequence $v_1=r,v_2,v_3,\dots$ such that $v_{i}$ is a child of $v_{i-1}$ that maximizes $\ell_{i} \doteq \lvert\leafset(T_{v_{i}})\rvert$, with ties broken arbitrarily. Notice that $\{\ell_i\}$ is a decreasing sequence, and since $T$ has degree $k$, $\ell_{i} \le k\ell_{i+1}$ for all $i$. We claim that some $v_i$ in this sequence satisfies the conditions of the claim. If not, then for some $i$, $\ell_i> \frac{k}{k+1}\cdot n$ and $\ell_{i+1}< \frac{1}{k+1}\cdot n$, which contradicts the fact that $\ell_i\le k\ell_{i+1}$.
\end{proof}

By choosing a subtree satisfying \cref{fact:k-ary:balance}, we can apply \cref{lemma:k-ary:subtree-sensitivity} and conclude that $s(\Pi_T)\ge \frac{k}{(k+1)^2} \cdot n -1$. The Sensitivity Lemma then gives the ``in particular'' part of Theorem~\ref{thm:main:general}.


\subsection{Decomposition of the objective function}
\label{sec:decomposition}

For use in this section as well as later parts of the paper, we now explain how the objective of inversion minimization on trees decomposes. We introduce the notion of root inversion along the way, and observe the effect of adjacent-rank transpositions on the decomposition.

The objective $\MInv(T,\ranking)$ can be written as the sum of contributions from each of the individual nodes. A node $v$ contributes those inversions that reside in the subtree $T_v$ and go through the root $v$ of $T_v$. We refer to them as the root inversions in $T_v$. 
\begin{definition}[root inversions, $\RInv(\cdot,\cdot,\cdot)$, $\MRInv(\cdot,\cdot)$]
Given a tree $T$, a ranking $\ranking$ of the leaves of $T$, and an ordering $\ordering$ of $T$, a root inversion of $\ordering$ with respect to $\ranking$ is an inversion $(\ell_1,\ell_2)$ of $\ordering$ with respect to $\ranking$ for which the lowest common ancestor $\LCA(\ell_1,\ell_2)$ is the root of $T$. The number of root inversions of $\ordering$ with respect to $\ranking$ in $T$ is denoted by $\RInv(T,\ranking,\ordering)$. The minimum number of root inversions in $T$ with respect to $\ranking$ is denoted 
\begin{equation}\label{eq:MRInv}
\MRInv(T,\ranking) \doteq \min_\ordering \RInv(T,\ranking,\ordering),
\end{equation}
where $\ordering$ ranges over all possible orderings of $T$.
\end{definition}
The only aspect of the ordering $\ordering$ of $T_v$ that affects $\RInv(T_v,\ranking,\ordering)$ is the relative order of the children of $v$. For a node $v$ with $k$ children $u_1,\dots,u_k$, by abusing notation and using $\ordering$ to also denote the ranking of the children induced by the ordering of the tree, we have
\begin{equation}\label{eq:cost}
\RInv(T,\ranking,\ordering) \doteq \sum_{1 \le i < j \le k} \XInv_\ranking(L_{\ordering(i)},L_{\ordering(j)}),
\end{equation}
where $L_i$ is a short-hand for the leaf set $L(T_{u_i})$. The contributions of the nodes can be optimized independently:
\begin{equation}\label{eq:MInv:decomposition}
\MInv(T,\ranking) = \sum_v \MRInv(T_v,\ranking),
\end{equation}
where $v$ ranges over all nodes of $T$ with degree $\deg_T(v)>1$. 

When we apply an adjacent-rank transposition $\transposition$ to a ranking $\ranking$, at most one of terms in the decomposition \eqref{eq:MInv:decomposition} can change, and the change is at most one unit. We capture this observation for future reference
as it will be helpful in several sensitivity analyses.
\begin{proposition}\label{prop:change}
Let $\ranking$ be a ranking of the leaf set $X$ of a tree $T$, $\transposition$ an adjacent-rank transposition, and $\ell_1$ and $\ell_2$ be the affected leaves. Then
\[ \MRInv(T_v,\ranking) = \MRInv(T_v,\transposition \ranking) \]
for all nodes $v$ in $T$ except possibly $v = \LCA(\ell_1,\ell_2)$. Moreover, the difference is at most 1 in absolute value.
\end{proposition}
\begin{proof}
Since the ranks of $\ell_1$ and $\ell_2$ under $\ranking$ are adjacent, for any leaf $\ell$ other than $\ell_1$ and $\ell_2$, the relative order of $\ell$ under $\ranking$ is the same with respect to $\ell_1$ as it is with respect to $\ell_2$. This means that the adjacent-rank transposition $\transposition$ does not affect whether a pair of leaves constitutes an inversion unless that pair equals $\{\ell_1,\ell_2\}$. As a result, the only term on the right-hand side of \eqref{eq:MInv:decomposition} that can be affected by the transposition $\transposition$ is the one corresponding to the node $v$, and it can change by at most one unit.
\end{proof}

\subsection{Lipschitz continuity}
\label{sec:lipschitz}

Average-case notions typically do not change much under small changes to the input. This is indeed the case for the average sensitivity when ``small'' is interpreted as affecting few of the subtrees. The following lemma quantifies the property and can be viewed as a form of Lipschitz continuity.
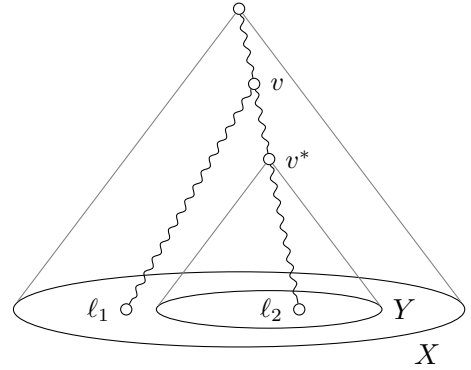
\begin{wrapfigure}{r}{0.375\linewidth}
    \centering
    \begin{tikzpicture}[]
        \node (top) at (0,0)[circle, draw, inner sep=1.5pt]{};
        \node (v) at (0.2,-1)[circle, draw, inner sep=1.5pt, label={right:\small $v$}]{};
        \node (v1) at (0.4,-2)[circle, draw, inner sep=1.5pt, label={right:\small $v^*$}]{};
        
        \node (l1) at (-1.5,-4)[circle, draw, inner sep=1.5pt, label={left:\small $\ell_1$}]{};
        \node (l2) at (0.8,-4)[circle, draw, inner sep=1.5pt, label={left:\small $\ell_2$}]{};
        
        \draw (0,-4) ellipse (3 and 0.5);
        \draw (0.4,-4) ellipse (1.5 and 0.25);
        
        \draw[thin,gray] (-2.95,-3.9)--(top)--(2.95,-3.9);
        \draw[thin,gray] (-1.075,-3.95)--(v1)--(1.875,-3.95);
        
        \draw[decorate, decoration=snake, segment amplitude=1pt, segment length=2.2mm] (top)--(v);
        \draw[decorate, decoration=snake, segment amplitude=1pt, segment length=2.2mm] (v)--(v1);
        \draw[decorate, decoration=snake, segment amplitude=1pt, segment length=2.2mm] (v)--(l1);
        \draw[decorate, decoration=snake, segment amplitude=1pt, segment length=2.2mm] (v1)--(l2);
        
        \node at (2.5,-4.6)[]{$X$};
        \node at (2.2,-4)[]{$Y$};
    \end{tikzpicture}
    \caption{Effects of changing $T_{v^*}$}
    \label{fig:lipschitz}
\end{wrapfigure}

\begin{lemma}
    \label{lemma:sensitivity:lipschitz}    
    Given a tree $T$, if a subtree $T_{v^*}$ with $\ell$ leaves is replaced with a tree $T'_{v^*}$ with the same number of leaves, resulting in the tree $T'$, then
    \[|s(\Pi_T)-s(\Pi_{T'})| \le \frac{\ell(\ell-1)}{n}. \]
\end{lemma}

\begin{proof}
We think of the leaf sets of $T$ and $T'$ as being the same set $X=\leafset(T)=\leafset(T')$, and fix a ranking $\ranking$ of $X$. Consider an ordering of $T$ and the ranking $\ordering$ of $X$ that it induces. Outside of $T'_{v^*}$ we can order $T'$ in the same way as $T$. Irrespective of how we order $T'$ inside $T'_{v^*}$, the induced ranking $\ordering'$ of $X$ agrees with $\ordering$ on all leaves in $X$ except possibly those in $Y \doteq \leafset(T_{v^*})=\leafset(T'_{v^*})$. Moreover, under both $\ordering$ and $\ordering'$, the set $Y$ gets mapped to the same contiguous interval. It follows that for all pairs $(\ell_1,\ell_2)$ of distinct leaves of which at least one lies outside of $Y$, $(\ell_1,\ell_2)$ constitutes an inversion of $\ordering$ with respect to $\ranking$ if and only if $(\ell_1,\ell_2)$ constitutes an inversion of $\ordering'$ with respect to $\ranking$. For any node $v$ outside of $T_{v^*}$, root inversions in $T_v$ cannot involve leaves that are both in $Y \doteq \leafset(T_{v^*})$. See Figure~\ref{fig:lipschitz} for an illustration. Thus, for such nodes $v$, $\RInv(T_v,\ranking,\ordering) = \RInv(T'_v,\ranking,\ordering')$. By taking the minimum over all orderings, we conclude:
\begin{claim}\label{claim:root:insensitive}
$\MRInv(T_v,\ranking) = \MRInv(T'_v,\ranking)$ holds for every node $v$ outside of $T_{v^*}$ (or equivalently, outside of $T'_{v^*}$).
\end{claim}

Consider a ranking $\ranking$ and an adjacent-rank transposition $\transposition=(r,r+1)$. We claim that, unless $(\ell_1,\ell_2) \doteq (\ranking^{-1}(r),\ranking^{-1}(r+1)) \in Y \times Y$, $\Pi_T$ is sensitive to $\transposition$ at $\ranking$ if and only if $\Pi_{T'}$ is sensitive to $\transposition$ at $\ranking$. This is because by Proposition~\ref{prop:change} the only term in the decomposition \eqref{eq:MInv:decomposition} of $\MInv(T,\ranking)$ that can be affected by $\transposition$ is the contribution $\MRInv(T_v,\ranking)$ for $v = \LCA(\ell_1,\ell_2)$. If at least one of $\ell_1$ or $\ell_2$ is not inside $T_{v^*}$, then $v$ is not inside $T_{v^*}$ either, so by Claim~\ref{claim:root:insensitive}, $\MRInv(T_v,\ranking) = \MRInv(T'_v,\ranking)$. By the same token, $\MRInv(T_v,\transposition \ranking) = \MRInv(T'_v,\transposition \ranking)$. It follows that $\MInv(T,\ranking) \ne \MInv(T,\transposition \ranking)$ if and only if $\MInv(T',\ranking) \ne \MInv(T',\transposition \ranking)$.

We bound the expected number of values of $r$ for which $(\ranking^{-1}(r),\ranking^{-1}(r+1)) \in Y \times Y$ with $Y \doteq \leafset(T_v)$ when $\ranking$ is chosen uniformly at random. For $r \in [n-1]$, the probability that $\ranking^{-1}(r) \in Y$ is $\frac{\ell}{n}$, and the probability that $\ranking^{-1}(r+1)\in Y$ given that $\ranking^{-1}(r)\in Y$ is $\frac{\ell-1}{n-1}$. Using linearity of expectation on the indicators, the expected number of said $r$ is 
    \begin{equation*}
        (n-1)\left(\frac{\ell(\ell-1)}{n(n-1)}\right) = \frac{\ell(\ell-1)}{n}.
    \end{equation*}
\end{proof}

Lemma~\ref{lemma:sensitivity:lipschitz} helps to extend query lower bounds based on average sensitivity to larger classes. Suppose we have established a good lower bound on the sensitivity $s(\Pi_T)$ for a class $C$ of trees. Consider a class $C'$ obtained by taking a tree $T$ in class $C$ and replacing some of the subtrees $T_v$ by other subtrees $T'_v$ on the same number of leaves. For this new class $C'$ the same lower bound on the sensitivity of inversion minimization applies modulo the Lipschitz loss. For example, Theorem~\ref{thm:main:binary} holds by virtue of a lower bound of the form $s(\Pi_T) \ge (n-1) - c \log(n)$ for every binary tree $T$ with $n$ leaves, where $c$ is a universal constant. If we allow some of the subtrees of $T$ to be replaced by, say freely arrangeable ones on the same leaves, applying Lemma~\ref{lemma:sensitivity:lipschitz} for each of the modified subtrees in sequence shows that the resulting new tree $T'$ has 
\[ s(\Pi_{T'}) \ge s(\Pi_T) - \frac{\alpha n (\alpha n - 1)}{n} 
\ge (n-1) - c \log(n) - \alpha^2 (n-1) = (1-\alpha^2)(n-1) - c \log(n), \]
where $\alpha$ denotes the fraction of leaves that belong to one of the replaced subtrees.

In fact, the notion of average sensitivity is robust with respect to the following, more refined type of surgery. From any tree $T$, let $R$ be a connected subset of $T$ that includes no leaves. Let $T'$ be the subtree rooted at the LCA of $R$ ($T'$ contains all of $R$), and let $T'_1, \dots, T'_k$ be the disjoint maximal subtrees of $T'$ that are strictly below $R$. Let $R'$ be any tree that has $k$ leaves. Replace $T'$ by $R'$, and then replace the leaves of $R'$ by $T'_1, \dots, T'_k$.

The effect is that the region $R$ has been ``reshaped" to look like $R'$, but the rest of $T$ is unaffected. The cost of such a surgery is at most $(n-1)$ times the probability that a 
uniformly random pair of distinct leaves has their LCA in $R$. The bound follows from thinking of sensitivity as $(n-1)$ times the probability that a uniformly random edge in the full permutahedron is bichromatic. Provided the LCA of 
the affected leaves is outside $R$, then we get sensitivity before the surgery if and only if we get it after the surgery. Surgeries can be iterated, and the costs accumulate additively. In combination with our strong lower bound on the average sensitivity of binary trees (\cref{lemma:sensitivity:binary}), this allows for a robust sense in which ``mostly-binary" trees have high average sensitivity.

\section{Refined Sensitivity Approach for Binary Trees}
\label{sec:sensitivity:binary}

In this section we show how to refine the sensitivity approach for lower bounds on the query complexity of the problem $\Pi_T$ of inversion minimization on trees in the important special case of binary trees $T$. In Section~\ref{sec:binary:criterion} we first develop a criterion for when a particular ranking $\ranking$ is sensitive to a particular adjacent-rank transposition $\transposition$.
We then analyze the root sensitivity of binary trees in Section~\ref{sec:sensitivity:root} and finally establish a strong lower bound on the average sensitivity in Section~\ref{sec:sensitivity:average}. An application of the Sensitivity Lemma then yields Theorem~\ref{thm:main:binary}.

\subsection{Sensitivity criterion}
\label{sec:binary:criterion}

Recall the decomposition of the objective function $\MInv(T,\ranking)$ into contributions attributed to each node $v$ of degree $\deg_T(v)>1$, as given by \eqref{eq:MInv:decomposition} in Section~\ref{sec:decomposition}. In the case of binary trees, the contribution of node $v$ can be calculated simply as
\begin{equation}\label{eq:contribution:binary:retake}
\MRInv(T_v,\ranking) = \min(\XInv_\ranking(L_1,L_2),\XInv_\ranking(L_2,L_1)),
\end{equation}
where $u_1$ and $u_2$ denote the two children of $v$, and $L_1 \doteq \leafset(T_{u_1})$ and $L_2 \doteq \leafset(T_{u_2})$ their leaf sets. This simplicity makes a precise analysis of sensitivity feasible, as we will see next. 

For a given ranking $\ranking$ of $T$ and a given adjacent-rank transposition $\transposition$, we would like to figure out the effect of $\transposition$ on the objective $\MInv(T,\cdot)$, in particular when $\MInv(T,\transposition \ranking) = \MInv(T,\ranking)$. 
Let $\ell_\lo$ and $\ell_\hi$ denote the two leaves that are affected by the transposition $\transposition$ on the ranking $\ranking$, where the subscript ``$\lo$'' indicates the lower of the two leaves with respect to $\ranking$, and ``$\hi$'' the higher of the two. Let $v$ be the lowest common ancestor $\LCA(\ell_\lo,\ell_\hi)$. We use the same subscripts ``$\lo$'' and ``$\hi$'' for the two children of $v$: $u_\lo$ denotes the child whose subtree contains $\ell_\lo$, and $u_\hi$ its sibling. Similarly, we denote by $L_\lo$ the leaf set of $T_{u_\lo}$, and by $L_\hi$ the leaf set of $T_{u_\hi}$. See Figure~\ref{fig:sensitivity:binary} for the subsequent analysis. 

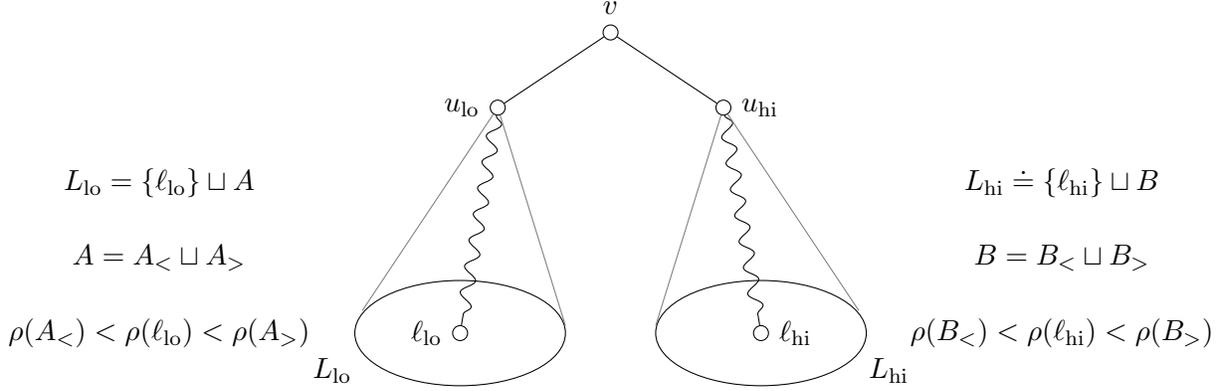
\begin{figure}[ht]
    \centering
    \begin{tikzpicture}
        \node (top) at (0,0)[circle,draw,fill=white,inner sep=2pt,label={above:$v$}]{};
        \node (ulo) at (-1.5,-1)[circle,draw,fill=white,inner sep=2pt,label={left:$u_{\lo}$}]{};
        \node (uhi) at (1.5,-1)[circle,draw,fill=white,inner sep=2pt,label={right:$u_{\hi}$}]{};
        \node (llo) at (-2,-4)[circle,draw,fill=white,inner sep=2pt,label={left:$\ell_{\lo}$}]{};
        \node (lhi) at (2,-4)[circle,draw,fill=white,inner sep=2pt,label={right:$\ell_{\hi}$}]{};
        
        \draw (ulo) -- (top) -- (uhi);
        \draw[decorate, decoration=snake] (ulo) -- (llo);
        \draw[decorate, decoration=snake] (uhi) -- (lhi);
        
        \draw (llo) ellipse (1.4 and 0.7);
        \draw (lhi) ellipse (1.4 and 0.7);
        
        \draw [thin, gray] (ulo)--(-3.3,-3.7);
        \draw [thin, gray] (ulo)--(-0.6,-3.9);
        \draw [thin, gray] (uhi)--(3.3,-3.7);
        \draw [thin, gray] (uhi)--(0.6,-3.9);
        
        \node at (-3.7,-4.5)[]{$L_\lo$};
        \node at (3.7,-4.5)[]{$L_\hi$};
        
        \node at (-6,-2)[]{$L_\lo = \{\ell_\lo\}\sqcup A$};
        \node at (-6,-3)[]{$A = A_< \sqcup A_>$};
        \node at (-6,-4)[]{$\ranking(A_<) < \ranking(\ell_\lo) < \ranking(A_>)$};
        \node at (6,-2)[]{$L_\hi \doteq \{\ell_\hi\}\sqcup B$};
        \node at (6,-3)[]{$B = B_< \sqcup B_>$};
        \node at (6,-4)[]{$\ranking(B_<) < \ranking(\ell_\hi) < \ranking(B_>)$};
    \end{tikzpicture}
    \caption{Sensitivity analysis for binary trees}
    \label{fig:sensitivity:binary}
\end{figure}

By Proposition~\ref{prop:change}, the situation before and after the application of $\transposition$ is as follows, where $x \doteq \XInv_\ranking(L_\lo,L_\hi)$ and $y \doteq \XInv_\ranking(L_\hi,L_\lo)$. 

\begin{equation*}
    \begin{array}{c|cc|c}
      \text{ranking} & \RInv(T_v,\cdot, \ordering(\ell_\lo) < \ordering(\ell_\hi)) &  \RInv(T_v,\cdot, \ordering(\ell_\hi) < \ordering(\ell_\lo)) & \MRInv(T_v,\cdot) \\
      \hline
      \ranking & x & y & \min(x,y) \\
      \transposition \ranking & y-1 & x+1 & \min(y-1,x+1)
    \end{array}
\end{equation*}

The objective function remains the same iff $\min(x,y) = \min(y-1,x+1)$, which happens iff $x-y=-1$, or equivalently iff
\begin{equation}\label{eq:criterion:derivation}
\DInv_\ranking(L_\lo,L_\hi) = \DInv_\ranking(\{\ell_\lo\},\{\ell_\hi\}),
\end{equation}
where we introduce the following short-hand: 
\begin{definition}[cross inversion difference, $\DInv_{\cdot}(\cdot,\cdot)$]
For a ranking $\ranking$ of a set $X$, and two subsets $A, B \subseteq X$,
\[ \DInv_\ranking(A,B) \doteq \XInv_\ranking(A,B) - \XInv_\ranking(B,A). \]
\end{definition}
We can split $L_\lo$ as $L_\lo = \{\ell_\lo\} \sqcup A = A_< \sqcup \{\ell_\lo\} \sqcup A_>$, where $A_<$ contains all leaves in $L_\lo$ that $\ranking$ ranks before $\ell_\lo$, and $A_>$ contains all the leaves in $L_\lo$ that $\ranking$ ranks after $\ell_\lo$. We similarly split $L_\hi$, as indicated in Figure~\ref{fig:sensitivity:binary}. We have that
\[ \DInv_\ranking(L_\lo,L_\hi) = \DInv_\ranking(\{\ell_\lo\},\{\ell_\hi\}) + \DInv_\ranking(\{\ell_\lo\},B) + \DInv_\ranking(A,\{\ell_\hi\}) + \DInv_\ranking(A,B).\]
Since the ranks of $\ell_\lo$ and $\ell_\hi$ under $\ranking$ are adjacent, we have that $\DInv_\ranking(\{\ell_\lo\},B) = |B_<| - |B_>|$ and $\DInv_\ranking(A,\{\ell_\hi\}) = |A_>| - |A_<|$. Plugging everything into \eqref{eq:criterion:derivation} we conclude:

\begin{proposition}\label{prop:sensitivity:criterion}
Let $T$ be a binary tree, $\ranking$ a ranking of the leaves of $T$, $\transposition$ an adjacent-rank transposition, $\ell_\lo$ and $\ell_\hi$ the two leaves affected by $\transposition$ under $\ranking$ such that $\ranking$ ranks $\ell_\lo$ before $\ell_\hi$. Referring to the notation in Figure~\ref{fig:sensitivity:binary}, we have that 
\begin{align}
\MInv(T,\ranking) = \MInv(T, \transposition \ranking)
&\Leftrightarrow \DInv_\ranking(L_\lo,L_\hi) = \DInv_\ranking(\{\ell_\lo\},\{\ell_\hi\}) \label{eq:criterion:binary:simple} \\
&\Leftrightarrow \DInv_\ranking(A,B) = |A_<| - |A_>| + |B_>| - |B_<|.
\label{eq:criterion:binary:simple:refined}
\end{align}
\end{proposition}

\subsection{Root sensitivity}
\label{sec:sensitivity:root}

Given a ranking $\ranking$ and an adjacent-rank transposition $\transposition$, we know by Proposition~\ref{prop:change} that at most one of the terms in the decomposition \eqref{eq:MInv:decomposition} of $\MInv(T,\ranking)$ is affected by the transposition, namely $\MRInv(T_v,\ranking)$ where $v$ is the lowest common ancestor of the affected leaves $\ell_\lo$ and $\ell_\hi$. It follows that we can write the average sensitivity of $\Pi_T \doteq \MInv(T,\cdot)$ as the following convex combination:
\begin{align}\label{eq:sensitivity:combination}
s(\Pi_T) & = (n-1) \cdot \Pr[ \MInv(T,\ranking) \ne \MInv(T,\transposition \ranking)] \nonumber \\
& = (n-1) \sum_v \Pr[ v = \LCA(\ell_\lo,\ell_\hi) ] \cdot \Pr[ \MRInv(T_v,\ranking) \ne \MRInv(T_v,\transposition \ranking) \, | \, 
v = \LCA(\ell_\lo,\ell_\hi)], 
\end{align}
where the probability is over a uniformly random choice of the ranking $\ranking$ and the adjacent-rank transposition $\transposition$, and $\ell_\lo$ and $\ell_\hi$ denote the affected leaves. The conditional probability on the right-hand side of \eqref{eq:sensitivity:combination} only depends on the subtree $T_v$. The ranking $\ranking$ of all leaves induces a ranking $\ranking'$ of the leaves of $T_v$ that is uniform under the conditioning. Similarly, the adjacent-rank transposition $\transposition$ for $\ranking$ induces an adjacent-rank transposition $\transposition'$ for $\ranking'$; the distribution of $\transposition'$ under the conditioning is independent of $\ranking'$ and uniform among all adjacent-rank transpositions such that the affected leaves live in subtrees of different children of $v$. Thus, the probability on the right-hand side of \eqref{eq:sensitivity:combination} coincides with the following notion for the subtree $T_v$.
\begin{definition}[root sensitivity]\label{def:root-sensitivity}
Let $T$ be a tree. The \emph{root sensitivity} of $T$ is the probability that $\MRInv(T_v,\ranking) \ne \MRInv(T_v,\transposition \ranking)$ when $\ranking$ is a uniform random  ranking of $\leafset(T)$, and $\transposition$ a uniform random adjacent transposition with the condition that the affected leaves are in subtrees of different children of the root of $T$.
\end{definition}
Note that the only nodes $v$ that need to be considered in the sum  on the right-hand side of \eqref{eq:sensitivity:combination} are those that can appear as the lowest common ancestor of two leaves, and such that $T_v$ is not freely arrangeable. In the case of binary trees, this means that we only need to consider nodes $v$ of degree 2 such that $T_v$ contains more than 2 leaves. In this section we prove a strong lower bound on the root sensitivity of such trees $T_v$.

Consider the binary tree $T$ with root $v$ in Figure~\ref{fig:sensitivity:binary}. The distribution underlying Definition~\ref{def:root-sensitivity} can be generated as follows: Pick a leaf on each side of the root $v$ uniformly at random, and let $\ranking$ be a ranking of the leaves of $T$ that is uniformly random on the condition that the selected leaves receive adjacent ranks; $\transposition$ then is the adjacent-rank transposition that swaps the two selected leaves. The root sensitivity of $T$ is the complement of the probability that the right-hand side of \eqref{eq:criterion:binary:simple:refined} holds under this distribution. Let us analyze the left-hand side of \eqref{eq:criterion:binary:simple:refined} further. As $A = A_< \sqcup A_>$ and $B = B_< \sqcup B_>$, we have that
\[ \DInv_\ranking(A,B) = \DInv_\ranking(A_<,B_<) + \DInv_\ranking(A_<,B_>) + \DInv_\ranking(A_>,B_<) + \DInv_\ranking(A_>,B_>). \]
By the defining properties of the sets involved (see Figure~\ref{fig:sensitivity:binary}), we know that $\DInv_\ranking(A_<,B_>) = - a_< \, b_>$ and $\DInv_\ranking(A_>,B_<) = a_> \, b_<$, where $a_< \doteq |A_<|$, $a_> \doteq |A_>|$, $b_< \doteq |B_<|$, and $b_> \doteq |B_>|$. Thus, we can rewrite criterion \eqref{eq:criterion:binary:simple:refined} as:
\begin{equation}\label{eq:criterion:expanded}
\DInv_\ranking(A_<,B_<) + \DInv_\ranking(A_>,B_>) = a_< \, b_> - a_> \, b_< + a_< - a_> + b_> - b_<.
\end{equation}
A critical observation that helps us to bound the probability of \eqref{eq:criterion:expanded} is that, conditioned on all four values $a_{\cdot}$ and $b_{\cdot}$, the right-hand side of \eqref{eq:criterion:expanded} is fixed, but the left-hand side still contains a lot of randomness. In fact, under the conditioning stated, the ranking that $\ranking$ induces on $A_< \sqcup B_<$ is still distributed uniformly at random, the same holds for the ranking that $\ranking$ induces on $A_> \sqcup B_>$, and both distributions are independent. This means that, under the same conditioning, the left-hand side of \eqref{eq:criterion:expanded} has the same distribution as the sum $X_{a_<,b_<} + X_{a_>,b_>}$ of two independent random variables of the following type:
\begin{definition}[cross inversion distribution, $\XInvD_{\cdot,\cdot}$]
For nonnegative integers $a$ and $b$, $\XInvD_{a,b}$ denotes the random variable $\XInv(A,B)$ that counts the number of cross inversions from $A$ to $B$, where $A$ is an array of length $a$, $B$ an array of length $b$, and the concatenation $AB$ is a uniformly random permutation of $[a+b]$.
\end{definition}
In Section~\ref{sec:cross-inversions} we establish the following upper bound on the probability that the number of cross inversions takes on any specific value.
\begin{lemma}\label{thm:gaussian:main}
There exists a constant $C$ such that for all integers $a,b\ge 1$ and $0\le k\le ab$, 
\begin{equation}\label{eq:gaussian:main}
\Pr[\XInvD_{a,b}=k]\le \frac{C}{\sqrt{ab(a+b)}}. 
\end{equation}
\end{lemma}
Using Lemma~\ref{thm:gaussian:main} we can establish an upper bound of the same form as the right-hand side of \eqref{eq:gaussian:main} for the probability that \eqref{eq:criterion:expanded} holds: For some constant $C'$
\begin{equation}\label{eq:sum:goal}
\Pr[ X_{a_<,b_<} + X_{a_>,b_>} = a_< \, b_> - a_> \, b_< + a_< - a_> + b_> - b_<] \le \frac{C'}{\sqrt{ab(a+b)}}, 
\end{equation}
where $a \doteq a_< + a_> \ge 1$ and $b \doteq b_< + b_> \ge 1$.
We consider several cases based on the relative sizes of $a_<$ vs $a_>$, and $b_<$ vs $b_>$.
\begin{itemize}
\item[(i)] In case both $a_< \ge \eta a$ and $b_< \ge \eta b$, the bound \eqref{eq:sum:goal} follows from \eqref{eq:gaussian:main} as long as $C' \ge C / \eta^{3/2}$. 
\item[(ii)] By switching the order in (i), the same holds true in case both $a_> \ge \eta a$ and $b_> \ge \eta b$.
\item[(iii)] In case $a_> \le \eta a$ and $b_< \le \eta b$, the left-hand side of \eqref{eq:criterion:expanded} is at most $2 \eta ab$, whereas the right-hand side is at least
\[ (1-\eta)^2 ab - \eta^2 ab + (1-\eta)a - \eta a + (1-\eta)b - \eta b 
= (1-2\eta)(ab+a+b). \]
As long as $2 \eta \le 1-2\eta$, or equivalently, $\eta \le 1/4$, this case cannot occur.
\item[(iv)] By switching the roles of $A$ and $B$ in (iii), the same holds true in case $b_> \le \eta b$ and $a_< \le \eta a$.
\end{itemize}
As long as $\eta \le 1/2$, it holds that either $a_< \ge \eta a$ or $a_> \ge \eta a$, and either $b_< \ge \eta b$ or $b_> \ge \eta b$. 
Distributing the “and” over the “or”, we obtain the four cases we considered, which are therefore exhaustive. We conclude that \eqref{eq:sum:goal} holds for $C' = 4^{3/2} C$ whenever $a \doteq |A| \ge 1$ and $b \doteq |B| \ge 1$. 

In the case where $a = 0$ and $b \ge 1$, the right-hand side of
\eqref{eq:criterion:binary:simple:refined} vanishes, as do $|A_<|$ and $|A_>|$, so \eqref{eq:criterion:binary:simple:refined} holds if and only if $|B_<| = |B_>|$, or equivalently, the leaf $\ell_\hi$ is ranked exactly in the middle of the leaf set $L_\hi$. As the ranking $\ranking$ is chosen uniformly at random, this happens with probability $1/(b+1)$ where $b \doteq |B| = |L_\hi|-1$. The case where $a \ge 1$ and $b=0$ is symmetric. The remaining case, $a=b=0$, is one we do not need to consider as the tree $T$ then only has two leaves. In all other cases we obtain a strong upper bound on the probability that \eqref{eq:criterion:binary:simple:refined} holds, and by complementation a strong lower bound on the root sensitivity. We capture the lower bound in the following single expression that holds for all cases under consideration.
\begin{lemma}\label{lemma:sensitivity:root}
There exists a constant $c$ such that for every binary tree $T$ with at leaves 3 leaves and a root of degree 2, the root sensitivity of $T$ is at least 
\begin{equation}\label{eq:root:bound}
1 - \frac{c}{\sqrt{n_1 n_2 (n_1+n_2)}},
\end{equation}
where $n_1$ and $n_2$ denote the number of leaves in the subtrees rooted by the two children of the root.
\end{lemma}

\subsection{Average sensitivity}
\label{sec:sensitivity:average}

We are now ready to establish that, except for trivial cases, the average sensitivity of a binary tree is close to maximal. The trivial cases are those where the tree has at most two leaves, in which case the sensitivity is zero. 
\begin{lemma}
\label{lemma:sensitivity:binary}
The average sensitivity of $\Pi_T$ for binary trees $T$ with $n \ge 2$ leaves is $n-O(1)$.
\end{lemma}
\begin{proof}
We use the expression \eqref{eq:sensitivity:combination} for the average sensitivity of $\Pi_T$, where $v$ ranges over all nodes of degree 2 such that $T_v$ contains as least two leaves.
Consider a node $v$ of degree 2 such that $T_v$ contains $n_{v,1}$ leaves one one side and $n_{v,2}$ leaves on the other side, where $n_{v,1}+n_{v,2} \ge 3$. If we choose the ranking $\ranking$ and the adjacent-rank transposition $\transposition$ uniformly at random, each of the $\binom{n}{2}$ pairs of leaves are equally likely to be the affected pair. As there are $n_{v,1} \cdot n_{v,2}$ choices that result in $v$ as their lowest common ancestor, we have that
$\Pr[ v = \LCA(\ell_\lo,\ell_\hi) ] = \frac{2 n_{v,1} n_{v,2}}{n(n-1)}$. Combining this with the root sensitivity lower bound given by \eqref{eq:root:bound}, we have that
\begin{align*} 
s(\Pi_T) & \ge (n-1) \sum_v \frac{2 n_{v,1} n_{v,2}}{n(n-1)} \cdot \left( 1 - \frac{c}{\sqrt{n_{v,1} n_{v,2} (n_{v,1}+n_{v,2})}} \right)  \\
& = (n-1) - \frac{2c}{n} \sum_v \sqrt{\frac{n_{v,1} n_{v,2}}{n_{v,1}+n_{v,2}}}. 
\end{align*}
The following claim then completes the proof.
\end{proof}

\begin{claim}
    \label{claim:sensitivity:sum}
    There is a constant $c'$ such that for all binary trees $T$ with $n$ leaves 
    \begin{equation}
      \sum_v \sqrt{\frac{n_{v,1}n_{v,2}}{n_{v,1}+n_{v,2}}} \le c'n,
    \end{equation}
    where the sum ranges over all nodes $v$ of degree 2 such that $T_v$ contains at least 3 leaves.
  \end{claim}
\begin{proof}[Proof of Claim~\ref{claim:sensitivity:sum}]
We use structural induction to prove a somewhat stronger claim, namely that
\begin{equation}\label{CS.tree-mininvs.sens-thm-conj.claim.eq-induction}
  \sum_v \sqrt{\frac{n_{v,1}n_{v,2}}{n_{v,1}+n_{v,2}}} \le c' n - d' \sqrt{n}
\end{equation}
for some constants $c'$ and $d'$ to be determined.
As the base case we consider binary trees $T$ with at most two leaves. In this case, the left-hand side of \eqref{CS.tree-mininvs.sens-thm-conj.claim.eq-induction} is zero and the right-hand side is non-negative provided $c' \ge d'$, so \eqref{CS.tree-mininvs.sens-thm-conj.claim.eq-induction} holds. 

For the inductive step, the case where the root of $T$ has degree 1 immediately follows from the inductive hypothesis for the subtree $T_u$ rooted by the child $u$ of the root of $T$. The remaining case is where the root of $T$ has degree 2. Let $u_1$ and $u_2$ be the two children for the root, $n_1 = \lvert\leafset(T_{u_1})\rvert$, and 
$n_2 = \lvert\leafset(T_{u_2})\rvert$. The sum on the left-hand side of \eqref{CS.tree-mininvs.sens-thm-conj.claim.eq-induction} has three contributions: $\sqrt{\frac{n_1 n_2}{n_1+n_2}}$ from the root, and the contributions from $T_{u_1}$ and $T_{u_2}$, to which we can individually apply the inductive hypothesis. This gives us an upper bound of
\[ \sqrt{\frac{n_1 n_2}{n_1+n_2}} + (c' n_1 - d' \sqrt{n_1}) + (c' n_2 - d' \sqrt{n_2})
=  \sqrt{\frac{n_1 n_2}{n_1+n_2}} + c' (n_1 + n_2) - d' (\sqrt{n_1} + \sqrt{n_2}), \]
which we want to upper bound by
\[ c' n - d' \sqrt{n} = c' (n_1 + n_2) - d' \sqrt{n_1+n_2}. \]
Writing $n_1 = \alpha n$ for some $\alpha \in [0,1]$ and rearranging terms, the upper bound holds if and only if
\[ \sqrt{\alpha(1-\alpha)} \le d' (\sqrt{\alpha} + \sqrt{1-\alpha} -1). \]
We claim that the upper bound holds for $d'= (\sqrt{2}+1)/2$. Let 
\[F(\alpha) \doteq d'(\sqrt{\alpha}+\sqrt{1-\alpha}-1)-\sqrt{\alpha(1-\alpha)}.\]
It suffices to show that $F(\alpha)\ge 0$. Since $F$ is continuous on $[0,1]$, it attains a minimum on $[0,1]$. On $(0,1)$, $F$ is differentiable. It can be verified that $F'$ has a unique zero in $(0,1/2)$, which needs to be a maximum as $F$ is increasing at $\alpha=0$. By the symmetry $F(\alpha)=F(1-\alpha)$, it follows that the minimum of $F$ on $[0,1]$ is attained at the midpoint $\alpha=1/2$ or at one of the endpoint $\alpha=0$ or $\alpha=1$. At all three points $F(\alpha)=0$. 
We conclude that \eqref{CS.tree-mininvs.sens-thm-conj.claim.eq-induction} holds for any constants $d' \ge (\sqrt{2}+1)/2$ and $c' \ge d'$.

\end{proof}

\section{Sensitivity Approach for Bounded Error}
\label{sec:random}

In this section, we apply the sensitivity approach to obtain lower bounds on the query complexity of problems in the comparison-query model against randomized algorithms with bounded error. We derive a generic result that query lower bounds against deterministic algorithms that are based on the Sensitivity Lemma, also hold against bounded-error randomized algorithms with a small loss in strength. The approach works particularly well when we have linear lower bounds on the average sensitivity, in which case there is only a constant-factor loss in the strength of the query lower bound. Among others, this applies to the $\Omega(n \log n)$ query lower bound for inversion minimization on trees of bounded degree.

\paragraph{Generic lower bound.}

Our approach is based on Yao's minimax principle \cite{Yao77}, which lower bounds worst-case complexity against randomized algorithms with bounded error by average-case complexity against deterministic algorithms with bounded distributional error. We view a deterministic algorithm with small distributional error for a problem $\Pi$ as an exact deterministic algorithm for a slightly modified problem $\Pi'$. The idea is to then apply the sensitivity approach to $\Pi'$, and capitalize on the closeness of the average sensitivities of $\Pi$ and $\Pi'$ to obtain a lower bound in terms of the sensitivity of $\Pi$. By using the Sensitivity Lemma as a black-box, the approach yields a lower bound on the query complexity of bounded-error algorithms that is worst-case with respect to the input and with respect to the randomness, i.e., the lower bound holds for some input and some computation path on that input. By delving into the proof of the Sensitivity Lemma, we are able to obtain a lower bound that is worst-case with respect to the input but average-case with respect to the randomness, i.e., the lower bound holds for the expected number of queries on some input.\footnote{In fact, the approach yields a lower bound that is average-case with respect to the input (chosen uniformly at random) as well as the randomness. This follows because the proof of Yao's minimax principle allows us to replace the left-hand side of \eqref{eq:yao} by the average of the expected number of queries with respect to the distribution $\dist$, which we pick to be uniform in our application of the principle.}

We first define the notions of randomized complexity and distributional complexity.

\begin{definition}[randomized query complexity, $\RQ_{\cdot}(\cdot)$, and distributional query complexity, $\DistQ_\cdot(\cdot, \cdot)$]
Let $\Pi$ be a problem in the comparison-query model and $\eps \in [0,1]$.

A randomized algorithm $R$ for $\Pi$ is said to have error $\eps$ if on every input $\ranking$, the algorithm outputs $\Pi(\ranking)$ with probability at least $1-\eps$. The query complexity of $R$ is the maximum, over all inputs $\ranking$, of the expected number of queries that $R$ makes on input $\ranking$. The \emph{$\eps$-error randomized query complexity} of $\Pi$, denoted $\RQ_{\eps}(\Pi)$, is the minimum query complexity of $R$ over all $\eps$-error randomized algorithms $R$ for $\Pi$.

Let $\dist$ be a probability distribution on the inputs $\ranking$. A deterministic algorithm $A$ for $\Pi$ has error $\eps$ with respect to $\dist$ if the probability that $A(\ranking)=\Pi(\ranking)$ is at least $1-\eps$ where the input $\ranking$ is chosen according to $\dist$. The query complexity of $A$ with respect to $\dist$ is the expected number of queries that $A$ makes on input $\ranking$ when $\ranking$ is chosen according to $\dist$. The \emph{$\eps$-error distributional query complexity} of $\Pi$ with respect to $\dist$, denoted $\DistQ_{\eps}(\Pi,\dist)$, is the minimum query complexity of $A$ with respect to $\dist$ over all deterministic algorithms $A$ for $\Pi$ that have error $\eps$ with respect to $\dist$.
\end{definition}

The relationship between randomized complexity and distributional complexity is described by Yao's principle. 

\begin{lemma}[Yao's minimax principle \cite{Yao77}]
	\label{lemma:yaos-principle}
	Let $\Pi$ be a problem in the comparison-query model, $\eps \in [0,1/2]$, and $\dist$ a distribution on the inputs $\ranking$.
	\begin{equation}\label{eq:yao}
	\RQ_{\eps}(\Pi) \ge \frac{1}{2}\DistQ_{2\eps}(\Pi, \dist). 
	\end{equation}
\end{lemma}

We now prove lower bounds on the distributional query complexity, and thus on randomized query complexity, of comparison-query problems $\Pi$ based on average sensitivity bounds. For these bounds, we always set $\dist$ to be the uniform distribution, the distribution underlying the notion of average sensitivity. 

We start by studying average-case query complexity, i.e., zero-error distributional query complexity, and its relationship to the average sensitivity. We follow a strategy similar to the one in the proof of the Sensitivity Lemma. Whereas a bound on deterministic complexity $\DQ$ follows purely from the number of execution traces $\traces$, here, the execution traces are weighted by their depth and their probability of occurring.

Recall that $g$ in the statement of the Strong Sensitivity Lemma denotes any convex function $g: [1,\infty) \to \mathbb{R}$ with $g(x) = x!$ for $x \in [n]$; for deriving the Sensitivity Lemma from the Strong Sensitivity Lemma we also need $g$ to be nondecreasing. One such function is $g(x) = \Gamma(x+1)$. To prove a lower bound on the zero-error distributional complexity, we need the function $g$ to be not only convex, but \emph{log-convex}, i.e., $\log_2 g(x)$ needs to be convex. The function $g(x) = \Gamma(x+1)$ satisfies this constraint as well.

\begin{proposition}
	\label{claim:sensitivity:average}
Let $\Pi$ be a problem in the comparison-query model with $n$ items, $\dist$ the uniform distribution on the inputs $\ranking$, and $g: [1,\infty) \to \mathbb{R}$ a nondecreasing log-convex function with $g(x) = x!$ for $x \in [n]$.
\[\DistQ_0(\Pi,\dist)\ge \log_2 (g(s(\Pi)+1)/n) \]
\end{proposition}

\begin{proof}
	Let $k$ be the number of distinct execution traces of a deterministic algorithm $A$ for $\Pi$, and let $R_1,\dots,R_k$ denote the corresponding sets of rankings. Interpreting $A$ as a binary decision tree, let $\depth(R_i)$ be the depth of the execution trace corresponding to $R_i$. By Kraft's inequality,
	\[\sum_{i=1}^{k}2^{-\depth(R_i)}\le 1. \]
	Let $f(x)=\frac{1}{n}\frac{g(x+1)}{n!}$ and define the weight function $w(\ranking)=f(\deg_{\overline{H}}(\ranking))$, where $\overline{H}$ refers to the notation of the Strong Sensitivity Lemma: $\overline{H}$ denotes the subgraph of the full permutahedron that only consists of the bichromatic edges when the vertices are colored with their execution trace under $A$. By Claim~\ref{claim:encoding}, the sum of the weights of all  rankings $\ranking$ in $R_i$ is at most $1$. Therefore,
	\[\sum_{\ranking} 2^{-\depth(\ranking)} w(\ranking) \le 1. \]
	Dividing both sides by $n!$ and taking the logarithm of both sides, we get that
	\begin{equation}
		\label{eq:average-case-weight}
		\log_2 \Expect\left[2^{-\depth(\ranking)}w(\ranking)\right] \le \log_2(1/n!),
	\end{equation}
	where the expectation is with respect to a uniform distribution over the inputs $\ranking$.
	By Jensen's inequality, since $\log$ is concave, we get 
	\[\log_2 \Expect\left[2^{-\depth(\ranking)}w(\ranking)\right] \ge \Expect\left[\log_2 \left(2^{-\depth(\ranking)}w(\ranking) \right)\right]= \Expect[-\depth(\ranking)]+\Expect[\log_2 w(\ranking)], \]
	which, in combination with (\ref{eq:average-case-weight}), implies
	\[\Expect[\log_2 w(\ranking)] \le \Expect[\depth(\ranking)]+\log_2(1/n!). \]
	Note that since $g$ is log-convex, so is $f$. By applying Jensen's inequality again,  
	\[\Expect[\log_2 w(\ranking)]=\Expect[\log_2 f(\deg_{\overline{H}}(\ranking))]\ge \log_2 f(\Expect[\deg_{\overline{H}}(\ranking)]),\]
	implying 
	\[ \log_2\left(\frac{1}{n}\cdot \frac{g(\Expect[\deg_{\overline{H}}(\ranking)]+1)}{n!}\right) \le \Expect[\depth(\ranking)]+\log_2(1/n!), \]
	or equivalently,
	\[ \Expect[\depth(\ranking)] \ge \log_2(g(\Expect[\deg_{\overline{H}}(\ranking)]+1)/n). \]
	The result follows since $A$ is an arbitrary deterministic algorithm for $\Pi$, $\Expect[\depth(\ranking)]$ equals the query complexity of $A$ with respect to the uniform distribution $\dist$, $\Expect[\deg_{\overline{H}}(\ranking)] \ge \Expect[\deg_{\overline{G}(\Pi)}(\ranking)] = s(\Pi)$, and $g$ is nondecreasing. 
\end{proof}

Proposition~\ref{claim:sensitivity:average} allows us to prove a lower bound on the $\eps$-error distributional query complexity of $\Pi$ with respect to the uniform distribution. In order to do so, we view a deterministic algorithm with distributional error $\eps$ for $\Pi$ as an exact deterministic algorithm for a modified problem $\Pi'$, apply Proposition~\ref{claim:sensitivity:average}, and lower bound the sensitivity of $\Pi'$ in terms of the sensitivity of $\Pi$.

\begin{proposition}
	\label{claim:sensitivity:partial}
	Let $\Pi$ be a problem in the comparison-query model with $n$ items, $\dist$ the uniform distribution on the inputs $\ranking$, $\eps\in [0,1]$, and $g: [1,\infty) \to \mathbb{R}$ a nondecreasing log-convex function with $g(x) = x!$ for $x \in [n]$.
	\begin{equation}\label{eq:prop:sensitivity:randomized}
	\DistQ_{\eps}(\Pi,\dist)\ge \log_2 \left(\frac{g(s(\Pi)+1-2(n-1)\eps)}{n}\right)
	\end{equation}
\end{proposition}

\begin{proof}
	Consider any algorithm $A$ with error $\eps$ for $\Pi$, or in other words, $\Pr[A(\ranking)\neq\Pi(\ranking)] \le \eps$. Let $\Pi_A$ be the problem of determining the output of $A$. We prove that 
	\[ s(\Pi_A)\ge s(\Pi)-2(n-1)\eps,\]
	which implies the desired result by Proposition~\ref{claim:sensitivity:average}, since $A$ is a deterministic algorithm for $\Pi_A$ and $g$ is nondecreasing.
	
	Let $G$ denote the full permutahedron graph for $n$ items. We use the fact that $s(\Pi)=(n-1)\cdot \Pr_{e\in G}[e \in G(\Pi)]$, and similarly, $s(\Pi_A)=(n-1)\cdot \Pr_{e\in G}[e \in G(\Pi_A)]$, where all the underlying distributions are uniform. Suppose the endpoints of $e$ are $\ranking_1$ and $\ranking_2$. Note that if $e \in G$ is picked uniformly at random, then the marginal distributions of both $\ranking_1$ and $\ranking_2$ are also uniform. If $A(\ranking_1)=\Pi(\ranking_1)$, $A(\ranking_2)=\Pi(\ranking_2)$, and $e\in G(\Pi_A)$, then $e\in G(\Pi)$, as well. By a union bound, the probability that $A(\ranking_1) \neq \Pi(\ranking_1)$ or $A(\ranking_2) \neq \Pi(\ranking_2)$ is at most $2\eps$. 
	\begin{align*}
		\Pr_{e\in G}[e\in G(\Pi_A)] &\ge \Pr_{e\in G}[e\in G(\Pi)] - \Pr_{e\in G}[A(\ranking_1)\neq \Pi(\ranking_1)\text{ or }A(\ranking_2)\neq \Pi(\ranking_2)]\\
		&\ge \Pr_{e\in G}[e\in G(\Pi)] - 2\eps.
	\end{align*}
	Multiplying both sides by $n-1$ gives $s(\Pi_A)\ge s(\Pi)-2(n-1)\eps$.
\end{proof}

Since $s(\Pi) \le n-1$, Proposition~\ref{claim:sensitivity:partial} only yields nontrivial lower bounds for small $\eps$. In order to establish lower bounds for the standard $\eps=1/3$, we first reduce the error using standard techniques. Doing so such that the argument of $g$ on the right-hand side of \eqref{eq:prop:sensitivity:randomized} remains $\Omega(s(\Pi))$, and picking $g(x)=\Gamma(x+1)$, we conclude:

\begin{lemma}[Bounded-Error Sensitivity Lemma]
\label{lemma:randomized-sensitivity}
For any problem $\Pi$ in the comparison-query model with $n$ items, 
\[\RQ_{1/3}(\Pi) = \Omega \left(\frac{s \log s}{\log(2n/s)} \right), \]
where $s \doteq s(\Pi)$.
\end{lemma}

\begin{proof}
By taking the majority vote of multiple independent runs and a standard analysis, e.g, based on Chernoff bounds, we have that $\RQ_\eps(\Pi) = O(\log(1/\eps)) \RQ_{1/3}(\Pi)$ for any $\eps \le 1/3$. Combining this with Lemma~\ref{lemma:yaos-principle} and Proposition~\ref{claim:sensitivity:partial}, we have:
\[\RQ_{1/3}(\Pi) = \Omega \left( \frac{\RQ_{\eps}(\Pi)}{\log(1/\eps)} \right) = \Omega \left(  \frac{\DistQ_{2\eps}(\Pi)}{\log(1/\eps)} \right) =  \Omega \left( \frac{\log(g(s+1-4(n-1)\eps)/n)}{\log(1/\eps)}\right). \]
Setting $\eps$ such that $4n\eps = s/2$ yields
\[\RQ_{1/3}(\Pi) =  \Omega \left( \frac{\log(g(s/2+1)/n)}{\log(8n/s)}\right).
\]
Picking $g(x)=\Gamma(x+1)$ and using the fact that $\Gamma(x) \ge \sqrt{2\pi x} \left(\frac{x}{e}\right)^x$, we obtain
\[\RQ_{1/3}(\Pi) =  \Omega \left( \frac{(s/2)\log(s/(2e)) - \log n}{\log(8n/s)}\right) = \Omega \left( \frac{s\log s }{\log(2n/s)}\right),
\]
where the simplification can be verified by considering the cases of large $s$ (say $s \ge \sqrt{n}$) and small $s$ separately. 
\end{proof}

We can apply Lemma~\ref{lemma:randomized-sensitivity} to the sensitivity lower bounds of Lemma~\ref{lemma:k-ary:subtree-sensitivity} and produce randomized lower bounds for inversion minimization on bounded-degree trees. Using \cref{fact:k-ary:balance} we obtain: 

\begin{theorem}[lower bound against bounded-error 
for inversion minimization on trees]
	Let $T$ be a tree with $\deg(T)\le k$. The query complexity of $\Pi_T$ for bounded-error randomized algorithms is $\Omega(\frac{n\log(n/k)}{k\log(k)})$.
\end{theorem}
\section{Connectivity Lemma}
\label{sec:connectivity}

In this section we establish Lemma~\ref{lemma:connectivity} and use it to present some of the known lower bounds in a unified framework. We actually prove the following somewhat stronger result.

\begin{lemma}[Strong Connectivity Lemma]\label{lemma:connectivity:strong}
Consider an algorithm $A$ in the comparison-based model, color each vertex of the permutahedron with its execution trace under $A$, and let $H$ denote the subgraph with the same vertex set but only containing the monochromatic edges. The number of distinct execution traces of $A$ equals the number of connected components of $H$.
\end{lemma}
The Connectivity Lemma follows from Lemma~\ref{lemma:connectivity:strong} because the coloring with execution traces of an algorithm $A$ for $\Pi$ is a refinement of the coloring with $\Pi$. Note that the counterpart of Lemma~\ref{lemma:connectivity:strong} in the Boolean setting is trivial. This is because an execution trace in the Boolean setting is specified by values for a subset of the input bits, so the set of inputs that follow a particular execution trace form a subcube of the hypercube, the Boolean counterpart of the permutahedron. Subcubes are trivially connected inside the hypercube. In the comparison-query model, the sets of inputs that follow a particular execution trace can be more complicated, and their connectedness is no longer trivial but still holds.

\begin{proof}[Proof of Lemma~\ref{lemma:connectivity:strong}]
Two rankings $\ranking_1$ and $\ranking_2$ that have distinct execution traces under $A$ cannot be connected because any path between them needs to contain at least one bichromatic edge. For the remainder of the proof, we consider two rankings $\ranking_1$ and $\ranking_2$ that have the same execution trace under $A$, and construct a path from $\ranking_1$ to $\ranking_2$ in $H$.

If $\ranking_1 = \ranking_2$, we do not need to make any move and use an empty path. 

Otherwise, there exists a rank $r<n$ such that $\ranking_1$ and $\ranking_2$ agree on ranks less than $r$ and disagree on rank $r$. We have the following situation, where the item $y_r$ with rank $r$ under $\ranking_2$, has rank $s>r$ under $\ranking_1$. 

\begin{equation*}
    \begin{array}{r|ccccccccc}
      \text{rank} & 1 & \cdots &  r-1 &   r & \cdots & s-1 & s & \cdots &   n \\
      \hline
      \ranking_1^{-1} & x_1 & \cdots & x_{r-1} &   x_r  & \cdots & x_{s-1} & x_s = y_r & \cdots &     \\
       & {\scriptstyle =}  & {\scriptstyle =}  & {\scriptstyle =}  &  {\scriptstyle \ne}  &    &    &  &   &    \\
      \ranking_2^{-1} & y_1 & \cdots & y_{r-1} & y_r & \cdots &     &     & \cdots &    
    \end{array}
\end{equation*}

Considering ranking $\ranking_1$, we have that $\ranking_1(x_{s-1}) = s-1 < s = \ranking_1(x_s)$. Considering ranking $\ranking_2$, since $x_{s-1}$ differs from $y_i=x_i$ for every $i \in [r-1]$ and also differs from $y_r$, we have that $\ranking_2(x_{s-1}) > r = \ranking_2(y_r) = \ranking_2(x_s)$. Thus, the relative ranks of $x_{s-1}$ and $x_s$ under $\ranking_1$ and $\ranking_2$ differ. As $\ranking_1$ and $\ranking_2$ have the same execution trace, this means that the algorithm does not compare $x_{s-1}$ and $x_s$ on either input, and on $\ranking_1$ in particular. Let $\ranking'_1$ be the ranking obtained from ranking $\ranking_1$ by applying the adjacent-rank transposition $\transposition = (s-1,s)$. Since the algorithm does not compare the affected items, the execution trace for $\ranking'_1$ and $\ranking_1$ are the same, so the edge from $\ranking'_1$ to $\ranking_1$ is monochromatic and in $H$. We use this edge as the first on the path from $\ranking_1$ to $\ranking_2$ in $H$. What remains is to find a path from $\ranking'_1$ to $\ranking_2$ in $H$. The situation is the same as the one depicted above but with $r$ increased by one in case $s=r+1$, and with the same $r$ and $s$ decreased by one, otherwise. The proof then follows by induction on the ordered pair $(r,n-s)$. 
\end{proof}

\begin{remark}\label{remark:connectivity}
Suppose we allow an algorithm $A$ to have multiple valid execution traces on a given input $\ranking$, and let $R$ denote the set of rankings on which a particular execution trace is valid. The construction in the proof of Lemma~\ref{lemma:connectivity:strong} yields a path in the permutahedron between any two rankings in $R$ such that the path entirely stays within $R$. This means that we can replace $\traces(\Pi)$ in the statement of the Connectivity Lemma by its nondeterministic variant $\NQ(\Pi)$.
\end{remark}

The Connectivity Lemma captures all the prior lower bounds stated in Section~\ref{sec:model} except the elementary adversary argument (which is also based on connectivity considerations, but in an undirected graph other than $H$, namely $(X,E)$ where $E$ denotes the queries the algorithm makes on a given input ranking $\ranking$). It captures the generic information-theoretic lower bound because input rankings with different outputs cannot belong to the same connected component of $H$. We already explained in Section~\ref{sec:overview} how the Connectivity Lemma shows that counting inversions and inversion parity amount to sorting, and require at least $\log(n!)$ queries. We now illustrate its use for a classical problem that is easier than sorting, namely median finding.

Let $\Pi$ denote the selection problem with rank $r = \ceil{n/2}$. For any ranking, the adjacent-rank transpositions $\transposition$ that change the item with rank $r$ are the two that involve rank $r$: $\transposition = (r-1,r)$ and $\transposition = (r,r+1)$. Those transpositions are the ones that correspond to missing edges in the permutahedron graph $G(\Pi)$. As a result, for any two rankings, there exists a path between them in $G(\Pi)$ if and only if they have the same median as well as the same set of items with rank less than $r$ (and also the same set of items with rank greater than $r$). As there are $n$ possibilities for the median and, for each median, $\binom{n-1}{r-1}$ possibilities for the set of items that have rank less than $r$, $G(\Pi)$ has $n \cdot \binom{n-1}{r-1}$ connected components. It follows that any algorithm for $\Pi$ has at least $n \cdot \binom{n-1}{r-1} = \Omega(\sqrt{n}\cdot 2^n)$ distinct execution paths, and therefore needs to make at least $n + \frac{1}{2} \log(n) - O(1)$ queries. 

As a side note, this example clarifies a subtlety in the equivalence between ordinary selection and the instantiation of partial order production that is considered equivalent to selection. Whereas selection of rank $r$ ordinarily requires outputting only the item
of rank $r$, the instantiation of partial order production additionally requires partitioning the remaining items according to whether their ranks are less than or greater than $r$. The above analysis implies that it is impossible for the algorithm to know the item of rank $r$ without also knowing how to partition the remaining items into those of rank less than and greater than $r$. It follows that, in the comparison-based model, ordinary selection and the instantiation of partial order production are equivalent. 

\section{Connectivity Approach}
\label{sec:connectivity:both}

This section covers the connectivity approach for obtaining query lower bounds in the comparison-query model. Our main focus is the problem $\Pi_T$ of inversion minimization on a fixed tree $T$, for which we derive very strong query lower bounds in the case of the special types of trees in Theorem~\ref{thm:main:special}. Some parts of the analysis carry through for a broader class of problems $\Pi$, namely those that satisfy a certain partition property. We first develop the property and apply the Connectivity Lemma to a generic problem $\Pi$ with the property. We then present sufficient conditions for the problem $\Pi_T$ to have the property and perform a detailed analysis, leading to Theorem~\ref{thm:main:special}. Finally, we apply the same ideas to the problem of counting cross inversions, for which we obtain the query lower bound of Theorem~\ref{thm:Mann-Whitney}, as well as to the closely related problem of inversion minimization on the Mann--Whitney trees of Figure~\ref{fig:Mann-Whitney}.

\subsection{Partition property}

In order to obtain good lower bounds on $\traces(\Pi)$ using the Connectivity Lemma, it is sufficient to find good upper bounds on the size of the connected components of a typical vertex in $G(\Pi)$. For the problem $\Pi_T$, we can assume without loss of generality that $T$ has no internal nodes of degree 1, i.e., no nodes with exactly one child. With that assumption, $\Pi_T$ is insensitive to any adjacent-rank transposition $\transposition$ at a ranking $\ranking$ for which the affected leaves are siblings in $T$. Thus, the corresponding edges from the permutahedron are always present in $G(\Pi_T)$. From the perspective of ensuring small connected components in $G(\Pi_T)$, the ideal situation would be if there were no other edges in $G(\Pi_T)$. That is to say, $\Pi_T$ is \emph{sensitive} at $\ranking$ to \emph{every} adjacent-rank transposition $\transposition$ \emph{except} when the affected leaves are siblings. We will investigate conditions on $T$ that guarantee this situation in the next two subsections. In this subsection, we analyze the size of the connected components of $G(\Pi_T)$ when $T$ is of the desired type, and use it obtain a query lower bounds via the Connectivity Lemma. Our analysis applies more generally to any problem $\Pi$ with the following property.
\begin{definition}[partition property]\label{def:partition-property}
A computational problem $\Pi$ in the comparison-query model on a set $X$ of $n$ items has the \emph{partition property} if the set $X$ can be partitioned into sets $X_i$ such that for any ranking $\ranking$ of $X$ and adjacent-rank transposition $\transposition = (r, r+1)$ with $r \in [n-1]$, $\Pi(\ranking) = \Pi(\transposition \ranking)$ if and only if $\ranking^{-1}(r)$ and $\ranking^{-1}(r+1)$ belong to the same partition class $X_i$. If every partition class $X_i$ has size at most $k$, we say that $\Pi$ has the partition property with class size at most $k$.
\end{definition}
In other words, a problem $\Pi$ has the partition property if the underlying universe can be partitioned in such a way that 
adjacent-rank transpositions that do not change the answer are exactly those whose affected items fall within the same partition class. In the case of the problem $\Pi_T$, the partition classes $X_i$ correspond to the leaf child sets $\leafchildset(v)$ from Definition~\ref{def:leaf-child-set}, where $v$ ranges over the leaf parents.

Let us investigate the size of the connected components of $G(\Pi)$ when $\Pi$ satisfies the partition property. Consider a walk in $G(\Pi)$. As the only steps we can take correspond to adjacent-rank transpositions $\transposition$ that swap elements in the same partition class, the sets $X_i$ remain invariant, irrespective of the ranking $\ranking$ we start from. Depending on $\ranking$, there may be more structure inside each partition class $X_i$; the set $X_i$ may be broken up into smaller subsets that are each invariant. For our analysis, we list the elements of each partition class in order of increasing rank under $\ranking$, and include an edge between elements that have successive ranks. We introduce the term ``successor graph" to capture this structure, viewed as a graph with the ranks as vertices.
\begin{definition}[successor graph, $S(\cdot,\cdot)$]
Let $\Pi$ be a computational problem in the comparison-query model on a set $X$ of $n$ items, and $\ranking$ a ranking of $X$. The successor graph of $\Pi$ on $\ranking$, denoted $S(\Pi,\ranking)$, has vertex set $[n]$ and contains all edges of the form $(r,r+1)$ with $r \in [n-1]$ such that $\Pi(\ranking) = \Pi(\transposition \ranking)$, where $\transposition$ denotes the adjacent-rank transposition $(r,r+1)$.
\end{definition}

We have the following connection.
\begin{proposition}\label{prop:successor:components}
Let $\Pi$ be a computational problem in the comparison-query model on the set $X$, and let $\ranking$ be a ranking of $X$. If $\Pi$ has the partition property, then the connected component of $\ranking$ in $G(\Pi)$ has size $\prod_j (n_j!)$, where the $n_j$'s denote the sizes of the connected components of $S(\Pi,\ranking)$.
\end{proposition}
\begin{proof}
The connected components of the successor graph $S(\Pi,\ranking)$ correspond to subsets of the classes $X_i$ that each remain invariant under walks in $G(\Pi)$. Within each of the subsets, independently for each subset, every possible ordering can be realized by such walks. This is because for any adjacent-rank transposition $\transposition$, the successor graphs $S(\Pi,\ranking)$ and $S(\Pi,\transposition \ranking)$ are the same, and every ordering can be realized by a sequence of swaps of adjacent elements. It follows that the number of rankings that can be reached from $\ranking$ in $G(\Pi)$ equals the product over all connected components of $S(\Pi,\ranking)$ of the number of possible orderings of the elements in the connected component. 
\end{proof}

Figure~\ref{fig:components} depicts an example for a problem of type $\Pi_T$ and a partition consisting of 4 classes, namely the leaf child sets $\leafchildset_1, \leafchildset_2, \leafchildset_3$ and $\leafchildset_4$. The tree $T$ and ranking $\ranking$ are represented in Figure~\ref{fig:sibling:leaves}. Figure~\ref{fig:invariant:sets} represents the part of the successor graph $S(\Pi_T,\ranking)$ involving the leaf child set $\leafchildset_3$ and illustrates the subpartitioning into invariant subsets.

\begin{figure}[ht]
\centering
\begin{subfigure}[c]{0.45\textwidth}
\centering
\begin{tikzpicture}
[
scale = 0.95,
every node/.style = {circle, inner sep=2pt},
level 2/.style = {level distance = 1cm}
]
\def\dist{0.06cm};
\def\dista{0.01cm};
\node [draw]{} [level distance=0.75cm]
    child
    {node [draw, xshift=-0.2cm]{} [sibling distance=0.7cm]
        child
        {node (n4) [draw, label={[label distance = {\dist}]below:\scriptsize $4$}]{}}
        child
        {node [draw, label={[label distance = {\dista}]below:\scriptsize $12$}]{}}
        child
        {node (n6) [draw, label={[label distance = {\dist}]below:\scriptsize $6$}]{}}}
    child
    {node [draw, xshift = 0.5cm]{} [sibling distance = 0.8cm]
        child
        {node (n5) [draw, label={[label distance = {\dist}]below:\scriptsize $5$}]{}}
        child
        {node (n10) [draw, label={[label distance = {\dista}]below:\scriptsize $10$}]{}}
        child
        {node [draw, xshift = 1cm]{} [sibling distance = 0.5cm]
            child
            {node (n11) [fill, label={[label distance = {\dista}]below:\scriptsize $11$}]{}}
            child
            {node [fill, label={[label distance = {\dist}]below:\scriptsize $2$}]{}}
            child
            {node [fill, label={[label distance = {\dist}]below:\scriptsize $3$}]{}}
            child
            {node [fill, label={[label distance = {\dist}]below:\scriptsize $9$}]{}}
            child
            {node [fill, label={[label distance = {\dista}]below:\scriptsize $14$}]{}}
            child
            {node [fill, label={[label distance = {\dist}]below:\scriptsize $7$}]{}}
            child
            {node (n8) [fill, label={[label distance = {\dist}]below:\scriptsize $8$}]{}}}}
    child
    {node (n1) [draw, xshift = 1cm, label={[label distance = {\dist}]below:\scriptsize $1$}]{}}
    child
    {node (n13) [draw, xshift = 0.5cm, label={[label distance = {\dista}]below:\scriptsize $13$}]{}};

\draw[rounded corners, gray] ($(n4)+(-0.2,0.2)$) rectangle ($(n6)+(0.2,-0.6)$) {};
\draw[rounded corners, gray] ($(n5)+(-0.2,0.2)$) rectangle ($(n10)+(0.2,-0.6)$) {};
\draw[rounded corners, gray] ($(n11)+(-0.2,0.2)$) rectangle ($(n8)+(0.2,-0.6)$) {};
\draw[rounded corners, gray] ($(n1)+(-0.2,0.2)$) rectangle ($(n13)+(0.2,-0.6)$) {};

\node at ($(n4)+(-0.6,-0.2)$){\small $\leafchildset_1$};
\node at ($(n10)+(0.6,-0.2)$){\small $\leafchildset_2$};
\node at ($(n11)+(-0.6,-0.2)$){\small $\leafchildset_3$};
\node at ($(n13)+(0.6,-0.2)$){\small $\leafchildset_4$};
\end{tikzpicture}
\caption{Leaf child sets and leaf ranks under $\ranking$}\label{fig:sibling:leaves}
\end{subfigure}
\hfill
\begin{subfigure}[c]{0.45\textwidth}
\centering
\vspace{0.85cm}
\begin{tikzpicture}[scale=0.9]
\node at (0,0)[]{$r$};
\node at (1,0)[]{$2$};
\node at (2,0)[]{$3$};
\node at (3,0)[]{$7$};
\node at (4,0)[]{$8$};
\node at (5,0)[]{$9$};
\node at (6,0)[]{$11$};
\node at (7,0)[]{$14$};

\node at (0,-1)[]{$\ranking(r)$};
\node (2) at (1,-1)[circle, fill, inner sep=1.5pt]{};
\node (3) at (2,-1)[circle, fill, inner sep=1.5pt]{};
\node (7) at (3,-1)[circle, fill, inner sep=1.5pt]{};
\node (8) at (4,-1)[circle, fill, inner sep=1.5pt]{};
\node (9) at (5,-1)[circle, fill, inner sep=1.5pt]{};
\node (11) at (6,-1)[circle, fill, inner sep=1.5pt]{};
\node (14) at (7,-1)[circle, fill, inner sep=1.5pt]{};
\draw (2)--(3);
\draw (7)--(8)--(9);
\end{tikzpicture}
\vspace{0.85cm}
\caption{Inside leaf child set $\leafchildset_3$}\label{fig:invariant:sets}
\end{subfigure}
\caption{Connected component analysis}\label{fig:components}
\end{figure}
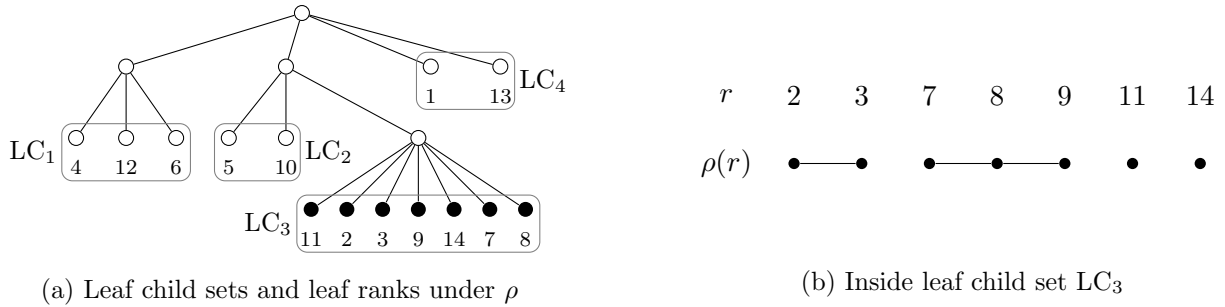

If each of the partition classes $X_i$ has size at most $k$, then each of the connected components of $S(\Pi,\ranking)$ has size $n_j \le k$, irrespective of $\ranking$. The maximum value that $\prod_j (n_j!)$ can take under the constraints $\sum_j n_j = n$ and $n_j \le k$ is no more than $(k!)^{n/k}$. By the Connectivity Lemma, we conclude that $\traces(\Pi) \ge n! / (k!)^{n/k}$, and that the query complexity is at least $\log_2(n!) - O(n \log(k))$. 

We can do better by observing that, for a random ranking $\ranking$, the number of adjacent-rank transpositions $\transposition$ that do not jump from one partition class to another is not much larger than the average size of the partition classes.
\begin{lemma}[lower bound for problems with the partition property]\label{lemma:connectivity:sensitive}
Let $\Pi$ be a computational problem in the comparison-query model on a set of size $n$. If $\Pi$ satisfies the partition property with class size at most $k$, then then $\traces(\Pi) \ge n!/(2(k!)^2)$.
\end{lemma}
\begin{proof}
For any rank $r \in [n-1]$, the probability that $\ranking^{-1}(r)$ and $\ranking^{-1}(r+1)$ belong to the same partition class equals $\sum_i \frac{|X_i|}{n} \frac{|X_i|-1}{n-1}$, which is at most $\frac{k-1}{n-1}$ provided each partition class $X_i$ has size at most $k$. It follows that the expected number of adjacent-rank transpositions $\transposition$ that do not change partition class, is at most $k-1$, so for a fraction at least half of the rankings $\ranking$ the number is at most $2(k-1)$.

The number of adjacent-rank transpositions $\transposition$ that do not change partition class for a given ranking $\ranking$ equals the number of edges in the successor graph $S(\Pi,\ranking)$. In terms of the sizes $n_j$, the number equals $\sum_j (n_j-1)$. We are considering rankings $\ranking$ for which the sum is at most $2(k-1)$. The maximum of $\prod_j (n_j!)$ under the constraints that $\sum_j (n_j-1) \le 2(k-1)$ and that each individual $n_j \le k$, is reached when two of the $n_j$'s equal $k$ and the rest are 1. Thus, if each of the partition classes $X_i$ are of size at most $k$, for a fraction at least half of the rankings $\ranking$, the size of the connected component of $\ranking$ in $G(\Pi)$ is at most $(k!)^2$. It follows that the number of connected components of $G(\Pi)$ is at least $n!/(2(k!)^2)$. The Connectivity Lemma then yields the claimed lower bound on $\traces(\Pi)$.
\end{proof}

Lemma~\ref{lemma:connectivity:sensitive} yields a lower bound of $\log(n!) - O(k \log(k))$ on the query complexity of $\Pi$ whenever $\Pi$ satisfies the partition property with class size at most $k$. 

Next we turn to sufficient conditions on the tree $T$ that guarantee the partition property for $\Pi_T$. For didactic reasons we first develop the conditions for binary trees, and then generalize them to arbitrary trees.

\subsection{Binary trees}
\label{sec:connectivity:binary}

In the case of binary trees $T$, the sensitivity analysis of Section~\ref{sec:binary:criterion} leads to a simple sufficient condition for the partition property to hold for $\Pi_T$.
Recall that we are assuming without loss of generality that $T$ has no internal nodes of degree 1, which in the case of binary trees is equivalent to saying that the tree is full: Every internal node has the maximum degree of 2. 

Consider criterion~\ref{eq:criterion:binary:simple} in Proposition~\ref{prop:sensitivity:criterion}. The right-hand side is always $-1$. As for the left-hand side, we know the following.
\begin{fact}\label{fact:DInv:parity}
For all disjoint sets $A, B \subseteq X$ and any ranking $\ranking$ of $X$, $\DInv_\ranking(A,B) = |A| \cdot |B| \bmod 2$. 
\end{fact}
\begin{proof}
As every pair in $\A \times B$ constitutes a cross-inversion for either $A$ to $B$, or $B$ to $A$, we have 
$\XInv_\ranking(A,B) + \XInv_\ranking(B,A) = |A| \cdot |B|$.
Thus,
\begin{align}
\DInv_\ranking(A,B) 
& \doteq  \XInv_\ranking(A,B) - \XInv_\ranking(B,A) \nonumber \\
& = (\XInv_\ranking(A,B) + \XInv_\ranking(B,A)) - 2 \XInv_\ranking(B,A) \nonumber \\
& = |A| \cdot |B| - 2 \XInv_\ranking(B,A). \label{eq:DInv:as:XInv}
\end{align}
As $\XInv_\ranking(B,A)$ in an integer, the claim follows.
\end{proof}
Fact~\ref{fact:DInv:parity} implies that whenever at least one of the leaf sets $L_\lo$ or $L_\hi$ is of even cardinality, then \eqref{eq:criterion:binary:simple} fails to hold, and $\Pi_T$ is sensitive to the underlying $\transposition$ at $\ranking$. Thus, we can guarantee that $\Pi_T$ satisfies the partition property provided that for any two siblings $u_1$ and $u_2$ in $T$ that are not both leaves, at least one of $\lvert\leafset(T_{u_1})\rvert$ or $\lvert\leafset(T_{u_2})\rvert$ is even. We refer to the latter condition as the \emph{product condition}. In trees without nodes of degree 1, the product condition can be expressed alternately in terms of the leaf child sets. We state and prove the result for arbitrary trees as it will help us in the next subsection to generalize the analysis.
\begin{proposition}\label{prop:connectivity:equivalence}
Let $T$ be a tree without nodes of degree 1. The following two conditions are equivalent:
\begin{itemize}
\item[(a)] For any two siblings $u_1$ and $u_2$ that are not both leaves, at least one of $\lvert\leafset(T_{u_1})\rvert$ or $\lvert\leafset(T_{u_2})\rvert$ is even.
\item[(b)] At most one leaf child set is odd, and if there exists a node $v^*$ with an odd leaf child set $\leafchildset(v^*)$, then all ancestors of $v^*$ have an empty leaf child set.
\end{itemize}
In the case of binary trees, (b) can be simplified to: At most one leaf has a non-leaf sibling.
\end{proposition}
\begin{proof}
We establish the two directions of implication separately. 

\medskip

\noindent
$\Rightarrow$: We argue the contrapositive. Suppose that at least two of the leaf child sets are odd. Start with the root of $T$ as the node $v$, and iterate the following: If $v$ has a child $u$ such that $T_u$ contains at least two nodes with an odd leaf child set, replace $v$ by such a child $u$. When the process ends, one of the following situations applies:
\begin{itemize}
\item There are two distinct children $u_1$ and $u_2$ of $v$ that are not leaves and each contain a single node with an odd leaf child set. In this case both $T_{u_1}$ and $T_{u_2}$ contain an odd number of leaves, violating (a).
\item There exists a unique child $u_1$ of $v$ that is not a leaf and contains a single node with an odd leaf child set, and $v$ itself has an odd number of leaf children. In this case, setting $u_2$ to any one leaf child of $v$ (which exists as their number is odd), leads to a violation of (a).
\end{itemize} 
Next, suppose that there exists a unique node $v^*$ that has an odd leaf child set, and that an ancestor $v$ of $v^*$ has a leaf child $u_1$. Setting $u_2$ to the child of $v$ that contains $v^*$ in its subtree, yields a violation of (a) as $T_{u_2}$ contains an odd number of leaves.

\medskip

\noindent
$\Leftarrow$: If neither $u_1$ nor $u_2$ are leaves, the first condition of (b) guarantees that at most one of $T_{u_1}$ or $T_{u_2}$ contains an odd number of leaves. If $u_1$ is a leaf and $u_2$ is not, then the second condition of (b) implies that $T_{u_2}$ cannot contain a node with an odd leaf child set, and therefore has an even number of leaves.

\medskip

In the case of binary trees, the first condition of (b) implies the second one, which can therefore be dropped from the equivalence statement. Moreover, for binary trees the first condition of (b) can be expressed as: At most one leaf has a non-leaf sibling.
\end{proof}

\begin{wrapfigure}{r}{0.4\linewidth}
    \centering
    \begin{tikzpicture}
        [
        level 1/.style = {sibling distance = 3cm, level distance = 0.8cm},
        level 2/.style = {sibling distance = 1.5cm, level distance = 0.8cm},
        level 3/.style = {sibling distance = 1.2cm, level distance = 0.8cm}
        ]

        \node [circle,draw,inner sep=2pt]{}
            child
            {node [circle,draw,inner sep=2pt]{}
                child
                {node [circle,draw,inner sep=2pt, label={below:\small $1$}]{}}
                child
                {node [circle,draw,inner sep=2pt]{}
                    child
                    {node [circle,draw,inner sep=2pt, label={below:\small $3$}]{}}
                    child
                    {node [circle,draw,inner sep=2pt, label={below:\small $6$}]{}}}}
            child
            {node [circle,draw,inner sep=2pt]{}
                child
                {node [circle,draw,inner sep=2pt, label={below:\small $2$}]{}}
                child
                {node [circle,draw,inner sep=2pt]{}
                    child
                    {node [circle,draw,inner sep=2pt, label={below:\small $4$}]{}}
                    child
                    {node [circle,draw,inner sep=2pt, label={below:\small $5$}]{}}}};
    \end{tikzpicture}
    \caption{Tree insensitive to non-sibling transposition}
    \label{fig:connectivity:insensitive-binary}
\end{wrapfigure}
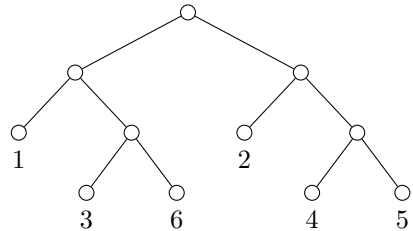

By Proposition~\ref{prop:connectivity:equivalence}, in a binary tree the condition that at most one leaf has no sibling is equivalent to the product condition, which implies the partition property of $\Pi_T$, so the lower bound of Lemma~\ref{lemma:connectivity:sensitive} applies. This establishes Theorem~\ref{thm:main:special} in the case of binary trees. 

As a side note, Fig.~\ref{fig:connectivity:insensitive-binary} shows the simplest example of a full binary tree $T$ and a ranking $\ranking$ for which there exists an adjacent-rank transposition $\transposition$ to which $\Pi_T$ is insensitive at $\ranking$ while the affected leaves are not siblings. The tree has two leaves without siblings, namely 1 and 2. The adjacent-rank transposition $(3,4)$ acts on nodes that are not siblings, but leaves the minimum number of inversions at 4.

\subsection{General trees}
\label{sec:connectivity:general}

For general trees $T$ the sensitivity analysis of $\Pi_T$ becomes more complicated than for binary trees, and we do not know of a simple sensitivity criterion like Proposition~\ref{prop:sensitivity:criterion}, but we can nevertheless extend the result for binary trees to arbitrary trees with similar constraints. 
For a given ranking $\ranking$ of $\leafset(T)$ and a given adjacent-rank transposition $\transposition$, we would like to figure out the effect of $\transposition$ on the objective $\MInv(T,\cdot)$, in particular when $\MInv(T,\transposition \ranking) = \MInv(T,\ranking)$. Recall the decomposition \eqref{eq:MInv:decomposition} of $\MInv(T,\cdot)$ from Section~\ref{sec:decomposition}. By Proposition~\ref{prop:change} the only term on the right-hand side of \eqref{eq:MInv:decomposition} that can be affected by the transposition $\transposition$ is
\[ \MRInv(T_v,\ranking) \doteq \min_\ordering \RInv(T,\ranking,\ordering) \]
corresponding to the node $v$ that is the least common ancestor $\LCA(\ell_\lo,\ell_\hi)$ of the two leaves $\ell_\lo$ and $\ell_\hi$ that are affected by $\transposition$ under $\ranking$. In Section~\ref{sec:sensitivity:binary} we considered the two possible relative orderings $\ordering_1$ and $\ordering_2$ of the children of $v$, and derived a criterion for when the lowest cost does not change under $\transposition$. More precisely, when 
\begin{equation}\label{eq:pairwise}
\min(\RInv(T,\ranking,\ordering_1), \RInv(T,\ranking,\ordering_2)) = \min(\RInv(T,\transposition \ranking,\ordering_1), \RInv(T,\transposition \ranking,\ordering_2)). 
\end{equation}
There are two complications in generalizing this approach from binary to general trees.
\begin{itemize}
\item 
The expression \eqref{eq:cost} for $\RInv(T,\ranking,\ordering)$ involves multiple terms instead of just one as in \eqref{eq:contribution:binary:retake}. This complicates probabilistic analyses like the one we did in Section~\ref{sec:sensitivity:binary} because the difference in cost of the two relative orderings of two children is also affected by parts of the tree outside of their combined subtrees. The issue did not 
matter for the analysis in Section~\ref{sec:sensitivity:general}. We will be able to manage it here, as well.
\item There now are not just two but multiple possible orderings $\ordering$, and it is not clear what pairs $(\ordering_1,\ordering_2)$ we need to impose \eqref{eq:pairwise} on in order to guarantee that $\MRInv(T_v,\ranking) = \MRInv(T_v,\transposition \ranking)$ but no more.
\end{itemize}
In Section~\ref{sec:sensitivity:general} we circumvented the second issue by only considering sensitivities that decrease the objective function, and establishing a lower bound on their occurrence independent of the ordering $\ordering$. Here we are also able to handle the second issue by shooting for a sufficient condition for sensitivity rather than a criterion. We do so by requiring that for no pair of distinct orderings $\ordering_1$ and $\ordering_2$, condition~\eqref{eq:pairwise} holds (unless the two affected leaves are siblings). Similar to the case of binary trees, we guarantee that \eqref{eq:pairwise} fails based on parity considerations given the product condition. 

For the analysis we again assume without loss of generality that $T$ has no nodes of degree 1. We use the same notation as in Section~\ref{sec:sensitivity:binary}: Let  $\ell_\lo$ denote the affected leaf that is smaller with respect to $\ranking$, and $\ell_\hi$ the other affected leaf. Let $L_i$ denote the leaf set $L_i \doteq \leafset(T_{u_i})$, where $u_1, \dots, u_k$ are the children of $v = \LCA(\ell_\lo,\ell_\hi)$. We also write $u_\lo$ for the child of $v$ that contains $\ell_\lo$ in its subtree, and $L_\lo$ for the leaf set of the subtree rooted at $u_\lo$, and define $\ell_\hi$, $u_\hi$, and $L_\hi$ similarly. See Figure~\ref{fig:sensitivity-general} for a sketch of the setting.

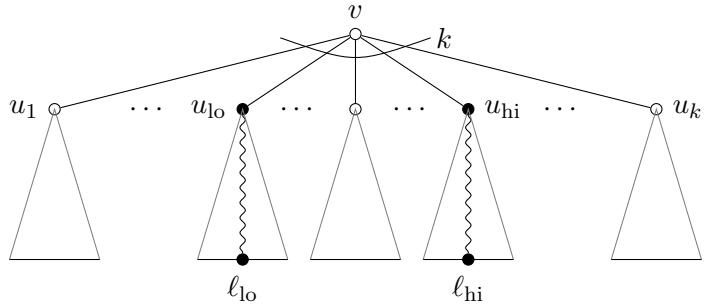
\begin{figure}[ht]
    \centering
    \begin{tikzpicture}
        \node (top) at (0,0)[circle,draw,inner sep=1.5pt,label={above:$v$}]{};
        \node (u1) at (-4,-1)[circle,draw,inner sep=1.5pt, label={left:$u_1$}]{};
        \node at (-2.75,-1){$\cdots$};
        \node (ulo) at (-1.5,-1)[circle,draw,fill=black,inner sep=1.5pt,label={left:$u_{\lo}$}]{};
        \node at (-0.75,-1){$\cdots$};
        \node (u2) at (0,-1)[circle,draw,inner sep=1.5pt]{};
        \node at (0.75,-1){$\cdots$};
        \node (uhi) at (1.5,-1)[circle,draw,fill=black,inner sep=1.5pt,label={right:$u_{\hi}$}]{};
        \node at (2.75,-1){$\cdots$};
        \node (u3) at (4,-1)[circle,draw,inner sep=1.5pt, label={right:$u_k$}]{};

        \node (l1) at (-4,-3){};
        \node (llo) at (-1.5,-3)[circle,draw,fill=black,inner sep=1.5pt,label={below:$\ell_{\lo}$}]{};
        \draw[decorate, decoration=snake, segment amplitude=1pt, segment length=2mm] (ulo)--(llo);
        \node (l2) at (0,-3){};
        \node (lhi) at (1.5,-3)[circle,draw,fill=black,inner sep=1.5pt,label={below:$\ell_{\hi}$}]{};
        \draw[decorate, decoration=snake, segment amplitude=1pt, segment length=2mm] (uhi)--(lhi);
        \node (l3) at (4,-3){};
        
        \def\gap{0.6}
        \def\drop{2}
        \draw[thin, gray] (u1) ++({-\gap},{-\drop}) -- +({\gap},{\drop}) -- +({2*\gap},0);
        \draw (l1) ++({-\gap},0) -- +({2*\gap},0);
        \draw[thin, gray] (ulo) ++({-\gap},{-\drop}) -- +({\gap},{\drop}) -- +({2*\gap},0);
        \draw (llo) ++({-\gap},0) -- +({2*\gap},0);
        \draw[thin, gray] (u2) ++({-\gap},{-\drop}) -- +({\gap},{\drop}) -- +({2*\gap},0);
        \draw (l2) ++({-\gap},0) -- +({2*\gap},0);
        \draw[thin, gray] (uhi) ++({-\gap},{-\drop}) -- +({\gap},{\drop}) -- +({2*\gap},0);
        \draw (lhi) ++({-\gap},0) -- +({2*\gap},0);
        \draw[thin, gray] (u3) ++({-\gap},{-\drop}) -- +({\gap},{\drop}) -- +({2*\gap},0);
        \draw (l3) ++({-\gap},0) -- +({2*\gap},0);
        
        \draw (top)--(u1);
        \draw (top)--(ulo);
        \draw (top)--(u2);
        \draw (top)--(uhi);
        \draw (top)--(u3);
        
        \draw (-1,-0.05) .. controls (0,-0.4) .. (1,-0.05);
        \node at (0.8, -0.05)[label={right:$k$}]{};
    \end{tikzpicture}
    \caption{Sensitivity for general trees}
    \label{fig:sensitivity-general}
\end{figure}

We slightly abuse notation and  let $\ordering$ denote both the ordering of the entire tree $T$ as well as the ranking of the children of $v$. By the analysis of Section~\ref{sec:sensitivity:binary}, we have that for two orderings $\ordering_1$ and $\ordering_2$ the situation \eqref{eq:pairwise} can only occur if $\ordering_1(u_\lo) < \ordering_1(u_\hi)$, $\ordering_2(u_\lo) > \ordering_2(u_\hi)$, and 
\begin{equation}\label{eq:RInv:difference}
\RInv_\ranking(T_v,\ranking,\ordering_1) -
\RInv_\ranking(T_v,\ranking,\ordering_2) = -1 = \DInv_\ranking(\{\ell_\lo\},\{\ell_\hi\}).
\end{equation}
For any two disjoint sets of leaves, \eqref{eq:DInv:as:XInv} lets us write 
\begin{equation}\label{eq:XInv:DInv}
\XInv_\ranking(A,B) =  \frac{1}{2} \DInv_\ranking(A,B) + \frac{1}{2} |A| \cdot |B|.
\end{equation}
Applying \eqref{eq:XInv:DInv} to all the terms involved in \eqref{eq:cost}, we have
\begin{align*}
\RInv_\ranking(T_v,\ranking,\ordering) 
& = \sum_{1 \le i < j \le k} \XInv_\ranking(L_i,L_j) \cdot \Indicator[\ordering(i) < \ordering(j)] \\
& \;\;\;\;\; + \sum_{1 \le i < j \le k} \XInv_\ranking(L_j,L_i) \cdot \Indicator[\ordering(i) > \ordering(j)] \\
& = \frac{1}{2} \sum_{1 \le i < j \le k} \left( \DInv_\ranking(L_i,L_j) \cdot (-1)^{\Indicator[\sigma(i > \sigma(j)]} + |L_1| \cdot |L_2| \right) \\
\RInv_\ranking(T_v,\ranking,\ordering_1) -
\RInv_\ranking(T_v,\ranking,\ordering_2)
& = 
\sum\limits_{\substack{1 \le i < j \le k \\ \ordering_1(i) < \ordering_2(j) \\ \ordering_2(i) > \ordering_2(j)}}\DInv_\ranking(L_i,L_j)
+ \sum\limits_{\substack{1 \le i < j \le k \\ \ordering_1(i) > \ordering_2(j) \\ \ordering_2(i) < \ordering_2(j)}} \DInv_\ranking(L_i,L_j)
\end{align*}
By combining the last equation with \eqref{eq:RInv:difference} and separating out the term for $(i,j)=(\lo,\hi)$, we obtain the following necessary condition for \eqref{eq:pairwise} to hold:
\begin{equation}\label{eq:necessary}
\DInv_\ranking(L_\lo,L_\hi) - \DInv_\ranking(\{\ell_\lo\},\{\ell_\hi\})
= -\sum\limits_{\substack{1 \le i < j \le k \\ \ordering_1(i) < \ordering_2(j) \\ \ordering_2(i) > \ordering_2(j) \\ (i,j) \ne (\lo,\hi)}} \DInv_\ranking(L_i,L_j)
+ \sum\limits_{\substack{1 \le i < j \le k \\ \ordering_1(i) > \ordering_2(j) \\ \ordering_2(i) < \ordering_2(j)}} \DInv_\ranking(L_i,L_j).
\end{equation} 

In order for $\Pi_T$ to have the partition property, it suffices to ensure that \eqref{eq:necessary} fails whenever $u_\lo$ and $u_\hi$ are not both leaves. By \eqref{eq:DInv:as:XInv} each of the terms $\DInv_\ranking(L_i,L_j)$ in \eqref{eq:necessary} has the same parity as $|L_i| \cdot |L_j|$. Since $\DInv_\ranking(\{\ell_\lo\},\{\ell_\hi\})$ is odd, it follows that \eqref{eq:necessary} fails whenever at most one of the leaf sets $L_i$ involved is odd, which is condition (a) in Proposition~\ref{prop:connectivity:equivalence}. Switching to the equivalent condition (b) from Proposition~\ref{prop:connectivity:equivalence} allows us to conclude via Lemma~\ref{lemma:connectivity:sensitive}:
\begin{theorem}
Let $T$ be a tree without nodes of degree 1 such that the leaf child sets have size at most $k$, at most one of them is odd, and if there exists an odd one, say $\leafchildset(v^*)$, then all ancestors of $v^*$ have empty leaf child sets. Then $\traces(\Pi_T) \ge n!/(2(k!)^2)$.
\end{theorem}
Theorem~\ref{thm:main:special} follows by taking the base-2 logarithm of the bound.

\subsection{Counting cross inversions and evaluating the Mann--Whitney statistic}
\label{sec:connectivity:cross-inversions}

We now apply the connectivity approach that we captured in Proposition~\ref{prop:successor:components} to the problem $\Pi_{\XInv}$ of computing the number of cross inversions between two disjoint sets $A$ and $B$ with respect to a ranking $\ranking$ of $X = A \sqcup B$. Note that this problem is a refinement of evaluating the Mann--Whitney statistic, or equivalently, of inversion minimization on the tree $T$ of Figure~\ref{fig:Mann-Whitney}: Any algorithm that solves $\Pi_{\XInv}$ with $q$ queries, can be transformed into an algorithm for $\Pi_T$ with $q$ queries, namely by transforming the output $y$ of the algorithm for $\Pi_{\XInv}$ to $\min(y,|A| \cdot |B|-y)$. Viewed in the contrapositive, a lower bound for $\Pi_{\XInv}$ is easier to obtain than one for $\Pi_T$ on the tree $T$ of Figure~\ref{fig:Mann-Whitney}. We first establish a lower bound for $\Pi_{\XInv}$ and then see how it extends to $\Pi_T$.

One can think of $\Pi_{\XInv}$ as inversion minimization on the Mann--Whitney tree without allowing swapping the two children of the root. As a result, the problem $\Pi_{\XInv}$ is sensitive to \emph{every} adjacent-rank transposition between non-siblings, and therefore automatically satisfies the partition property (with $A$ and $B$ being the partition classes), so Proposition~\ref{prop:successor:components} applies. In contrast, the problem $\Pi_T$ of minimizing inversions on the Mann--Whitney tree may not have the partition property. This is why analyzing $\Pi_{\XInv}$ is a bit simpler, and why we handle it first.

Let $a\doteq |A|$ and $b\doteq |B|$. By the partition property, the average sensitivity of $\Pi_{\XInv}$ equals $\frac{2ab}{a+b}$, which via the Sensitivity Lemma yields a query lower bound of $\Omega(a \log(a))$ for $a \le b$. To obtain the stronger lower bound of $\Omega((a+b)\log(a))$ we need a more detailed analysis of the connectivity of the permutahedron graph $G(\Pi_{\XInv})$. 

For a given ranking $\ranking$, let $x_1,\dots,x_{a}$ be the elements of $A$ listed in increasing order, and similarly for $y_1,\dots,y_b$ for the elements of $B$. We define $m_1,\dots,m_{b+1}$ such that for each $i$, $m_i$ is the number of elements of $A$ between $y_{i-1}$ and $y_i$. (Here, $y_0$ and $y_{b+1}$ serve as sentinels with an infinitely low and infinitely high rank.) Similarly, we define $n_1,\dots,n_{a+1}$ as the number of elements in $B$ between successive elements of $A$. The numbers $m_i$ and $n_i$ are the sizes of the connected components of the successor graph $S(\Pi_{\XInv},\ranking)$ (possibly with some additional zeroes). By Proposition~\ref{prop:successor:components}, the connected component of $\ranking$ in $G(\Pi_{\XInv})$ has size
\begin{equation}
    \label{eq:mann-whitney:cc-size}
    (m_1)!\cdots (m_{b+1})!(n_1)!\cdots (n_{a+1})!.
\end{equation}
Depending on the values of $m_1,\dots,m_{b+1},n_1,\dots,n_{a+1}$, some connected components may be much larger than others. We apply the Connectivity Lemma in a similar way as in Lemma~\ref{lemma:connectivity:sensitive} and only count the rankings $\ranking$ that are in small connected components, which are the rankings for which $m_1,\dots,m_{b+1},n_1,\dots,n_{a+1}$ are bounded. Let $m^*$ and $n^*$ be the minimum integers for which 
\begin{equation*}
    \Pr[m_1,\dots,m_{b+1} \le m^*] \ge \frac{3}{4}\quad\text{and}\quad
    \Pr[n_1,\dots,n_{a+1} \le n^*] \ge \frac{3}{4}.
\end{equation*}
By a union bound, the probability that both of these events hold is at least $1/2$. In other words, there are least $(a+b)!/2$ rankings $\ranking$ for which $m_1,\dots,m_{b+1}\le m^*$ and $n_1,\dots,n_{a+1}\le n^*$.
\begin{proposition}\label{prop:mann-whitney:cc-bound}
    \begin{equation}\label{eq:mann-whitney:cc-bound}
        \traces(\Pi_T) \ge \frac{(a+b)!}{2(m^*)!^{a/m^*}(n^*)!^{b/n^*}}.
    \end{equation}
\end{proposition}

\begin{proof}
We consider the rankings for which $m_1,\dots,m_{b+1}\le m^*$ and $n_1,\dots,n_{a+1}\le n^*$. We first argue the following upper bound on the size \eqref{eq:mann-whitney:cc-size} of the connected component in $G(\Pi_{\XInv})$ of any such ranking $\ranking$:
     \begin{equation}\label{eq:bound:product}
        (m_1)!\cdots (m_{b+1})!(n_1)!\cdots (n_{a+1})! \le (m^*)^{a/m^*}(n^*)^{b/n^*}.
    \end{equation}
For nonnegative integers $n$, $n!^{1/n}$ is increasing. This can be seen by noticing that $\log(n!^{1/n})$ is the average of $\log(1),\dots,\log(n)$. As a result, for $i \in [b+1]$, $(m_i)!^{1/m_i}\le (m^*)!^{1/m^*}$, or equivalently, $(m_i)!\le (m^*)^{m_i/m^*}$. Using the fact that $m_1+\dots +m_{b+1}=a$,
    \begin{equation*}
        (m_1)!\cdots (m_{b+1})!\le (m^*)!^{(m_1+\dots+m_{b+1})/m^*}=(m^*)!^{a/m^*}.
    \end{equation*}
We can apply similar reasoning to get $(n_1)!\cdots (n_{a+1})!\le (n^*)!^{b/n^*}$. From this, we conclude that the size of the connected components among the rankings under consideration is at most the right-hand side of \eqref{eq:bound:product}. 
    
Since there are at least $(a+b)!/2$ of the rankings under consideration, we derive the stated bound on $\traces(\Pi_T)$ by applying the Connectivity Lemma.
\end{proof}

Now, we find concrete bounds on $m^*$ and $n^*$.

\begin{proposition}
    \label{prop:mann-whitney:chunk-size}
    \[\max\left(1,\frac{a}{b+1}\right)\le m^* \le \frac{a+b}{b}\ln(4(b+1)). \]
    Symmetrically,
    \[\max\left(1,\frac{b}{a+1}\right)\le n^* \le \frac{a+b}{a}\ln(4(a+1)). \]
\end{proposition}

\begin{proof}
We first prove the upper bound on $m^*$. Let $k$ be a positive integer. We compute the probability, over an average ranking $\ranking$, that $m_i> k$ for a specific $i$. 
Notice that there is a one-to-one correspondence between the ranking $\ranking$ and the corresponding sequence of nonnegative integers $m_1,\dots,m_{b+1}$ such that $m_1+\dots+m_{b+1}=a$, because the ranks of $B$ can be uniquely recovered as $m_1+1,m_1+m_2+2,\dots$, and the remaining ranks form $A$. By stars and bars, there are $\binom{a+b}{b}$ such sequences. Now, if $m_i>k$, then $m_i-(k+1)$ is an arbitrary nonnegative integer, and $m_1+\dots+(m_i-(k+1))+\dots+m_{b+1}=a-k-1$. By stars and bars, there are $\binom{a+b-k-1}{b}$ such sequences. Therefore, the probability that $m_i>k$ is $\binom{a+b-k-1}{b}/\binom{a+b}{b}$. Continuing,
\begin{align*}
\Pr[m_i> k]&=\frac{(a+b-k-1)\cdots (a-k)}{(a+b)\cdots (a+1)}
        \le \left(\frac{a+b-k-1}{a+b}\right)^b
        \le \exp\left(-b\cdot \frac{k+1}{a+b} \right),
\end{align*}
where the last step uses the bound $1+x \le \exp(x)$.
By a union bound and taking the complement,
    \begin{equation}\label{eq:stars}
        \Pr[m_1,\dots,m_{b+1}\le k] \ge 1-(b+1)\exp\left(-b\cdot \frac{k+1}{a+b} \right).
    \end{equation}
If $k$ is such that the right-hand side of \eqref{eq:stars} is at least $3/4$, we know that $m^* \le k$. Solving for $k$ yields the stated bound on $m^*$.
    
For the lower bounds, $m^*\ge 1$ because at least one of $m_1,\dots,m_{b+1}$ is at least 1, and $m^*\ge \frac{a}{b+1}$ because $m_1+\dots+m_{b+1}=a$, and $m^*$ is greater than or equal to the average term in the sum.
\end{proof}

The first part of Theorem~\ref{thm:Mann-Whitney} now comes from taking the logarithm of $\traces(\Pi_{\XInv})$ in Proposition~\ref{prop:mann-whitney:cc-bound} and using the bounds in Proposition~\ref{prop:mann-whitney:chunk-size}.

\begin{proof}[Proof of the first part of Theorem~\ref{thm:Mann-Whitney}]
We mainly make use of the following approximation based on Stirling's formula.
    \begin{equation}
        \label{eq:stirling-ln}
        \ln(n!) = \left(n+\frac{1}{2}\right)\ln(n)-n+O(1).
    \end{equation}
    In order to estimate $\ln \traces(\Pi_{\XInv})$, we need to estimate $\ln((a+b)!)-\frac{a}{m^*}\ln(m^*)-\frac{b}{n^*}\ln(n^*)$. By (\ref{eq:stirling-ln}),
    \begin{align}
        \ln((a+b)!) &= \left(a+b+\frac{1}{2}\right)\ln(a+b) - (a+b) + O(1)\\
        \frac{a}{m^*}\ln((m^*)!) &= \left(a+\frac{a}{2m^*}\right)\ln(m^*) - a + \frac{a}{m^*}O(1) \label{eq:mann-whitney:1}\\
        \frac{b}{n^*}\ln((n^*)!) &= \left(b+\frac{b}{2n^*}\right)\ln(n^*) - b + \frac{b}{n^*}O(1) \label{eq:mann-whitney:2}
    \end{align}
    We can use the lower bounds in Proposition~\ref{prop:mann-whitney:chunk-size}, namely that $m^*\ge 1$ and $n^*\ge \frac{b}{a+1}$, to simplify the occurrences of $m^*$ and $n^*$ in the denominators of \eqref{eq:mann-whitney:1} and \eqref{eq:mann-whitney:2}. Therefore,
    \begin{align*}
        \ln \traces(\Pi_{\XInv}) &\ge a\ln\left(\frac{a+b}{m^*}\right)+b\ln\left(\frac{a+b}{n^*}\right) - \frac{a}{2}\ln(m^*) - \frac{a+1}{2}\ln(n^*) + \frac{1}{2}\ln(a+b) - O(a)\\
        &= a\ln\left(\frac{a+b}{m^*\sqrt{m^*n^*}}\right)+\left(b+\frac{1}{2}\right)\ln\left(\frac{a+b}{n^*}\right)- O(a).
    \end{align*}
    Using the upper bounds in Proposition~\ref{prop:mann-whitney:chunk-size}, absorbing low-order terms, and using the condition that $a \le b$, we get
    \begin{align*}
        \log \traces(\Pi_{\XInv}) &\ge \Omega\left(a \log\left(\frac{b\sqrt{ab}}{a+b}\right) + b \log(a)\right) \ge \Omega(a \log(\sqrt{ab}/2) + b \log(a)) = 
        \Omega((a+b)\log(a)).
    \end{align*}
\end{proof}

\paragraph{Evaluating the Mann--Whitney statistic.} 
We now argue how the second part of Theorem~\ref{thm:Mann-Whitney} follows, i.e., that the lower bound of $\Omega((a+b)\log(a))$ holds for inversion minimization on the Mann--Whitney tree $T$ of Figure~\ref{fig:Mann-Whitney} with $a \le b$. We do so by tweaking our lower bound argument for $\Pi_{\XInv}$ to the setting of $\Pi_T$.

How does the permutahedron graph for $\Pi_T$ relate to the one for $\Pi_{\XInv}$? The problem $\Pi_T$ is a coarsening of the problem $\Pi_{\XInv}$: Output values $y$ and $ab-y$ for $\Pi_{\XInv}$ are both mapped to $\min(y,ab-y)$ under $\Pi_T$. This means that all edges present in $G(\Pi_{\XInv})$ are also present in $G(\Pi_T)$, but there may be more, and some of the connected components in $G(\Pi_{\XInv})$ corresponding to output value $y$, may be merged in $G(\Pi_T)$ with some of the connected components of $G(\Pi_{\XInv})$ corresponding to output value $ab-y$. However, by the reasoning behind Proposition~\ref{prop:change}, edges in $G(\Pi_{\XInv})$ can only go between rankings whose value under $\Pi_{\XInv}$ differ by at most one. This means that the above merging of connected components can only happen if the difference between $y$ and $ab-y$ is 1, i.e., for the values $\floor{ab/2}$ and $\ceil{ab/2}$, and only if $ab$ is odd. In fact, this is exactly the situation that we analyzed in Figure~\ref{fig:sensitivity:binary}, where $v$ coincides with the root of $T$. 

If we ignore the rankings with value $\floor{ab/2}$ or $\ceil{ab/2}$ under $\Pi_{\XInv}$, our lower bound argument for $\Pi_{\XInv}$ carries over verbatim to $\Pi_T$, except that on the right-hand side of Proposition~\ref{prop:mann-whitney:cc-bound} the factor of $\frac{1}{2}$ is replaced by $\frac{1}{2}-W$, where $W$ represents the fraction of rankings with value $\floor{ab/2}$ or $\ceil{ab/2}$ under $\Pi_{\XInv}$. Lemma~\ref{thm:gaussian:main} tells us that $W \le 2C/\sqrt{ab(a+b)}$, where $C$ denotes the constant from the lemma. Thus, we obtain a lower bound for $\traces(\Pi_T)$ that is a negligible fraction smaller than the one for $\traces(\Pi_{\XInv})$. Taking logarithms, we obtain the same lower bound for the query complexity up to an additive term. In particular, we obtain a query lower bound of $\Omega((a+b)\log(a))$ for $\Pi_T$ in case $a \le b$. This is the second part of Theorem~\ref{thm:Mann-Whitney}.

\section{Cross-Inversion Distribution}
\label{sec:cross-inversions}

In this section we prove the upper bound we need for the proof of Theorem~\ref{thm:main:binary} in Section~\ref{sec:sensitivity:root}, namely Lemma~\ref{thm:gaussian:main}. Recall that $X_{a,b}$ denotes a random variable that counts the number of cross inversions $\XInv(A,B)$ from $A$ to $B$, where $A$ is an array of length $a$, $B$ an array of length $b$, and the concatenation $AB$ is a random permutation of $[a+b]$.
Lemma~\ref{thm:gaussian:main} states that for all positive integers $a$ and $b$, $X_{a,b}$ takes on no value with probability more than $C/\sqrt{ab(a+b)}$, where $C$ is a universal constant.

We establish Lemma~\ref{thm:gaussian:main} by considering the 
characteristic function $\charfn_{a,b}(t)$ of $X_{a,b}$, which is the 
Fourier transform of the density function of $X_{a,b}$: $\charfn_{a,b}(t) \doteq \Expect(e^{itX_{a,b}})$. The probabilities can be retrieved from the characteristic function by applying the inverse Fourier transform. This allows us to express the probabilities as the following integrals: For any integer $k$ in $\{0,\dots,a+b\}$
\begin{equation}
    \label{eq:fouriertransform}
    \Pr[X_{a,b}=k] = \frac{1}{2\pi} \int_{-\pi}^{\pi}\charfn_{a,b}(t)e^{-itk}\,dt.
\end{equation}
The right-hand side of \eqref{eq:fouriertransform} is the general formula for the inverse Fourier transform of a periodic function from $\RR$ to $\CC$ with period $2\pi$. The formula applies as the density function of an integer-valued random variable can be extended to a periodic function with period $2\pi$. An alternate argument from first principles observes that the characteristic function of a finite distribution over the nonnegative integers is a polynomial in $z=e^{it}$, where the coefficient of degree $k$ equals the probability of the outcome $k$.%
\footnote{In the case at hand, after multiplication by $\binom{a+b}{a}$, the resulting polynomial is known as the Gaussian polynomial with parameter $(a,b)$.} 
Formula \eqref{eq:fouriertransform} then follows because 
\[
\int_{-\pi}^\pi z^d e^{-ikt} \, dt = 
\int_{-\pi}^\pi e^{i(d-k)t} \, dt =
\left\{ \begin{array}{cl} 2\pi & d=k \\
                             0 & d \ne k
\end{array} \right.
\]
The following lemma then represents the essence of the proof of 
 Lemma~\ref{thm:gaussian:main}.
\begin{lemma}
    \label{thm:integral}
    Then there exists a constant $C$ such that for all integers $a$, $b$ with $b \ge a \ge 2$
      \begin{equation}
        \label{eq:mainbound}
        \int_{-\pi}^{\pi} \lvert \charfn_{a,b}(t)\rvert\,dt \le \frac{C}{b\sqrt{a}}.
    \end{equation}  
where $\charfn_{a,b}(t) \doteq \Expect(e^{itX_{a,b}})$. 
\end{lemma}
\begin{proof}[Proof of Lemma~\ref{thm:gaussian:main}]
By symmetry, it suffices to consider the case where $a \le b$. In the case where $a=1$, the distribution of $X_{a,b}$ is uniform over $\{0,\dots,b\}$, so the maximum probability is $\frac{1}{b+1} \le C/\sqrt{ab(a+b)}$ for any constant $C \ge \sqrt{2}/2$. Otherwise, 
we have
\[ \Pr[X_{a,b}=k] = \frac{1}{2\pi} \int_{-\pi}^{\pi}\charfn_{a,b}(t)e^{-itk}\,dt \le \frac{1}{2\pi} \int_{-\pi}^{\pi} \lvert\charfn_{a,b}(t)\rvert\,dt \le \frac{C}{2\pi b\sqrt{a}} \le \frac{C}{\sqrt{ab(b+1)}} \]
by \eqref{eq:fouriertransform} and Lemma~\ref{thm:integral}.
\end{proof}

To establish Lemma~\ref{thm:integral}, as $\lvert \charfn_{a,b}(t) \rvert$ is an even function, it suffices to take the integral (\ref{eq:mainbound}) over the domain $[0,\pi]$ and multiply by two: 
\begin{equation}
    \label{eq:integral}
    \int_{-\pi}^{\pi} \lvert \charfn_{a,b}(t)\rvert\,dt 
    = 2 \int_0^{\pi} \lvert \charfn_{a,b}(t)\rvert\,dt 
\end{equation}  

We divide the domain of integration on the right-hand side of \eqref{eq:integral} into two regions: one close to zero, and the rest. The integrand is well-behaved in the center near zero, with it being approximated accurately by a normal curve. It is harder to analyze the behavior of the function away from zero. In this region, a pole reduction lemma (captured by \cref{lemma:bijection}) that hinges on a combinatorial matching result (\cref{lemma:intervalmatching}), plays a crucial role in eliminating most of the messy behavior of the function and still providing an effective bound.

We first derive an expression for the characteristic function of $X_{a,b}$ in Section~\ref{sec:characteristic}, bound the central part of the integral in Section~\ref{sec:central}, the peripheral part in Section~\ref{sec:peripheral}, and conclude with the pole reduction lemma in Section~\ref{sec:pole-reduction}.

\subsection{Characteristic function}
\label{sec:characteristic}

The characteristic function of a generic random variable $X$ is defined as  $\charfn_X(t):\mathbb{R}\to\mathbb{C}: t \mapsto \Expect(e^{itX})$. It always exists and has the following interesting property (among others):
\begin{fact}
    \label{fact:charfn-mult}
    For independent random variables $X,Y$,
    \[\charfn_{X+Y}(t)=\charfn_X(t)\charfn_Y(t). \]
\end{fact}
\begin{proof}
    \[\charfn_{X+Y}(t)=\Expect(e^{it(X+Y)})=\Expect(e^{itX}e^{itY})=\Expect(e^{itX})\Expect(e^{itY})=\charfn_X(t)\charfn_Y(t). \]
\end{proof}

Our derivation of the characteristic function $\charfn_{a,b}$ of $X_{a,b}$ is based on a connection between cross inversions and inversions in arrays.
\begin{fact}
    \label{claim:cross-inversion-property}
    Let $A$ and $B$ be arrays, and let $AB$ be the concatenation of $A$ with $B$. Then
    \[\Inv(AB)=\Inv(A)+\Inv(B)+\XInv(A,B). \]
\end{fact}
\begin{proof}
    Any inversion in $AB$ is either between two elements of $A$, two elements of $B$, or one element of $A$ and one element of $B$. In each case, the inversion is counted in $\Inv(A)$, $\Inv(B)$, or $\XInv(A,B)$, respectively.
\end{proof}

Let $\InvD_a$ be the random variable that counts the number $\Inv(A)$ of inversions in an array $A$ that is a uniform permutation of $[a]$, and let $\charfn_a(t)$ be the characteristic function of $\InvD_a$.
\begin{claim}
    \label{claim:inversion-charfn}
    We have
    \[\charfn_a(t)=\prod_{k=1}^{a}\left(\frac{e^{it(k-1)}}{k}\cdot\frac{\sin(kt/2)}{\sin(t/2)}\right). \]
\end{claim}
\begin{proof}
Consider the process of placing the elements $1,\dots,a$ one by one, each time placing each new element between two elements or on some end of the array, to form an array $A$. 

For $k=1,\dots,a$, consider the random variable that counts the number of new inversions formed with $k$ when $k$ is placed. First of all, when $k$ is placed, the number of new inversions is equal to the number of elements to the left of $k$ at the time of placement (only the elements $1,\dots,k-1$ have been placed at this point). This means the random variable has a uniform distribution over $\{0,\dots,k-1\}$, which we denote by $U_{k-1}$. Furthermore, this situation applies regardless of the placement of the other elements, so this random variable $U_{k-1}$ is independent from all other previous random variables $U_{j-1}$ with $j<k$. 

Therefore, $\InvD_a$ can be written as the following sum of independent variables:
\[\InvD_a = \UniformD_0 + \dots + \UniformD_{a-1}. \]
We can use \cref{fact:charfn-mult} to calculate the characteristic function:
\begin{align}\label{eq:char:Xa}
\begin{split}
\charfn_a(t)&=\prod_{k=1}^{a} \mathbb{E}[e^{itU_{k-1}}]=\prod_{k=1}^{a}\left(\frac{1}{k}\sum_{m=0}^{k-1} e^{itm}\right) =\prod_{k=1}^{a}\left(\frac{e^{it(k-1)}}{k}\cdot\frac{\sin(kt/2)}{\sin(t/2)}\right) \\
&= e^{ia(a-1)/2} \prod_{k=1}^{a}\left(\frac{1}{k}\cdot\frac{\sin(kt/2)}{\sin(t/2)}\right).
\end{split}
\end{align}
The second-to-last step follows from the geometric sum formula and the identity $e^{it}-e^{-it}=2\sin(t)$, and the last step from the arithmetic sum formula. 
\end{proof}

Consider a random permutation of $[a+b]$, let $A$ be the array consisting of the first $a$ elements, and $B$ the array consisting of the remaining $b$. Then $\Inv(A)$ has distribution $X_a$, $\Inv(B)$ distribution $X_b$, $\Inv(AB)$ distribution $X_{a+b}$, and $\XInv(A,B)$ distribution $X_{a,b}$. By \cref{claim:cross-inversion-property}, we have:
\[ \InvD_{a+b} = \InvD_a + \InvD_b + \XInvD_{a,b}. \]\
Moreover, the values of $\Inv(A)$, $\Inv(B)$, and $\XInv(A,B)$ are independent. Hence, by \cref{fact:charfn-mult}
\[\charfn_{a+b}(t)=\charfn_{a}(t)\charfn_{b}(t)\charfn_{a,b}(t), \]
or
\[\charfn_{a,b}(t)=\frac{\charfn_{a+b}(t)}{\charfn_a(t)\charfn_b(t)}. \]
By \eqref{eq:char:Xa} we conclude:
\begin{proposition}
\label{prop:characteristicfunction}
For integers $n\ge 0$, let $s_n(t)=\prod_{k=1}^n\frac{\sin(kt)}{k}$. Then 
    \[\charfn_{a,b}(t)=e^{itab/2}\frac{s_{a+b}(t/2)}{s_a(t/2)s_b(t/2)}
    \text{ and } \lvert\charfn_{a,b}(t)\rvert=\left| \frac{s_{a+b}(t/2)}{s_a(t/2)s_b(t/2)}\right|. 
    \]
\end{proposition}

\subsection{Center bound}
\label{sec:central}

For the first piece of the integral on the right-hand side of \eqref{eq:integral}, we integrate $\lvert\charfn_{a,b}(t)\rvert$ over the interval $[0,2\pi/(a+b)]$. For the sake of convenience, we substitute $t$ with $2t$ in order to avoid the denominator of $2$ in the sine terms of the integrand. 

\begin{lemma}[center bound]
    \label{lemma:center}
    For integers $b\ge a\ge 2$,
    \[ \int_0^{\frac{2\pi}{a+b}} \lvert\charfn_{a,b}(t)\rvert\, dt = 2 \int_0^{\frac{\pi}{a+b}} \lvert\charfn_{a,b}(2t)\rvert\, dt = O\left(\frac{1}{b\sqrt{a}}\right). \]    
\end{lemma}

We can write 
\[\lvert\charfn_{a,b}(2t)\rvert=\prod_{k=1}^{a} \frac{k}{b+k}\cdot\frac{\sin((b+k)t)}{\sin(kt)} \]
as every term in the product on the right-hand side is nonnegative on this interval. We start with the following estimates.

\begin{claim}
    \label{claim:center3}
    For positive integers $k\le b$ and $x\in [0,\pi/(b+k)]$,
    \[\frac{k}{b+k}\cdot \frac{\sin((b+k)x)}{\sin(kx)}\le 1-\frac{b^2}{2\pi^2}x^2. \]
\end{claim}

To prove this claim, we first prove two trigonometric bounds. Refer to \cref{fig:trigplots} for a plot of the functions and bounds.

\begin{claim}
    \label{claim:center1}
    For positive integers $k$, and $x\in[0,\pi/k]$,
    \[\sin(kx)\le kx-\frac{(kx)^3}{\pi^2}. \]
\end{claim}

\begin{proof}
    Let $y=kx$. It is enough to argue that $\sin(y)\ge y-\frac{y^3}{\pi^2}$ in the range $y\in[0,\pi]$.
    
    Let $f(y)=\sin(y)-y+\frac{y^3}{\pi^2}$. Notice that $f(0)=f(\pi)=0$ and $f'(\pi)>0$. We will argue that there is a unique point $y^*\in (0,\pi)$ such that $f'(y^*)=0$, which will ensure that $f(y)\le 0$ for all $y\in [0,\pi]$.
    
    We can calculate that 
    \begin{align*}
        f'(y) &= \cos(y)-1+\frac{3y^2}{\pi^2}\\
        &= \frac{3y^2}{\pi^2}-2\sin^2\left(\frac{y}{2}\right).
    \end{align*}
    So $f'(y)=0$ if and only if $\sin(y/2)= \pm (\sqrt{6}/\pi)\cdot (y/2)$, which is satisfied by one unique point $y^*\in(0,\pi)$.
\end{proof}

\begin{claim}
    \label{claim:center2}
    For positive integers $k$, and $x\in (0,\pi/2k]$, 
    \[\cot(kx)\le \frac{1}{kx}. \]
\end{claim}

\begin{proof}
    Let $y=kx$. We will prove that $\cot(y)\le \frac{1}{y}$ for all $y\in (0,\pi/2]$. It is enough to prove that $y\cos(y)\le \sin(y)$ for all $y\in [0,\pi/2]$.
    
    The latter inequality follows because both sides are zero when $y=0$, and the derivative of the left hand side is bounded above by the derivative of the right hand side when $y\in [0,\pi/2]$:
    \[\cos(y)-y\sin(y)\le \cos(y).  \]
\end{proof}

\begin{figure}[ht]
    \centering
    \begin{subfigure}[b]{0.4\textwidth}
        \centering
        \begin{tikzpicture}
            \draw[very thin, color=gray] (-0.1,-0.1) grid (3.9,3.9);
            \draw[->] (-0.1,0)--(3.9,0) node[right]{$x$};
            \draw[->] (0,-0.1)--(0,3.9) node[above]{$y$};
            
            \draw[densely dotted, domain=0:3.1415926, variable=\x, thick] plot (\x,{sin(\x r)});
            \draw[dashed, domain=0:3.1415926, variable=\x, thick] plot (\x,{\x-(\x^3/9.8696)});
            \node at (3.1415926,-0.1)[anchor=north]{$\pi$};
        \end{tikzpicture}
        \caption{\cref{claim:center1}}
    \end{subfigure}
    \begin{subfigure}[b]{0.4\textwidth}
        \centering
        \begin{tikzpicture}
            \draw[very thin, color=gray] (-0.1,-0.1) grid (3.9,3.9);
            \draw[->] (-0.1,0)--(3.9,0) node[right]{$x$};
            \draw[->] (0,-0.1)--(0,3.9) node[above]{$y$};
            
            \draw[densely dotted, domain=0.5:3.1415926, variable=\x, thick] plot (\x,{cot(\x/2 r)});
            \draw[dashed, domain=0.5:3.1415926, variable=\x, thick] plot (\x,{2/\x});
            \node at (3.1415926,0)[anchor=north]{$\pi/2$};
        \end{tikzpicture}
        \caption{\cref{claim:center2}}
    \end{subfigure}
    \caption{Plots of Trigonometric Bounds. \\ Trigonometric functions are dotted, upper bounds are dashed.}
    \label{fig:trigplots}
\end{figure}
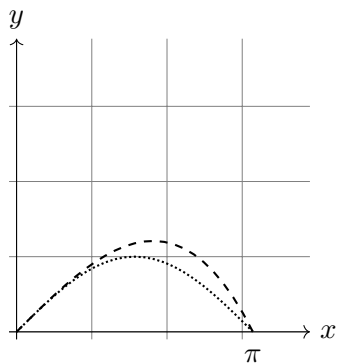
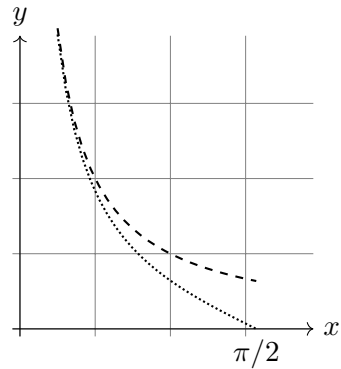

Now we finish the proof of \cref{claim:center3}.

\begin{proof}
    Notice that  
    \[\frac{\sin((b+k)x)}{\sin(kx)}=\frac{\sin(kx)\cos(bx)+\sin(bx)\cos(kx)}{\sin(kx)} \]
    \[= \cos(bx)+\cot(kx)\sin(bx). \]
    
    Of course, $\cos(bx)\le 1$. Furthermore, from \cref{claim:center1} and \cref{claim:center2}, we can see that in this domain of $x$, $\sin(bx)\le bx-\frac{b^3x^3}{\pi^2}$ and $\cot(kx)\le \frac{1}{kx}$. Additionally, $\sin(bx)\ge 0$. Therefore, using the fact that $k\le b$,
    \[\frac{k}{b+k}\cdot \frac{\sin((b+k)x)}{\sin(kx)}\le \frac{k}{b+k}\left(1+\frac{1}{kx}\left(bx-\frac{b^3x^3}{\pi^2}\right)\right)\]
    \[\le 1-\frac{b^3}{(b+k)\pi^2}x^2 \le 1-\frac{b^2}{2\pi^2}x^2. \]
\end{proof}

From this, we can now prove \cref{lemma:center}.

\begin{proof}
    Recall that on this interval
    \[\lvert\charfn_{a,b}(2t)\rvert=\prod_{k=1}^{a} \frac{k}{b+k}\cdot\frac{\sin((b+k)t)}{\sin(kt)}.\]
    
\cref{claim:center3} applies on all $t$ in the domain because $\pi/(b+a)\le \pi/(b+k)$ for all $k$.
    
    Therefore,
    \[\int_0^{\frac{\pi}{a+b}} \lvert\charfn_{a,b}(2t)\rvert\,dt \le \int_0^{\frac{\pi}{a+b}} \left(1-\frac{b^2t^2}{2\pi^2}\right)^a\,dt \le \int_0^{\frac{\pi}{a+b}} \exp\left(-\frac{ab^2t^2}{2\pi^2}\right)\,dt =O\left(\frac{1}{b\sqrt{a}}\right). \]
    Here, we use the fact that $1-x\le \exp(-x)$ for all $x$, and the Gaussian integral: the integral of $\exp(-t^2)$ over $\mathbb{R}$ is constant, and by scaling the argument, the integral of $\exp(-ct^2)$ over $\mathbb{R}$ is a constant factor of $c^{-1/2}$ for any parameter $c$. 
\end{proof}

\subsection{Peripheral bound}
\label{sec:peripheral}

We now bound $\lvert\charfn_{a,b}(t)\rvert$ in the region away from $0$, namely, the interval $[2\pi/(a+b),\pi]$. As for the other part of the integral on the right-hand side of \eqref{eq:integral}, we substitute $t$ with $2t$ in order to avoid the denominator of $2$ in the sine terms of the integrand. 

\begin{lemma}[peripheral bound]
    \label{lemma:far}
    For integers $b\ge a\ge 2$,
    \[ \int_{\frac{2\pi}{a+b}}^\pi \lvert\charfn_{a,b}(t)\rvert\, dt =  2 \int_{\frac{\pi}{a+b}}^{\frac{\pi}{2}} \lvert\charfn_{a,b}(2t)\rvert\, dt = O\left(\frac{1}{b\sqrt{a}}\right). \]    
\end{lemma}

In this region, the main problem is that the denominator of $\lvert\charfn_{a,b}(2t)\rvert$ often goes to zero, which could potentially blow up the integrand. However, the terms in the numerator always cancel out these blowups. The following lemma will be our main tool for bounding $\lvert\charfn_{a,b}(2t)\rvert$ in this region; it will allow terms in the numerator to cancel out bad terms in the denominator.

\begin{lemma}[pole reduction]
    \label{lemma:bijection}
    For every $t\in\mathbb{R}$, there exists a bijection $\bij_t:\{1,\dots,a\}\to\{b+1,\dots,b+a\}$ (depending on $t$) such that for every $k=1,\dots,a$,
    \[\left|\frac{1}{k}\sin(kt)\right| \ge \left|\frac{1}{\bij_t(k)}\sin(\bij_t(k)t)\right|. \]
\end{lemma}

A basic bound can be found by applying \cref{lemma:bijection} on $\lvert\charfn_{a,b}(2t)\rvert$ for $k=2,\dots,a$, resulting in an upper bound of $1/b\sin(t)$. This bound is usable, but is weak on points closer to 0. To remedy this, we can divide the domain of integration into multiple parts, where the points closer to 0 can safely include more terms in the denominator.

Let $2\le n\le a$ be an integer. We split the domain of integration into three intervals: $[\frac{\pi}{a+b}, \frac{\pi}{2n}]$, $[\frac{\pi}{2n},  \frac{\pi}{2}- \frac{\pi}{2n}]$, and $[\frac{\pi}{2}-\frac{\pi}{2n},\frac{\pi}{2}]$. By selecting a good value of $n$, we can get a reasonable upper bound on this region.
Here, we will make use of \cref{lemma:bijection} and the following linear approximation to sine: 

\begin{fact}
    \label{fact:linearlowerbound}
    For a positive integer $k$ and $t\in [0, \pi/2k]$,
    \[\sin(kt)\ge \frac{2kt}{\pi }. \]
\end{fact}

\begin{proof}
    Notice that $\sin(kt)=\frac{2kt}{\pi}$ when $t=0$ and $t=\frac{\pi}{2k}$. This fact then follows since sine is concave on this interval.     
\end{proof}

\paragraph{Region I.}

The first region of integration is $[\pi/(a+b),\pi/(2n)]$.

\begin{align*}
    \int_{\frac{\pi}{a+b}}^{\frac{\pi}{2n}} \lvert\charfn_{a,b}(2t)\rvert\,dt &\le  \int_{\frac{\pi}{a+b}}^{\frac{\pi}{2n}} \frac{n!}{\bij_t(1)\cdots \bij_t(n)} \left\lvert \frac{\sin(\bij_t(1)t)\cdots \sin(\bij_t(n)t)}{\sin(t)\cdots \sin(nt)}\right\rvert \, dt \\
    &\le \frac{n!}{b^n}\int_{\frac{\pi}{a+b}}^{\frac{\pi}{2n}} \left\lvert\frac{1}{\sin(t)\cdots \sin(nt)} \right\rvert\,dt \\
    &\le \frac{1}{b^n}\int_{\frac{\pi}{a+b}}^{\frac{\pi}{2n}} \left(\frac{\pi}{2}\right)^n \frac{1}{t^n}\,dt\\
    &\le \frac{1}{b^n}\cdot \left(\frac{\pi}{2}\right)^n \cdot\frac{1}{n-1}\cdot \left(\frac{a+b}{\pi}\right)^{n-1}\\
    &\le \frac{1}{b^n}\cdot \left(\frac{\pi}{2}\right)^n \cdot\frac{1}{n-1}\cdot \left(\frac{2b}{\pi}\right)^{n-1}\\
    &\le \frac{\pi}{2b(n-1)}.
\end{align*}

The first two steps involve applying \cref{lemma:bijection} on all $k \in \{n+1,\dots,a\}$, and then using the fact that $\lvert\sin(x)\rvert\le 1$ and $\bij_t(k)\ge b$. The third step uses \cref{fact:linearlowerbound} for $k \in \{1,\dots,n\}$, which applies on the interval $[\pi/(a+b),\pi/(2n)]$ for these values of $k$. From there, we bound with the left limit of integration.

\paragraph{Region II.}

We now bound $\lvert\charfn_{a,b}(2t)\rvert$ on the interval $[\frac{\pi}{2n},\frac{\pi}{2}-\frac{\pi}{2n}]$. We can use \cref{lemma:bijection} to eliminate all terms except $k \in \{1,2\}$ this time.
\begin{align*}
    \int_{\frac{\pi}{2n}}^{\frac{\pi}{2}-\frac{\pi}{2n}} \lvert\charfn_{a,b}(2t)\rvert\,dt &\le \int_{\frac{\pi}{2n}}^{\frac{\pi}{2}-\frac{\pi}{2n}} \frac{2}{\bij_t(1)\bij_t(2)} \left\lvert\frac{\sin(\bij_t(1)t)\sin(\bij_t(2)t)}{\sin(t)\sin(2t)} \right\rvert\,dt\\
    &\le \frac{2}{b^2}\int_{\frac{\pi}{2n}}^{\frac{\pi}{2}-\frac{\pi}{2n}} \left\lvert\frac{1}{\sin(t)\sin(2t)} \right\rvert\,dt.
\end{align*}

Because $\lvert\sin(t)\rvert$ is increasing here and $\lvert\sin(2t)\rvert$ is symmetric about $\frac{\pi}{4}$, the value of the integral on the interval $[\frac{\pi}{2n},\frac{\pi}{4}]$ exceeds the value on the interval $[\frac{\pi}{4},\frac{\pi}{2}-\frac{\pi}{2n}]$. Using \cref{fact:linearlowerbound},
\[\frac{2}{b^2}\int_{\frac{\pi}{2n}}^{\frac{\pi}{4}} \left\lvert\frac{1}{\sin(t)\sin(2t)} \right\rvert\,dt
\le \frac{1}{b^2}\int_{\frac{\pi}{2n}}^{\frac{\pi}{4}} \left(\frac{\pi}{2}\right)^2 \frac{1}{t^2}\,dt
\le  \frac{\pi^2}{4b^2}\cdot \frac{2n}{\pi}
=\frac{\pi n}{2b^2}.\]
We have now established that
\[\int_{\frac{\pi}{2n}}^{\frac{\pi}{2}-\frac{\pi}{2n}} \lvert\charfn_{a,b}(2t)\rvert\,dt \le \frac{\pi n}{b^2}.\]

\paragraph{Region III.}

Notice that $\lvert\sin(t)\rvert$ is increasing on the interval $[\frac{\pi}{2}-\frac{\pi}{2n},\frac{\pi}{2}]$, so we can bound $\lvert\sin(t)\rvert$ by $\lvert\sin(\frac{\pi}{2}-\frac{\pi}{2n})\rvert$. Similar to before, the first step follows from \cref{lemma:bijection}, this time applied to $k \in \{2,\dots,a\}$. 
\begin{align*}
    \int_{\frac{\pi}{2}-\frac{\pi}{2n}}^{\frac{\pi}{2}} \lvert\charfn_{a,b}(2t)\rvert\,dt 
    &\le \int_{\frac{\pi}{2}-\frac{\pi}{2n}}^{\frac{\pi}{2}} \frac{1}{\bij_t(1)} \left\lvert\frac{\sin(\bij_1(t))}{\sin(t)}\right\rvert\,dt
    \le \int_{\frac{\pi}{2}-\frac{\pi}{2n}}^{\frac{\pi}{2}} \frac{1}{|b \sin(t)|}\,dt\\
    &\le \frac{\pi}{2n}\cdot \frac{1}{b\sin(\frac{\pi}{2}-\frac{\pi}{2n})}
    \le \frac{\pi}{2n}\frac{1}{b(1-\frac{1}{n})}
    = \frac{\pi}{2b(n-1)}.
\end{align*}

\paragraph{Overall bound.}

Summing the above bounds, we can deduce that
\[\int_{\frac{\pi}{a+b}}^{\frac{\pi}{2}} \lvert\charfn_{a,b}(2t)\rvert\,dt \le \frac{\pi}{b(n-1)}+\frac{\pi n}{b^2}.  \]
By choosing $n=\lceil\sqrt{a} \,\rceil$ (keeping in mind that $b\ge a \ge 2$), we can deduce \cref{lemma:far}.

\subsection{Pole reduction}
\label{sec:pole-reduction}

Finally, we prove the pole reduction lemma (\cref{lemma:bijection}). The essence is an interval matching strategy capture \cref{lemma:intervalmatching}. Here is the intuition. 

Recall that we want to upper bound factors of the form $|\frac{\sin(\ell t)}{\ell}|$ by rescaled versions $|\frac{\sin(k t)}{k}|$ of the same pattern, where $\ell \in \{b+1,\dots,b+a\}$ and $k \in \{1,\dots,a\}$ are matched. The matching definitely needs to avoid situations like in Figure~\ref{fig:interval:bad}, where $|\frac{\sin(kt)}{k}|$ vanishes at the point $t$ while $|\frac{\sin(\ell t)}{\ell}|$ does not. Ideally, the period of the $k$-scaled version that contains the point $t$ encloses the period of the $\ell$-scaled version that contains $t$, like in Figure~\ref{fig:interval:good}. As long as the pattern is convex, this ensures that the $k$-scaled version is larger than the $\ell$-scaled version everywhere on the encompassed period. (We will formally prove this in Claim~\ref{claim:bijection-proof}.) Thus, if at every point $t$, we can set up a matching such that the enclosing relationship holds for all matched pairs, we are home free. Lemma~\ref{lemma:intervalmatching} below does exactly this.

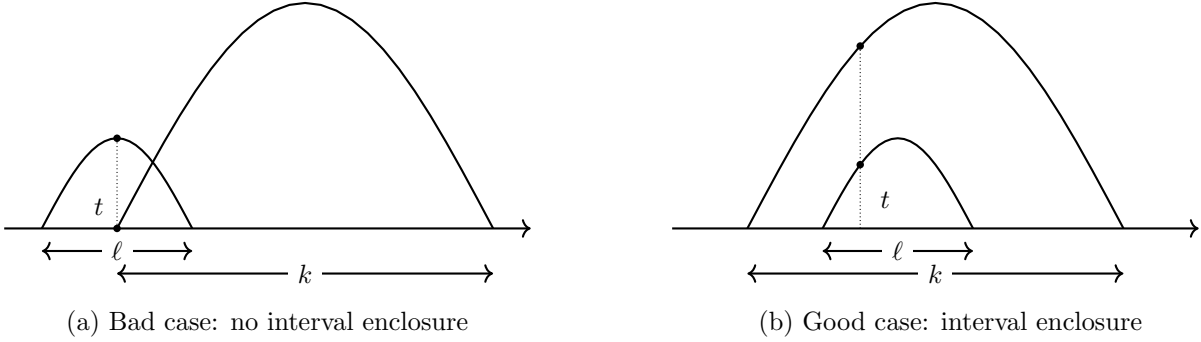
\begin{figure}[ht]
    \centering
    \begin{subfigure}[b]{0.45\textwidth}
        \centering
        \begin{tikzpicture}
            \draw[->,thick] (-0.5,0)--(6.5,0);
            
            \draw[domain=1:6,variable=\t,thick] plot ({\t}, {3*sin((\t-1)*36)});
            \draw[domain=0:2,variable=\t,thick] plot ({\t}, {1.2*sin((\t)*90)});
            
            \draw[<->, thick] (0,-0.3) --
                node[pos=0.5,style={fill=white}] {\small $\ell$}
            (2,-0.3);
            \draw[<->, thick] (1,-0.6) --
                node[pos=0.5,style={fill=white}] {\small $k$}
            (6,-0.6);
            
            \node at (1,0)[circle,fill,inner sep=1pt,label={above left:\small $t$}]{};
            \node at (1,{1.2*sin((1)*90)})[circle,fill,inner sep=1pt]{};
            \draw[densely dotted] (1,0)--(1,{1.2*sin((1)*90)});
        \end{tikzpicture}
        \caption{Bad case: no interval enclosure}\label{fig:interval:bad}
    \end{subfigure}
    \hfill
    \begin{subfigure}[b]{0.45\textwidth}
        \begin{tikzpicture}
            \draw[->,thick] (0,0)--(7,0);
            
            \draw[domain=1:6,variable=\t,thick] plot ({\t}, {3*sin((\t-1)*36)});
            \draw[domain=2:4,variable=\t,thick] plot ({\t}, {1.2*sin((\t-2)*90)});
            
            \draw[<->, thick] (2,-0.3) --
                node[pos=0.5,style={fill=white}] {\small $\ell$}
            (4,-0.3);
            \draw[<->, thick] (1,-0.6) --
                node[pos=0.5,style={fill=white}] {\small $k$}
            (6,-0.6);
            
            \node at (2.5,0)[label={above right:\small $t$}]{};
            \node at (2.5,{3*sin((2.5-1)*36)})[circle,fill,inner sep=1pt]{};
            \node at (2.5,{1.2*sin((2.5-2)*90)})[circle,fill,inner sep=1pt]{};
            \draw[densely dotted] (2.5,0)--(2.5,{3*sin((2.5-1)*36)});
        \end{tikzpicture}
        \caption{Good case: interval enclosure}\label{fig:interval:good}
    \end{subfigure}
    \caption{Enclosing intervals are needed for \cref{lemma:bijection}.}
    \label{fig:goodbad}
\end{figure}

Let us first introduce some notation. For every real number $t$ and positive integer $k$, there is a unique integer $n$ such that $t$ is contained in the half-open interval $[n/k, (n+1)/k)$. We call this interval the \emph{$k$-interval} of $t$. For positive integers $k,\ell$, we can say that the $k$-interval of $t$ \emph{encloses} the $\ell$-interval of $t$ if the $\ell$-interval of $t$ is a subset of the $k$-interval of $t$. We use the shorthand that $k$ encloses $\ell$ at $t$.


\begin{lemma}[interval matching]
    \label{lemma:intervalmatching}
    Let $a,b$ be positive integers. For any real $t$, there exists a bijection $\bij_t$ between $\{1,\dots,a\}$ and $\{b+1,\dots,b+a\}$ such that for all $k=1,\dots,a$, $k$ encloses $\ell=\bij_t(k)$ at $t$.
\end{lemma}
Note that the bijection can be different depending on $t$. In fact, this is necessary as otherwise $\ell$ would need to be a multiple of $k$, which is not possible with a bijection between the sets $\{1,\dots,a\}$ and $\{b+1,\dots,b+a\}$.

We can interpret \cref{lemma:intervalmatching} as a matching on a bipartite graph by using Hall's marriage lemma.

\begin{lemma}[Hall's marriage lemma]
    \label{thm:hall}
    Let $G$ be a bipartite graph with partitions $L,R$. For any $\A\subseteq L$, let $\Nbr(\A)$ be the set of all vertices in $R$ with at least one neighbor in $\A$. The graph $G$ admits a perfect matching if and only if for all such $\A$, $|\Nbr(\A)|\ge |\A|$. 
\end{lemma}

For $\A\subseteq \{1,\dots,a\}$, let $\Nbr_t(\A)$ be the set of all $\ell\in\{b+1,\dots,b+a\}$ such that there exists $k\in \A$ where $k$ encloses $\ell$ at $t$. To produce the desired bijection $\bij_t$ in \cref{lemma:intervalmatching}, it is sufficient to prove that $|\Nbr_t(\A)|\ge |\A|$ for all $t$ and $\A$.

There are some values of $\ell$ that are always contained in $\Nbr_t(\A)$ regardless of the value of $t$. Let $\Nbr(\A)$ be the set of all $\ell\in\mathbb{N}$ such that for all $t$, there exists $k\in \A$ where $k$ encloses $\ell$ (this $k$ can vary depending on $t$). As $\Nbr(\A)\cap \{b+1,\dots,b+a\}\subseteq \Nbr_t(\A)$, it is sufficient to prove that $|\Nbr(\A)\cap \{b+1,\dots,b+a\}|\ge |\A|$ for all $\A$ and apply \cref{thm:hall} to prove \cref{lemma:intervalmatching}.

\begin{example}
    \label{ex:intervalmatching}
    $5\in \Nbr(\{2,3\}).$
\end{example}

\begin{proof}
We only consider $t\in[0,1)$ for clarity. When $t\in [0,0.4)$, the $2$-interval for $t$ is $[0,0.5)$, while the $5$-interval for $t$ is either $[0,0.2)$ or $[0.2,0.4)$, which means $2$ encloses $5$. When $t\in[0.4,0.6)$, the $3$-interval for $t$ is $[1/3,2/3)$, which encloses the $5$-interval $[0.4,0.6)$. When $t\in [0.6,1)$, the $2$-interval is $[0.5,1)$ while the $5$-interval is either $[0.6,0.8)$ or $[0.8,1)$, so $2$ encloses $5$. These enclosures are shown in \cref{fig:intervalmatching}, where the marked $5$-intervals are contained within the respectively marked $2$ or $3$-intervals.

For all $t$, either 2 encloses 5, or 3 encloses 5. In other words, $5\in \Nbr(\{2,3\})$.
\end{proof}

\begin{figure}[ht]
    \centering
    \begin{tikzpicture}
        \node at (0,0)[circle,fill,inner sep=1.5pt]{};
        \node at (2,0)[circle,fill,inner sep=1.5pt]{};
        \node at (4,0)[circle,fill,inner sep=1.5pt]{};
        \node (x5y6) at (6,0)[circle,fill,inner sep=1.5pt]{};
        \node (x5y8) at (8,0)[circle,fill,inner sep=1.5pt]{};
        \node (x5y10) at (10,0)[circle,fill,inner sep=1.5pt]{};
        \draw[thick, dashed] (0,0) --(4,0);
        \draw[thick, densely dotted] (4,0)--(6,0);
        \draw[thick, decorate, decoration=snake] (x5y6)--(x5y8);
        \draw[thick, decorate, decoration=snake] (x5y8)--(x5y10);
        \node at (10.5,0)[anchor=west]{$\ell=5$};
        
        \node at (0,1)[circle,fill,inner sep=1.5pt]{};
        \node at (10/3,1)[circle,fill,inner sep=1.5pt]{};
        \node at (20/3,1)[circle,fill,inner sep=1.5pt]{};
        \node at (10,1)[circle,fill,inner sep=1.5pt]{};
        \draw (0,1)--(10/3,1);
        \draw (20/3,1)--(10,1);
        \draw[thick, densely dotted] (10/3,1)--(20/3,1);
        \node at (10.5,1)[anchor=west]{$k=3$};
        
        \node at (0,2)[circle,fill,inner sep=1.5pt]{};
        \node (x2y5) at (5,2)[circle,fill,inner sep=1.5pt]{};
        \node (x2y10) at (10,2)[circle,fill,inner sep=1.5pt]{};
        \draw[thick, dashed] (0,2)--(5,2);
        \draw[thick, decorate, decoration=snake] (x2y5)--(x2y10);
        \node at (10.5,2)[anchor=west]{$k=2$};
    \end{tikzpicture}
    \caption{\cref{ex:intervalmatching}}
    \label{fig:intervalmatching}
\end{figure}

We first consider a different characterization of $\Nbr(\A)$.

\begin{claim}
    \label{claim:intervalcondition}
    Let $\ell$ be a positive integer. Then $\ell\in \Nbr(\A)$ if and only if every open interval $I\subset \mathbb{R}$ that contains a fraction of denominator $k$ for every $k\in \A$ must also contain a fraction of denominator $\ell$. (These fractions do not have to be distinct or reduced.)
\end{claim}

The idea is that if no $k\in \A$ encloses $\ell$ for some $t$, then the endpoints of each $k$-interval form a fraction of denominator $k$ contained strictly within the $\ell$-interval of $t$. As a result, we have a contiguous interval $I$ containing a fraction of denominator $k$, and $I$ is contained strictly within the $\ell$-interval of $t$. As such, $I$ cannot contain a fraction of denominator $\ell$. This is illustrated in \cref{fig:intervalcondition}. The formal proof also considers the edge cases involving the endpoints of the intervals.

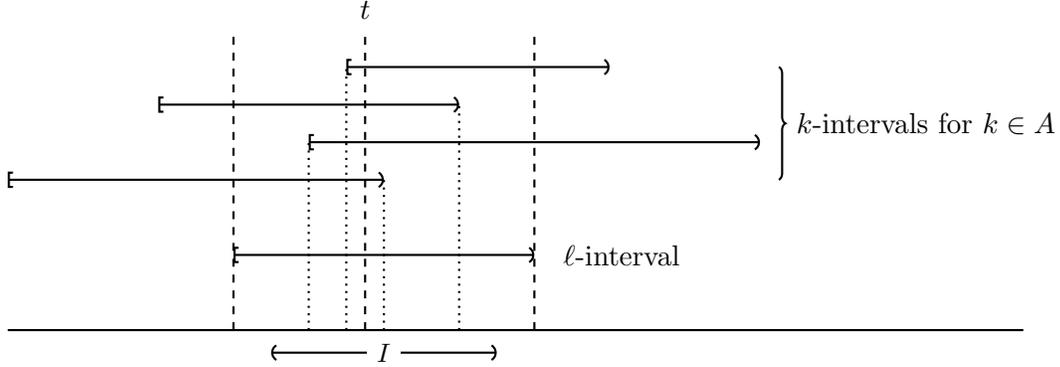
\begin{figure}[ht]
    \centering
    \begin{tikzpicture}
        \draw[thick] (0,0)--(13.5,0);
        
        \draw[dashed, thick] (3,0)--(3,4);
        \draw[dashed, thick] (7,0)--(7,4);
        \draw[dotted, thick] (5,0)--(5,2);
        \draw[dotted, thick] (4,0)--(4,2.5);
        \draw[dotted, thick] (6,0)--(6,3);
        \draw[dotted, thick] (4.5,0)--(4.5,3.5);
        
        \draw[dashed, thick] (4.75,0)--(4.75,4);
        \node at (4.75,4)[anchor=south]{$t$};
        
        \draw[Bracket-Parenthesis, thick] (3,1)--(7,1);
        
        \draw[Bracket-Parenthesis, thick] (0,2)--(5,2);
        \draw[Bracket-Parenthesis, thick] (4,2.5)--(10,2.5);
        \draw[Bracket-Parenthesis, thick] (2,3)--(6,3);
        \draw[Bracket-Parenthesis, thick] (4.5,3.5)--(8,3.5);
        
        \draw[decorate, decoration=brace, thick] (10.25,3.5)--(10.25,2);
        \node at (10.35,2.75)[anchor=west]{$k$-intervals for $k\in A$};
        \node at (7.25,1)[anchor=west]{$\ell$-interval};
        
        \draw[Parenthesis-Parenthesis, thick] (3.5,-0.3) -- 
            node[pos=0.5,style={fill=white}] {\small $I$}
        (6.5,-0.3);
    \end{tikzpicture}
    \caption{$I$ contains a fraction of denominator $k$ for all $k$, but no fraction of denominator $\ell$.}
    \label{fig:intervalcondition}
\end{figure}

For \cref{claim:intervalcondition},  we prove the following proposition.
\begin{proposition}
    \label{prop:intervalcondition-1}
    The following statements are equivalent for any half-open interval $[c,d)$ and positive integer $k$:
    \begin{enumerate}[(1)]
        \item $[c,d)$ is contained within a $k$-interval.
        \item $(c,d)$ is contained within a $k$-interval.
        \item $(c,d)$ contains no fraction of denominator $k$.
    \end{enumerate}
\end{proposition}

\begin{proof} $(1) \implies (2)$. This follows as $(c,d)$ is a subset of $[c,d)$.
        
        $(2) \implies (3)$. Suppose $(c,d)$ is contained in the $k$-interval $[n/k, (n+1)/k)$. The interval $(c,d)$ cannot contain a fraction of denominator $k$, otherwise said fraction would be strictly between $n/k$ and $(n+1)/k$.
        
        $(3) \implies (1)$. Let $n/k$ be the largest fraction of denominator $k$ less than or equal to $c$. It must be true that $(n+1)/k\ge d$, since $(n+1)/k$ cannot be contained in $(c,d)$. Therefore, $[c,d)$ is contained in the $k$-interval $[n/k, (n+1)/k)$.
\end{proof}

From this, we can prove \cref{claim:intervalcondition}. 

\begin{proof}
    By definition, the condition that $\ell\in N(A)$ is that any $\ell$-interval $J$ is contained in some $k$-interval for some $k\in A$. By condition (2) in \cref{prop:intervalcondition-1}, this is equivalent to saying $\Int(J)$ is contained in some $k$-interval, where $\Int(J)$ is the interior of $J$. This is true if and only if any open subinterval $I$ of $\Int(J)$ is contained in some $k$-interval. Equivalently, if $I$ is an open interval that is not contained in any $k$-interval, then $I$ is not contained in any $\ell$-interval. Using condition (3) in \cref{prop:intervalcondition-1}, this is finally equivalent to the condition that if $I$ contains a fraction of denominator $k$ for every $k\in \A$, then $I$ contains a fraction of denominator $\ell$.
\end{proof}

The characterization in \cref{claim:intervalcondition} allows us to prove the following key claim about $\Nbr(\A)$.

\begin{claim}
    \label{claim:sumcondition}
    If $\ell_1,\ell_2\in \Nbr(\A)$, then $\ell_1+\ell_2\in \Nbr(\A)$.
\end{claim}

\begin{proof}
    By \cref{claim:intervalcondition}, any interval $I$ that contains a fraction of denominator $k$ for every $k\in \A$ must also contain two fractions $x/\ell_1$ and $y/\ell_2$. For positive integers $a,b,c,d$, define the \emph{mediant} of $a/b$ and $c/d$ to be $(a+c)/(b+d)$. Then the mediant is always between the two fractions. In other words, if $a/b\le c/d$,
    \[\frac{a}{b}\le \frac{a+c}{b+d}\le \frac{c}{d}. \]
    This fact can be proven with elementary algebra, as both sides are equivalent to $a/b\le c/d$.
    
    From this, we see that the mediant $(x+y)/(\ell_1+\ell_2)$ is contained in $I$, as it is between two elements of $I$. As this applies to every such $I$, we can use \cref{claim:intervalcondition} to conclude that $\ell_1+\ell_2\in \Nbr(\A)$.
    
    Note that $x/\ell_1$ and $y/\ell_2$ do not have to be distinct. $\ell_1$ and $\ell_2$ might be equal, or $x/\ell_1=y/\ell_2$.
\end{proof}


As $k$ encloses $k$ for every $t$, we have that $\A\subseteq \Nbr(\A)$. Let $\Sum(\A)$ be the set of positive integers that can be written as the sum of not necessarily distinct elements of $\A$. \cref{claim:sumcondition} implies that $\Sum(\A)\subseteq \Nbr(\A)$.

\begin{claim}
    For nonnegative integers $n$, \[\lvert\Sum(\A)\cap \{n+1,\dots,n+a\}\rvert\ge |\A|.\]
\end{claim}

\begin{proof}
    Let $m=\max(\A)$. Because $\A\subseteq \Sum(\A)$, $\Sum(\A)$ contains every positive integer congruent to an element of $\A$ modulo $m$, since we can repeatedly add $m$ to any element of $\A$. As $a\ge m$, the set $\{n+1,\dots,n+a\}$ contains at least one element for every residue mod $m$. Therefore, $\Sum(\A)$ contains at least $|\A|$ elements of $\{n+1,\dots,n+a\}$. 
\end{proof}

Putting it all together, we have that
\[|\A|\le \lvert\Sum(\A)\cap\{b+1,\dots,b+a\}\rvert \le |\Nbr(\A)\cap \{b+1,\dots,b+a\}|\le |\Nbr_t(\A)|, \]
which proves \cref{lemma:intervalmatching} by our previous reasoning. To finish, we need to show that this implies \cref{lemma:bijection}.

\begin{claim}
    \label{claim:bijection-proof}
    \cref{lemma:intervalmatching} implies \cref{lemma:bijection}.
\end{claim}

\begin{proof}
    We prove the following more general claim. Let $f:\mathbb{R}\to\mathbb{R}$ be a function that has period 1, and additionally, $f$ is concave on $[0,1]$ and $f(0)=f(1)=0$. For $t\in [0,1]$, let $\bij_t$ be a bijection that satisfies the conditions of \cref{lemma:intervalmatching}. Then for any $k\in \{1,\dots,a\}$, let $\ell=\bij_t(k)$. We seek to prove that 
    \[\frac{f(kt)}{k}\ge \frac{f(\ell t)}{\ell} \]
    for all $k=1,\dots,a$. Notice that $f(t)=\lvert\sin(t)\rvert$ satisfies all the conditions of the claim, albeit after scaling the period from $\pi$ to $1$.
    
    To prove the general claim, let $\kappa=kt-\lfloor kt\rfloor$ and $\lambda=\ell t-\lfloor \ell t\rfloor$, noting that $\kappa,\lambda\in [0,1]$. By the enclosing property, it is true that $\kappa/k \ge \lambda/\ell$, since they represent the distance from $t$ to the left endpoints of the $k$ and $\ell$-intervals, respectively, and $k$ encloses $\ell$ at $t$. This can be seen in \cref{fig:enclosingpicture}.
    
    Suppose $\kappa\le \lambda$. We have 
    \[f(kt)=f(\kappa) \ge \frac{\kappa}{\lambda}f(\lambda)\ge \frac{k}{\ell}f(\lambda)=\frac{k}{\ell}f(\ell t). \]
    The first step follows from the periodicity of $f$, the second from the concavity of $f$ on $[0,1]$ and the fact that $f(0)=0$ (the point $(\kappa,f(\kappa))$ is above the line segment connecting $(\lambda,f(\lambda))$ and $(0,0)$), the third from the aforementioned inequality $\kappa/k \ge \lambda/\ell$, and the last from the periodicity of $f$ again. 
    
    If $\lambda<\kappa$, we can instead consider the function $\widetilde{f}(t)=f(1-t)$. Here, $\widetilde{\kappa}=1-\kappa$ and $\widetilde{\lambda}=1-\lambda$, and we can use our previous reasoning. This proves the general claim.
\end{proof}

\begin{figure}[ht]
    \centering
    \begin{tikzpicture}[scale=0.8]
        \draw[->,thick] (-1,0)--(11,0) node[anchor=west]{$x$};
        \draw[domain=0:10,variable=\t,thick] plot ({\t}, {5*sin((\t)*18)});
        \draw[domain=1:5,variable=\t,thick] plot({\t}, {2*sin((\t-1)*45)});
        \draw[densely dotted] ({1+0.4*4},-1)--({1+0.4*4},{5*sin((1+0.4*4)*18)});
        \draw[densely dotted] (0,-1)--(0,0);
        \draw[densely dotted] (1,-0.5)--(1,0);
        
        \draw[<->,thick] (1,-0.5) -- 
            node[pos=0.5,style={fill=white}] {\small $\lambda/\ell$}
        (1+0.4*4,-0.5);
        \draw[<->,thick] (0,-1) -- 
            node[pos=0.5,style={fill=white}] {\small $\kappa/k$}
        (1+0.4*4,-1);
        
        \node at (4.25,2.25){$\dfrac{f(\ell x)}{\ell}$};
        \node at (8,4.25){$\dfrac{f(k x)}{k}$};
        
        \node at ({1+0.4*4},0)[anchor=south west]{$t$};
        \node at ({1+0.4*4},{2*sin((0.4*4)*45)})[circle,fill,inner sep=1.5pt]{};
        \node at ({1+0.4*4},{5*sin((1+0.4*4)*18)})[circle,fill,inner sep=1.5pt]{};
    \end{tikzpicture}
    \caption{Proof of \cref{claim:bijection-proof}: $\lambda/\ell \le \kappa/k$}
    \label{fig:enclosingpicture}
\end{figure}
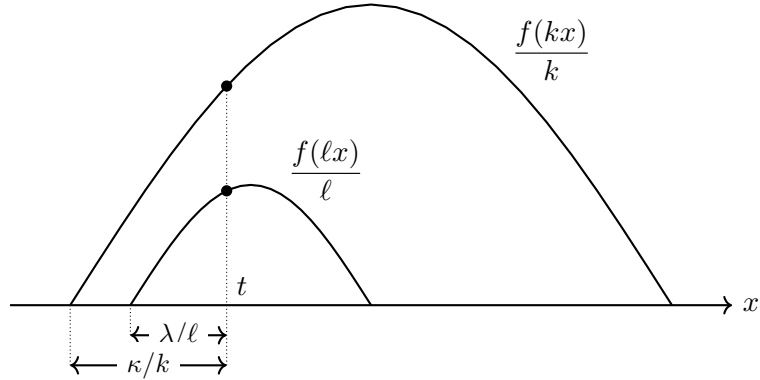

\section{Turing Complexity}
\label{sec:turing}

This section serves as a supplementary discussion on the time (and space) complexity of algorithms for computing the minimum number of inversions in trees. From the analysis in \cref{sec:decomposition}, computing $\MInv(T,\ranking)$ involves computing the sum of $\MRInv(T_v,\ranking)$ independently for each $v\in T$:
\begin{equation*}
   \MInv(T,\ranking) = \sum_v \MRInv(T_v,\ranking).
\end{equation*}
In order to compute $\MRInv(T_v,\ranking)$ for each $v$, we need to check the orderings $\ordering$ of the children $u_1, \dots, u_k$ of $v$ and determine the minimum number of cross inversions between the corresponding leaf sets $L_1,\dots,L_k$ in the order given by $\ordering$:
\begin{equation*}
    \RInv(T_v,\ranking,\ordering) \doteq \sum_{1 \le i < j \le k} \XInv_\ranking(L_{\ordering(i)},L_{\ordering(j)}).
\end{equation*}

A natural approach for computing $\MInv(T,\ranking)$ consists of two phases:
\begin{enumerate}
\item In a first phase, we compute $\XInv_\ranking$ between the leaf sets of any pair of siblings in $T$. This can be done in a bottom-up fashion similar to mergesort. More precisely, for a node $v$ with sorted leaf sets $L_1,\dots,L_k$, the cross inversions between every pair of leaf sets can be calculated using a merge operation in $O(k \cdot(\lvert L_1 \lvert+\dots+\lvert L_k \lvert))$ time, and sorts the concatenation of the leaf sets for use in future steps. Each leaf of depth $d$ appears in $d$ operations, which gives the total runtime of $O(\deg(T)\cdot \avgdepth(T)\cdot n)$. 
\item In a second phase, we check, for each node $v$ independently, which ordering $\ordering$ of the children of $v$ minimizes $\RInv(T_v,\ranking,\ordering)$. We then output the sum of the values $\MRInv(T_v,\ranking)$.
\end{enumerate}
Exhaustively testing all orderings $\ordering$ of the $k$ children of $v$ to compute $\MRInv(T_v,\ranking)$ takes $O(k \cdot k!)$ time, as $O(k)$ time is required to calculate $\RInv(T_v,\ranking,\ordering)$, for each $\ordering$, from the precomputed values of $\XInv$. This results in a total of $O(n \cdot \deg(T)!)$ time for the second phase, and an overall running time of $O((\deg(T)!+\deg(T)\cdot\avgdepth(T))\cdot n)$. 

For any constant bound on the degree of the tree $T$, the basic algorithm runs in time polynomial in $n$. The dependency of the running time on the degree can be improved. One way to do so is by reducing the problem of the second phase to the closely-related and well-studied problem of computing a minimum arc feedback set of a weighted directed graph.
\begin{definition}[Minimum Feedback Arc Set]
    Given a directed graph $G=(V,E)$ with an ordering $\ordering$ on the vertices $v_1,\dots,v_n$, a \emph{feedback arc} is an edge $e_k$ from $v_i$ to $v_j$ such that $\ordering(v_i)>\ordering(v_j)$. The minimum feedback arc set problem is finding the minimum number of feedback arcs induced by any ordering $\ordering$.
    In \emph{weighted} minimum feedback arc set, each edge from $v_i$ to $v_j$ has a weight $w_{ij}$, and the objective is to minimize $\sum_{e\in E} e_{ij} \cdot \Indicator[\ordering(v_i)>\ordering(v_j)]$.
\end{definition}

We can encode the problem of computing $\MRInv(T_v,\ranking)$ as an instance of weighted minimum arc feedback set, where each edge of the graph $G$ has a positive weight. If the leaf sets of the children of $v$ are $L_1,\dots,L_k$, we construct a graph $G$ with $k$ vertices $v_1,\dots,v_k$. For each pair of vertices $v_i$ and $v_j$, if $\XInv_\ranking(L_i,L_j)<\XInv_\ranking(L_j,L_i)$, we add an edge from $v_i$ to $v_j$ of weight $\XInv_\ranking(L_j,L_i)-\XInv_\ranking(L_i,L_j)$. We can extract the value of $\MRInv(T_v,\ranking)$ from the weight of the minimum feedback arc set.

As a consequence, we can use existing efficient algorithms for weighted minimum arc feedback set to construct algorithms for inversion minimization on trees that are more efficient than the basic algorithm we described. \cite{Bodlaender2011} gives two exact algorithms for weighted minimum arc feedback set. One algorithm {\cite[Algorithm 1]{Bodlaender2011}} is based on the Held-Karp algorithm for the traveling salesman problem \cite{HK62}; it uses dynamic programming to achieve a time complexity of $\Theta(n^22^n)$ and a space complexity of $\Theta(2^n)$ for a graph of $n$ vertices. Another algorithm {\cite[Algorithm 2]{Bodlaender2011}} uses a divide and conquer approach that achieves a time complexity of $O(\poly(n)\cdot 4^n)$, but has the advantage of only needing polynomial space.

An adaptation of the dynamic programming algorithm for calculating $\MRInv$ is given in \cref{alg:MRInv-DP}. Using this subroutine for computing $\MRInv$, we can improve the time complexity of our basic algorithm to $O((\deg(T)^2 2^{\deg(T)}+\deg(T)\cdot\avgdepth(T))\cdot n)$.

\begin{algorithm}[ht]
    \caption{$\MRInv(T_v, \ranking)$, Dynamic Programming}\label{alg:MRInv-DP}
    \begin{algorithmic}
        \Require Tree $T_v$ with child leaf sets $L_1,\dots,L_k$, ranking $\ranking$
        \State Initialize $\Cost[S]$, where $S$ is over all subsets of $\{1,\dots,k\}$.
        \State $\Cost[\varnothing]\gets 0$
        \For {$i$ from $1$ to $k$}
            \For {all sets $S$ of size $i$}
                \State $\Cost[S]\gets \min_{s\in S} (\Cost(S\setminus \{s\}) + \sum_{j\in S, j\neq s} \XInv_\ranking(L_{j}, L_{s})))$
            \EndFor
        \EndFor
        \State \Return $\Cost[\{1,\dots,k\}]$.
    \end{algorithmic}
\end{algorithm}

The improved running time is still not efficient for trees with unrestricted degree. This is to be expected, as there also exists a reduction from minimum feedback arc set to inversion minimization on trees, and the former is NP-hard \cite{Karp1972}.

\begin{figure}[ht]
    \centering
    \begin{tikzpicture}[]
        \node at (-8,0){$G$:};
        
        \node (Gu) at (-7,-0.2)[circle, draw, inner sep=1.5pt, label={below: $v_i$}]{};
        \node (Gv) at (-5,-0.2)[circle, draw, inner sep=1.5pt, label={below: $v_j$}]{};
        \draw[->] (Gu) -- node[pos=0.5, label={ $e_k$}]{} (Gv);
        
        \node at (-3,0){$T$:};
        \node at (-3,-3.9){$\ranking$:};
        
        \node (root) at (0,0)[circle, draw, inner sep=1.5pt]{};
        \node (u) at (-1.5,-1.5)[circle, draw, inner sep=1.5pt, label={above left: $v_i$}]{};
        \node (v) at (1.5,-1.5)[circle, draw, inner sep=1.5pt, label={above right: $v_j$}]{};
        \node (u1) at (-2.25,-3.5)[circle, draw, inner sep=1.5pt, label={below:\small $-2k$}]{};
        \node (u2) at (-0.75,-3.5)[circle, draw, inner sep=1.5pt, label={below:\small $2k-1$}]{};
        \node (v1) at (0.75,-3.5)[circle, draw, inner sep=1.5pt, label={below:\small $-2k+1$}]{};
        \node (v2) at (2.25,-3.5)[circle, draw, inner sep=1.5pt, label={below:\small $2k$}]{};
        
        \node at (0,-1.5){$\cdots$};
        \node at (-1.5,-3.5){$\cdots$};
        \node at (1.5,-3.5){$\cdots$};
        
        \draw (root)--(u)--(u1);
        \draw (u)--(u2);
        \draw (root)--(v)--(v1);
        \draw (v)--(v2);
        
        \draw (-0.6,-0.1) .. controls (0,-0.5) .. (0.6,-0.1);
        \node at (-0.8,-0.3){\small $n$};
        \draw (-1.8,-1.75) .. controls (-1.5,-1.95) .. (-1.2,-1.75);
        \node at (-2.5,-2){\scriptsize $2\deg_G(v_i)$};
        \draw (1.8,-1.75) .. controls (1.5,-1.95) .. (1.2,-1.75);
        \node at (0.5,-2){\scriptsize $2\deg_G(v_j)$};
    \end{tikzpicture}
    \caption{Encoding of an edge $e_k$}
    \label{fig:turing:MFAS-encoding}
\end{figure}
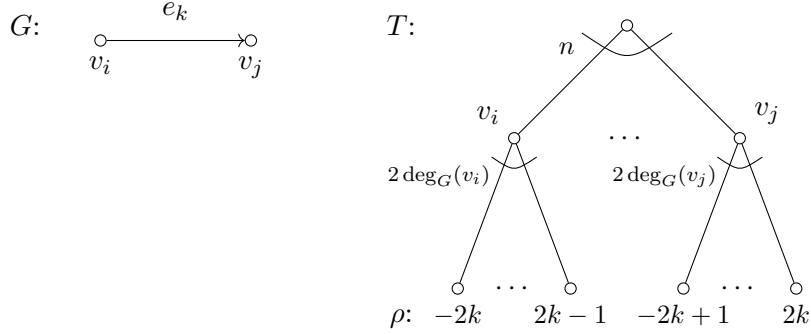

\begin{proposition}
    Computing $\MInv(T,\ranking)$ is NP-hard.
\end{proposition}

\begin{proof}
    For a graph $G$ with $n$ vertices $v_1,\dots,v_n$ and $m$ directed edges $e_1,\dots,e_m$, we construct a depth-2 tree $T$ and a ranking $\ranking$ of its leaves. We will assume that $G$ has no isolated vertices; this goes without loss of generality as isolated vertices can be dropped from an instance of minimum arc feedback set without affecting the answer. We also assume that between any two vertices at most one of the two directed edges is present; this is also without loss of generality since dropping the edges in case both are present reduces the answer by one. Finally, for ease of notation, we allow the ranking $\ranking$ to be an injective mapping into the integers; this can be changed easily by replacing each integer by its rank in the range.
        
    In the first layer, the root of $T$ has $n$ children corresponding to $v_1,\dots,v_n$. The second layer has leaves with ranks encoding the edges of $G$. For each edge $e_k$ going from $v_i$ to $v_j$, we add two leaves under $v_i$ with ranks $-2k$ and $2k-1$, and two leaves under $v_j$ with ranks $-2k+1$ and $2k$, as shown in \cref{fig:turing:MFAS-encoding}. All ranks are distinct.
    
    Consider the number of inversions in $T$ induced by an ordering $\ordering$ of $v_1,\dots,v_n$. For each edge $e_k$, the number of inversions between the leaves of rank $-2k$ and $2k-1$ and the leaves of rank $-2k+1$ and $2k$ is $1$ if $\sigma(v_i)<\sigma(v_j)$ and $3$ if $\sigma(v_i)>\sigma(v_j)$. These four leaves also form $8(m-1)$ inversions with all other leaves, keeping in mind that these inversions are counted twice when summed up over all edges.
    
    Therefore, the minimum number of inversions in $T$ is given by
    \begin{equation*}
        \MInv(T,\ranking)=\min_{\ordering} \left(4m(m-1)+m+ 2\cdot\sum_{e_k\in E}\Indicator[\sigma(v_i)>\sigma(v_j)]\right).
    \end{equation*}
    The size of the minimum arc feedback set is precisely $\min_{\ordering} \left( \sum_{e_k\in E}\Indicator[\sigma(v_i)>\sigma(v_j)] \right)$, which can be extracted from $\MInv(T,\ranking)$ with straightforward calculations. This completes the reduction.
\end{proof}

\paragraph{Approximation Algorithms.} If we relax our requirements to an approximate answer, we can approximate $\MRInv$ in polynomial time using existing approximation algorithms for weighted minimum feedback arc set. The best known such algorithm achieves an approximation ratio of $O(\log n\log\log n)$ on a graph with $n$ vertices \cite{ENSS98}. Adapting this algorithm for minimizing inversions in trees produces an approximation factor of $O(\log (\deg(T)) \log\log (\deg(T)))$ for $\MInv(T,\ranking)$. Under the unique games conjecture, there does not exist a constant-factor approximation algorithm for minimum feedback arc set on arbitrary digraphs \cite{GHMRC11}.

In the special case of tournament graphs, which have exactly one edge of weight 1 between every pair of vertices, there are efficient constant factor approximation algorithms for minimum arc feedback set. Some of these also apply to weighted tournaments, where for every pair of vertices $v_i,v_j$, the nonnegative edge weights $w_{ij},w_{ji}$ satisfy $w_{ij}+w_{ji}=1$. This case corresponds to the scenario of computing $\MRInv(T_v,\ranking)$ where all leaf sets $L_i$ of siblings have the same size. \cite{KS10} gives an algorithm with runtime $O^*(2^{O(\sqrt{\mathrm{OPT}})})$, given that the optimal answer is $\mathrm{OPT}$. \cite{KMS07} also gives an approximation algorithm in the case where $w_{ij}+w_{ji}\in [b,1]$ for some $b>0$: For any $\epsilon>0$, the algorithm produces a $(1+\epsilon)$-approximation of $\mathrm{OPT}$ in time $n 2^{\Tilde{O}(1/(\epsilon b)^{12}) }$. For the problem of computing $\MRInv(T_v,\ranking)$, the parameter $b$ represents the ratio between the smallest and largest possible values of $\lvert L_i\rvert\cdot\lvert L_j\rvert$.

\paragraph{Wilcoxon test.} As a final remark we point out an alternate way of computing $\Pi_T(\ranking)$ in the special case of the Mann--Whitney trees of Figure~\ref{fig:Mann-Whitney}. The number of cross inversions $\XInv_\ranking(A,B)$ can be written in terms of the rank sum $W_B \doteq \sum_{y \in B} \ranking(y)$ as follows, where $a \doteq |A|$ and $b \doteq |B|$:
\begin{equation}\label{eq:Wilcoxon}
\XInv_\ranking(A,B) = ab + \frac{b(b+1)}{2} - W_B.
\end{equation}
The quantity $W_B$ is known as the Wilcoxon rank-sum statistic for differences between random variables. Because of the relationship \eqref{eq:Wilcoxon} the Wilcoxon test is equivalent in power to the Mann--Whitney test. However, the evaluation based on the efficient computation of cross inversions (especially in the case of unbalanced set sizes $a$ and $b$) is superior to the evaluation based on the rank sum $S_B$, as the latter presumes sorting the combined set $X = A \sqcup B$.

\section*{Acknowledgements}

We would like to thank Greta Panova, Robin Pemantle, and Richard Stanley for pointers regarding Gaussian polynomials, Stasys Jukna for answering questions about the complexity measure $N$, and the anonymous reviewers for helpful suggestions. We appreciate the partial support for this research by the U.S.\ National Science Foundation under Grants No. 2137424 and 2312540. Any opinions, findings, and conclusions or recommendations expressed in this material are those of the authors and do not necessarily reflect the views of the National Science Foundation.

A preliminary version of this paper appears in the Proceedings of the 2023 ACM-SIAM Symposium on Discrete Algorithms \cite{HvMM23}.


\bibliography{refs}

\end{document}